\documentclass[9pt]{article}
\usepackage[utf8]{inputenc}
\usepackage{standalone}
\usepackage{float}
\pdfoutput=1
\usepackage[margin=1.1in,letterpaper]{geometry}
\usepackage{physics}
\usepackage{amsmath,amssymb,amsthm}
\usepackage{mathtools}
\usepackage{bm}
\usepackage{enumitem}
\usepackage{graphicx}
\usepackage{multirow}
\usepackage[numbers, square, sort&compress]{natbib}
\DeclareMathOperator\sign{sgn}
\usepackage{subcaption}
\usepackage[colorlinks=true, urlcolor=blue, citecolor=black]{hyperref}
\usepackage{authblk}
\usepackage[dvipsnames]{xcolor}
\usepackage{algorithm}
\usepackage{algpseudocode}

\newtheorem{theorem}{Theorem}[section]
\newtheorem{lemma}[theorem]{Lemma}
\newtheorem{corollary}[theorem]{Corollary}
\newtheorem{proposition}[theorem]{Proposition}

\theoremstyle{definition}
\newtheorem{definition}[theorem]{Definition}
\newtheorem{remark}[theorem]{Remark}

\numberwithin{equation}{section}

\newcommand{\add}[1]{\textcolor{black}{#1}}

\newcommand{\name}{
                    GEMINI
                    }
\newcommand{\Name}{
                    GEMINI
                    }

\title{\Name: Generalized Ensnarlment Measure from Incomplete-linkage of Network-network Interactions
        }
\author[1,2,3]{Yu Tian}%
\author[1,2,3]{Chinmayi Subramanya}
\author[1,2,3]{Carl D. Modes}
\affil[1]{\textit{Max-Planck Institute of Molecular Cell Biology and Genetics,
01307 Dresden, Germany}}
\affil[2]{\textit{Center for Systems Biology Dresden, 01307 Dresden, Germany}}
\affil[3]{\textit{Cluster of Excellence, Physics of Life, TU Dresden, 01307 Dresden, Germany}}

\date{}

\begin{document}

\maketitle

\begin{abstract}
    Spatially embedded networks are central to many physical and biological systems, where geometry and connectivity jointly shape structure and function. 
    Examples abound across the scales of biological organization, from network-like membrane-bound organelles in the cell to mesoscale tissue organization of multiple distinct flow networks in organs and beyond.
    In each of these cases, the complexity of the architectures has heretofore frustrated our ability to link mechanism or regulation of these structures to reduced modeling or even relevant characterization, putting structure-function relationships largely out of reach.
    Complex, functional spatial networks can be decomposed into tree-like and cyclic substructures, but we still lack both an understanding of how these elements intertwine to give rise to function, and the tools to holistically quantify both the topological and geometric aspects of these features in their full network context.
    To close this gap, we here introduce \textit{\Name}, a topology and geometry aware operator that directly characterizes incomplete linking and more general spatial associations between edges in spatially embedded network architectures.
    \Name contains information on edge-edge association through an incomplete version of the Gauss linking integral which simultaneously endows it with topological sensitivity when collections of edges form linked assemblages.
     Validation on both synthetic lattices and on mouse brain vasculature data demonstrates that \Name systematically captures and classifies the complexity of structural organizations. 
    Our results provide a general approach for analyzing spatial networks in realistic data, where topology and geometry together determine function, thus opening the door to a more complete understanding of structure-function relationships across a broad set of biological examples where complex network organization is key. 

\end{abstract}


\section{Introduction}
Spatially embedded networks arise ubiquitously across the natural and engineering sciences, where geometry and topology jointly constrain system-level organization and function \cite{barthelemy2022spatial}.  
Examples range from transportation and infrastructure systems to biological tissues, vascular networks, and neural architectures \cite{oconnor2022vascular,sporns2010brain,stiso2018spatial,seguin2023neuronet}. 
In such systems, spatial embedding is not a secondary attribute but a defining feature: distances, orientations, and other geometric constraints shape how interactions form, propagate, are transduced, and further adapt over time \cite{stiso2018spatial}. 
Understanding how spatial structure influences network organization is, therefore, central to uncovering the principles governing these systems' behavior. 

Since the late 1990s, research in complex networks has largely re-framed networks as abstract graphs, prioritizing topological descriptors over explicit geometric embeddings \cite{newman2018networks,barabasi2016networksci}.
Seminal models such as the Erd\"{o}s-R\'{e}nyi random graphs formalized connectivity as independent probabilistic events on a node set, with no metric notion of distance \cite{erdos1959random,Bollobas_2001}.
The subsequent small-world model and preferential attachment model reinforced the idea that short path lengths and heavy-tailed degrees can be understood as broadly universal features, without reference to physical space, in most cases \cite{watts1998model,barabasia1999model}.
These efforts enabled analytically solvable null models and scalable statistics, and matched the at-the-time absent or incomplete spatial metadata \cite{newman2003review}.
However, in genuinely spatial systems, abstracting away geometry can erase, for example, planarity constraints, wiring costs, and geometric bottlenecks, and purely topological analysis may be systematically biased \cite{barthelemy2022spatial}.
These tensions motivated a more recent shift towards explicitly spatial models and methods, where geometry is treated as a fundamental, co-determining variable rather than optional metadata.

Much of the existing research on spatially constrained networks is predominantly node-centric. 
In network science, spatial information is typically incorporated through node locations, node densities, or distance-based null and generative models, whereas edges are still treated primarily as abstract adjacency relations \cite{barthelemy2022spatial,penrose2003rgg,expert2011commuspatial,sarzynska2016null}. 
This perspective overlooks the fact that, in many real systems, edges themselves are physical entities as well, such as roads, vessels, fibers, or conduits, whose lengths, orientations, and mutual arrangements encode essential structural and functional information \cite{dehmamy2018physnet,oconnor2022vascular}.
Recent advances in data science and machine learning have likewise incorporated geometric information into their pipelines, such as geometry-aware graph Laplacians and learning architectures that integrate spatial or symmetry priors \cite{coifman2006diffusion,bronstein2017geometric}.
While powerful, these approaches still predominantly operate on node-level representations, leaving under-characterized key aspects of spatial organization associated to the architecture of the edges, where many physical and biological processes actually take place. 

Among previous efforts to capture edge-level organization in spatial networks, much attention has been directed toward higher-order structures, such as paths, cycles, cliques, and related topological descriptors \cite{katifori2010loop,corson2010ot,modes2016hierarchies,liu2021isotopy,pascucci2011topological,glover2024acn,shen2024kda}. 
Cycles, in particular, have been shown to play an important role in transport efficiency, robustness, and redundancy, particularly in the presence of damage and fluctuating loads, whereas tree-like structures often emerge as efficient solutions under stationary transport demands and developmental constraints \cite{katifori2010loop,corson2010ot}. 
More generally, the balance between cycles, branching, and disorder has been linked to functional trade-offs in both physical and biological systems \cite{dahl-jensen2018pancreas}. 

\begin{figure}[t]
    \centering
    \includegraphics[width=\textwidth]{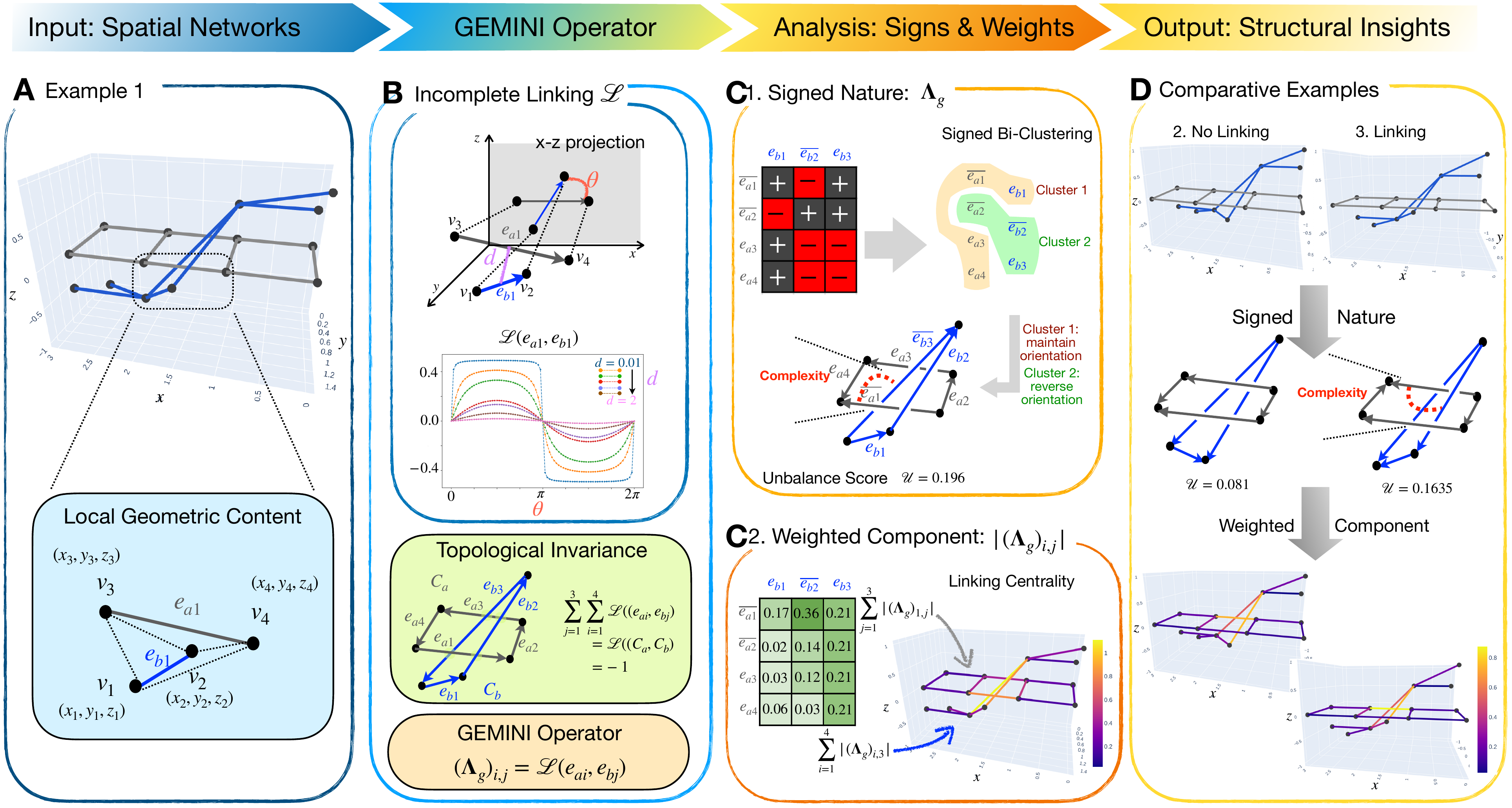}
    \caption{Illustration of the GEMINI operator.}
    \label{fig:schematic}
\end{figure}

Beyond the structure of a single network, complexity may also arise from interactions between two spatial networks sharing the same embedding space. A natural question is whether cycles from the two networks mutually constrain one another in the space, particularly when crossings are costly or forbidden. 
In this spirit, Kramer and Modes introduced \textit{ensnarlment} as a mesoscopic notion of inter-network intertwining: two spatially embedded network components are ``ensnarled'' if they cannot be separated without removing a subset of edges \cite{kramer2023ensnarl}. 
They quantified this notion through the linking number between cycles derived from the graphs' cycle bases, computed via the Gauss linking integral \cite{ricca2011gauss}. In this way, ensnarlment provides a principled description of complexity emerging between two spatial networks, revealing inter-network spatial constraints that are invisible to representations based solely on abstract adjacency. 

However, approaches based on explicit cycle detection face three fundamental challenges. 
First, identifying and enumerating cycle bases is computationally expensive and scales poorly with network size and spatial complexity. Second, for empirical systems, real-world spatial networks are often incompletely observed: missing, uncertain, or occluded links are common due to measurement limitations, noise, or biological variability. In such settings, cycle-based characterizations can become unstable or ill-defined. Third, and most critically, many real-world networks may simply not contain cycles. For example, as the pancreas develops, it builds and then transitions through cyclic structures, later simplifying into a tree topology \cite{dahl-jensen2018pancreas}.
Together, these challenges reveal a broader gap in spatial network analysis: the lack of tools that capture the interwoven spatial organization without relying on cycles or fully observed connectivity.

In this paper, we address this gap by introducing \textit{\Name}, an edge-centric, geometry-aware operator that directly quantifies incomplete linking between edges in spatial networks (Fig.~\ref{fig:schematic}) \footnote{\Name can be found at \url{https://github.com/yutian-research/EnsnarlNet}.}.
By employing an incomplete open-curve form of the Gauss linking integral, this operator bypasses explicit cycle enumeration and remains well-defined even when connectivity information is partial. 
A distinctive feature of the operator is that it naturally gives rise to both positive and negative interactions, representing the relative orientation of edges in space.
We show that, despite local orientation dependence, the operator's spectral properties constitute intrinsic characteristics of the spatial network and are invariant under changes in individual edge orientation. 

Building upon this perspective, we define an associated unbalance score that provides a principled measure of edge-level complexity in ensnarled spatial networks. 
We further develop an algorithm to detect patterns of intrinsic negative linking and identify edges and regions that play distinct functional roles. 
Furthermore, by exploring the associated \textit{linking space}, we propose a linking centrality that characterizes the importance of individual edges to the overall geometric and topological architecture of the system.
We validate our framework on both synthetic regular lattices and real spatial networks, including mouse brain vasculature, where it captures regional differences in complexity.
In particular, \Name reveals that in the cerebellum and medulla, flow alignment is prioritized by the vascular architecture, while in the bridging region including the midbrain, delivery between specific points is emphasized.
Overall, our results establish a general and scalable approach for edge-centric characterization of spatial networks, opening the door to a more complete understanding of spatial architecture in 3D network systems. 

\section{The \Name operator for spatial networks}
\label{sec:Gauss_linking_integral}
In order to build towards the \Name operator, we first relax the classical Gauss linking integral to characterize not only closed curves, but also open ones, the \textit{incomplete Gauss linking integral}. We show its invariance properties subject to isometry transformations, and discuss a particularly interesting special case when the underlying open curves are line segments. 

The geometric linking between closed curves in three-dimensional space can be quantified by the \textit{Gauss linking integral}, a classical topological invariant measuring the number of times the curves wind around each other \cite{ricca2011gauss}. 
Originally introduced in electromagnetism and later developed in knot theory, the linking number has found applications ranging from magnetic helicity in fluid dynamics \cite{Moffatt1969fluid} to the topology of polymers and DNA \cite{klenin2000glicomp}. Recently, it has also been applied to characterize spatial organization of networks through quantifying the linking between cycles \cite{kramer2023ensnarl}. However, such approaches implicitly rely on well-defined cycles and complete connectivity information, assumptions that can be violated in empirical networks. 
These motivate the development of a more general framework capable of capturing the incomplete linking within and between spatial networks, through the \textit{incomplete Gauss linking integral}: 

\begin{equation}
 \mathfrak{L}(\gamma_1, \gamma_2) 
        = \frac{1}{4\pi}\int_{0}^1\int_{0}^1({\gamma'}_1(s)\times {\gamma'}_2(t)) \cdot \frac{(\gamma_1(s) - \gamma_2(t))}{\norm{\gamma_1(s) - \gamma_2(t)}^3} ds \, dt
        \label{eq:incomplete gli}
\end{equation}
where $l_1, l_2$ are disjoint (open or closed) curves parametrized as $\gamma_1(s), \gamma_2(t)$, respectively, where $\gamma_1, \gamma_2: [0, 1]\to \mathbb{R}^3$. Note that the incomplete Gauss linking integral is still a symmetric, bilinear mapping, and further that composition of multiple curve segments can be done by linear superposition. The classical Gauss linking integral may therefore be recovered by composing multiple open curve segments into a pair of closed curves that are (or are not) linked with each other. See Appendix \ref{app:incomplete_GLI} for further details of the basic properties such as bilinearity and superposition and Appendix \ref{app:incomplete_GLI-invariance} for details of the useful isometries and invariance properties exhibited by the incomplete Gauss linking integral.

\label{sec:operator_analysis}

We now introduce the concept of \textit{generalized ensnarlment}, which is not restricted to the ensnarled state but a global feature that summarizes the ``intertwining'' between all pairs of edges in the two spatial networks.
We quantify the intertwining between edges through the incomplete linking, measured by the incomplete Gauss linking integral, and propose the \textit{GEMINI operator}: Generalized Ensnarlment Measure from Incomplete-linkage of Network-network Interactions. 
Note that our framework can also be applied to one spatial network instead of two, but for illustrative purposes, we discuss exclusively the case of two spatial networks in this section.

Specifically, for two spatial networks with the corresponding edge sets $E_1, E_2$, the \textit{\name operator}, denoted by $\mathbf{\Lambda}_g\in \mathbb{R}^{m_1\times m_2}$, is defined as
\begin{align}
    (\mathbf{\Lambda}_g)_{i,j} = \mathfrak{L}(e_i, e_j),
    \label{eq:gemini operator}
\end{align}
where $m_1 = \abs{E_1}$, $m_2 = \abs{E_2}$ are the sizes of the two edge sets, the function $\mathfrak{L}$ is the incomplete linking as defined in \eqref{eq:incomplete gli}, and for each edge, we fix one choice of direction, where we consider each edge $e=(i,j)$ as a line segment from node $i$ to $j$ ($i<j$). 
We will show later that, although the individual entries depend on the arbitrary choice of edge orientation, changing the orientation of one edge reverses the sign of its incomplete linking with all edges in the other network, the operator possesses well-defined global structures, with key spectral properties remaining invariant under such orientation changes. In principle, one should consider both possible choices of direction for each edge, but we argue that the \name operator we consider here already contains all relevant information (see Appendix \ref{app:gemini-property}).

It is further interesting to note that the \name operator derives a number of useful features from underlying bipartite structures. Specifically, there exists an explicit bipartite representation of the \name operator itself, as well as further bipartite structure emerging from the signed nature of the edge weights. The \name operator $\Lambda_g$ admits a direct graph theoretic interpretation. Given two spatial networks with edge sets $E_1$ and $E_2$, the \name operator assigns a real value to every pair of edges $(e_i, e_j) \in E_1 \times E_2$, giving rise to a natural weighted bipartite structure. Define a bipartite graph $G_b = (V_b, E_b)$ with bipartition $V_b = E_1 \cup E_2$, such that each node in the first partition corresponds to an edge $e_i \in E_1$ and each node in the second partition corresponds to an edge $e_j \in E_2$. For every pair $(e_i, e_j)$, a bipartite edge is introduced with weight $w_{ij} = (\Lambda_g)_{i,j}$ unless $(\Lambda_g)_{i,j} = 0$ in which case the edge is excluded.
Under this construction, the \name operator $\Lambda_g \in \mathbb{R}^{m_1 \times m_2}$ is precisely the weighted biadjacency matrix of the bipartite graph as defined in \eqref{eq:gemini operator}.
Equivalently, the weighted adjacency matrix of the full bipartite graph $G_b$ is 
\begin{align}
    \mathbf{W}_b = \begin{pmatrix}
        \mathbf{0} & \mathbf{\Lambda}_g \\ 
        \mathbf{\Lambda}_g^T & \mathbf{0}
    \end{pmatrix}.
\end{align}
This formulation explicitly shows that the \name operator encodes a weighted bipartite graph whose nodes are spatial edges and whose edge weights measure pairwise geometric ensnarlment. 

One prominent feature of the \name operator is that it can take both positive and negative values, and a negative value can become positive by flipping the orientation of one of the corresponding edges in the original spatial networks. We refer to these edges as ``spatial edges'' hereafter. One may thus exploit the signed nature of the bipartite graph underlying the \name operator. 
Specifically, we will show that even though the individual edge weight in the bipartite signed graph has an arbitrary sign, the eigenvalues of the weight matrix $\mathbf{W}_b$, and thus the \name operator itself, remain invariant. 

To achieve this goal, consider \textit{switching} and \textit{switching equivalence}, important concepts in classifying signed graphs. \textit{Switching} is the operation of reversing the signs of all edges connecting a subset $S\subseteq V$ and its complement $V\backslash S$, \add{where $V$ is the node set}. Two signed configurations of the edges are \textit{switching equivalent} if one configuration can be reached from the other by a switching operation. The corresponding equivalence classes are called switching classes.  An important feature of switching equivalence is that the spectra of signed graphs remains invariant in the same switching equivalence class \cite{Atay_signedCheeger_2020,zaslavsky1982signed}.

Among the switching equivalence classes, the one corresponding to \textit{structural balance} is of particular interest. Introduced in the 1940s \cite{heider_1946_psychology} and primarily motivated by social and economic networks, structural balance is a fundamental notion in the study of signed graphs \cite{harary_1953_balance,cartwrightharary_1956_gbalance}. 
A signed graph is structurally balanced if all cycles in the graph have an even number of negative edges, rendering the cycles ``positive''. Structurally balanced signed graphs exhibit desirable properties and have stimulated new methods in various contexts, including social networks \cite{Kunegis_2009_Zoo,Symeonidis_2013_multiway,Wu_2012_eigenspace}, economic systems \cite{tian_2021_role,Bartesaghi2025finance}, and biological networks \cite{Symeonidis_2013_biology}.
In particular, a signed graph is structurally balanced if and only if it is switching equivalent to the signed graph with all edge signs positive. 

In the context of the \name operator, the structural balance of the bipartite signed graph indicates that we can find a global orientation of the spatial edges such that the incomplete linking between each pair of them is nonnegative, i.e., their relative orientations are arranged in an aligned manner with how they (potentially) cross each other.
Therefore, structural balance imposes strong conditions on the relative (geometric and topological) alignments between the edges in the spatial networks: whether the bipartite signed graph associated with the \name operator is structurally balanced, and if not, how far the spatial edges can be oriented to approach structural balance, provide a natural measure to quantify the complexity of the underlying system. 
For instance, the \name operator can be interpreted within the framework of spin systems by associating each spatial edge with a spin variable and treating the operator entries as coupling strengths between spins. In this analogy, structural balance corresponds to the existence of a spin configuration in which all interactions are satisfied, resulting in a frustration-free state. The degree of frustration, commonly used to quantify \textit{unbalance} in signed networks, i.e., the opposite to structural balance, therefore provides a natural measure of the complexity of the corresponding spin system.
In our framework, we quantify the level of complexity, or frustration, induced by the \name operator of the spatial networks through the following \textit{unbalance score}: 
\begin{align}
    \lambda_{min} = \min\{\lambda: \mathbf{L}_{rw}\mathbf{x} = \lambda\mathbf{x},\, \mathbf{x}\ne \mathbf{0}\},
\end{align}
i.e., the smallest eigenvalue $\lambda_{min}$ of the ``signed random-walk Laplacian'' $\mathbf{L}_{rw} = \mathbf{I} -  \mathbf{D}_b^{-1}\mathbf{W}_b$, where $\mathbf{D}_b = \mathbf{Diag}(\mathbf{d})$, $\mathbf{d} = (d_i)$, and the degree $d_i = \sum_{j}\abs{(\mathbf{W}_b)_{ij}}$ (cf. \cite{tian2024sign}). 
Specifically, $\lambda_{\min} = 0$ if and only if the signed graph is structurally balanced. 
Note that there are many other walk-based measures proposed to measure the level of balance of the signed graph, 
and we choose the eigenvalue-based one inspired by invariance under edge orientation as shown in Theorem \ref{the:switching-eigvals}. 

If the associated bipartite signed graph is not structurally balanced, then a natural follow-up question is whether we can \textit{find a global orientation of the spatial edges that minimizes the negative edges in the associated bipartite signed graph}. 
This question is equivalent to the \textit{signed bi-clustering} problem, i.e., to find a bipartition of the signed graph such that the number of positive edges outside and the number of negative edges inside are minimized, since the negative edges outside will become positive if we perform a switching operation on one part \add{of the bipartition}.
The exact solution to this problem is equivalent to solving the frustration index or line index of balance, i.e., the least number of edges one needs to flip sign for the signed graph to become structurally balanced \cite{abelson1958frustration,harary1959frustration}. The latter is NP-hard and equivalent to the ground state calculation of the Sherrington-Kirkpatrick spin glass model \cite{barahona1982computational,sherrington1975solvable}. 
However, signed spectral clustering can provide a reasonably good candidate, especially when the corresponding signed graph is not far from being structurally balanced \cite{Atay_signedCheeger_2020, Kunegis_signspect_2010} (error bound from the signed Cheeger inequality in terms of the smallest eigenvalue of the signed Laplacian). 
We give the exact algorithm we use in our framework in Appendix \ref{app:signed_analysis}.

Finally, the \name operator allows us to construct a vector space over the embedded network architectures, which we refer to as the \textit{linking space}. Each row of the \name operator induces a feature vector for each edge in $E_1$, where each dimension encodes the incomplete linking with the corresponding edge in $E_2$. We can consider these feature vectors as embeddings in the corresponding space. Similarly, each column induces a feature vector for each edge in $E_2$, where each dimension characterizes the incomplete linking with the corresponding edge in $E_1$. The linking space is then where these feature vectors lie.

\add{To eliminate the effect of orientation, we consider edges in both directions. Specifically, we introduce the operator that contains the incomplete linking between each pair of edges in both directions, the full operator $\mathbf{\Lambda}_f \in \mathbb{R}^{2m_1\times 2m_2}$, where the first $m_1$ rows correspond to all edges in $E_1$ along the direction in the rows of $\mathbf{\Lambda}_g$ and the following $m_1$ rows correspond to all edges in the same order but along the opposite direction, and the columns are aligned in the same way but for edges in $E_2$.}
The feature vector for an edge $e_i\in E_1$ is then the $i$-th row vector of the full operator, $(\mathbf{\Lambda}_f)_{i:} = ((\mathbf{\Lambda}_f)_{i1}, (\mathbf{\Lambda}_f)_{i2}, \cdots)$, and the feature vector for an edge $e_j\in E_2$ is then the $j$-th column vector of the full operator $(\mathbf{\Lambda}_f)_{:j} = ((\mathbf{\Lambda}_f)_{1j}, (\mathbf{\Lambda}_f)_{2j}, \cdots)$. 

From the feature vectors, we propose a novel centrality measure of the spatial edges, characterizing their importance from the perspective of incomplete linking. Specifically, we define the centrality of an edge $e_i\in E_1$ as the sum of the absolute value (i.e. the 1-norm) of its feature vector, subject to a normalization:
\begin{align}
    c(e_i) = \frac{1}{2}\sum_{j}\abs{(\mathbf{\Lambda}_f)_{ij}} = \sum_{j}\abs{(\mathbf{\Lambda}_g)_{ij}} = \sum_{e_j\in E_2}\abs{\mathfrak{L}(e_i, e_j)}, 
\end{align}
and for an edge $e_j\in E_2$, similarly, we define the centrality as
\begin{align}
    c(e_j) = \frac{1}{2}\sum_{i}\abs{(\mathbf{\Lambda}_f)_{ij}} = \sum_{i}\abs{(\mathbf{\Lambda}_g)_{ij}} = \sum_{e_i\in E_1}\abs{\mathfrak{L}(e_i, e_j)}, 
\end{align}
Therefore, any contribution to the linking, regardless of being positive or negative, is taken into account in the centrality measure. 
Furthermore, to compare the linking centrality of edges in one set versus the other, we assume that the linking between edges inside the same set is negligible. Summing over the linking between all edges in both sets can then be well approximated by the linking between the edges in the other set. 

The linking centrality of an edge is larger if it has a higher absolute value of incomplete linking with other edges. In many applications, it may be important to identify the most ensnarled part, or to determine how to disentangle the structure by removing edges. With the linking centrality, we can propose simple, straightforward heuristics to solve these problems, e.g.: identify edges with the highest linking centrality for the former, and then remove edges according to their linking centrality for the latter. 

\section{\Name on simple examples}

\subsection{Ensnarled symmetric step lattices}
We begin our demonstration of the \Name operator by examining a simple example of regular lattices that are ensnarled with each other. Specifically, we consider the ``step" graphs, and construct the \name operator accordingly; see Fig.~\ref{fig:eg-step-p1}. Specifically, these lattices have the following unit: (i) one ladder lattice of two squares and folded in a way such that one square lies in the plane perpendicular to the other, and (ii) one copy of such a ladder lattice but rotated along the shared edge of the two squares by angle $\pi/2$ and translated along the direction such that the two edges perpendicular to the shared edges going through the centers of the two squares of the other ladder lattice; see Fig.~\ref{fig:eg-step-p1}A. We refer to lattices with such a unit as ensnarled step lattices, and refer to edges in one ladder lattice as one set of edges, $E_1$, and edges in the other ladder lattice as the other set of edges $E_2$. 

\begin{figure}[htbp]
    \centering
    \includegraphics[width=.98\textwidth]{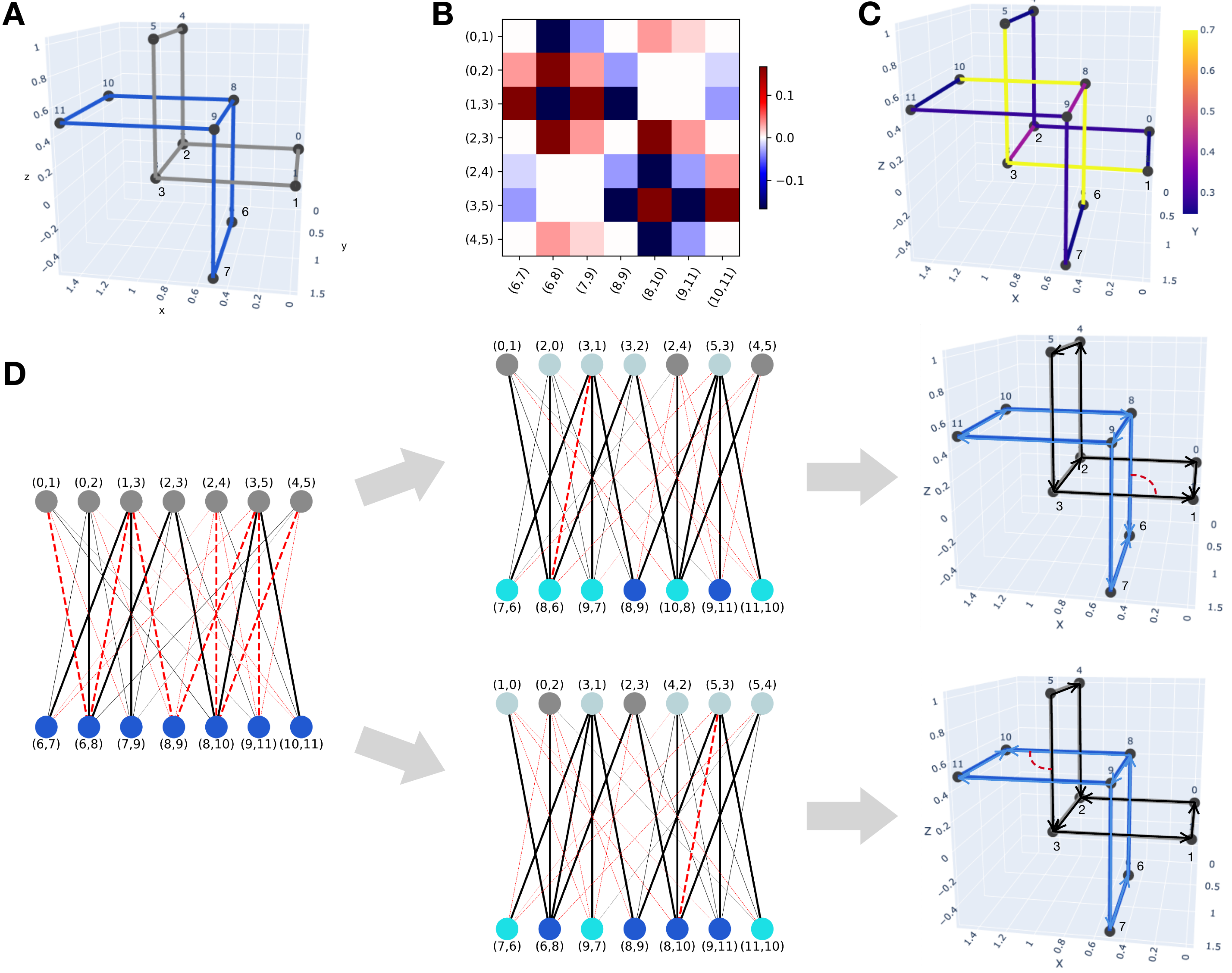}
    \caption{Example of the ensnarled step lattices with period 1. \textbf{A}. Visualization of the spatial network. \textbf{B}. \name operator from edges in gray to edges in blue. \textbf{C}. Linking centrality of edges in A. \textbf{D}. Bipartite signed graphs: each spatial edge is oriented (left) from the node with a smaller label to the node with a larger label, and then (middle) according to the optimal orientations that minimize the negative incomplete linking, with spatial edges whose orientations have been reversed marked in a lighter color, where the top and bottom figures correspond to two equivalent solutions. (right) The corresponding orientation of edges in \textbf{A} is shown on the right, where a red dashed line connects the edge pair with the most negative linking. }
    \label{fig:eg-step-p1}
\end{figure}   
For the repeat units of the ensnarled symmetric step lattices, there are two prominent symmetries: the first one is inside each edge set, and the second one is between the two edge sets. On the one hand, the unit is symmetric with respect to the plane where the shared edges of the two squares in the two sets of edges lie: for each specific edge $e$, its corresponding symmetric edge is in the same edge set as $e$'s. We note that each individual edge set is symmetric with respect to this plane. On the other hand, the unit is also symmetric with respect to the midpoint between the midpoints of the two shared edges. With the latter, for each specific edge $e$, we can find another edge $e'$ in the other edge set such that $e,e'$ are symmetric with respect to the point. We note that the two edge sets have the same size. These results can be extended to ensnarled step lattices of arbitrary sizes.  
 
For the bipartite signed graph derived from the \name operator of the ensnarled step lattice, the level of unbalance can be characterized by $\lambda_{\min} = 0.288$;  see Fig.~\ref{fig:eg-step-p1}D. 
Furthermore, corresponding to the prominent eigen-property of the \name operator, for the signed graph, its characteristic matrices, including the signed weight matrix, the signed Laplacian, and the signed normalized Laplacian, also have the property that \textit{each eigenvalue has multiplicity $2$}. The proof follows in a similar way to the \name operator, where we need to rearrange the edges in $E_1$ and $E_2$, but simultaneously, to construct the other eigenvector associated with the same eigenvalues.

The fact that there are two linearly independent eigenvectors corresponding to the smallest eigenvalue of the signed (normalized) Laplacian could cause a problem for the vanilla signed spectral bi-clustering algorithm because any linear combination of the two vectors provides an equivalently optimal solution to the relaxed signed bi-clustering problem. Therefore, we introduce an extra step to the vanilla signed spectral bi-clustering algorithm: to evaluate the objective function with different combinations of the eigenvectors and select the one(s) that returns the optimal value; see Appendix \ref{app:signed_analysis}.
With the updated signed bi-clustering, the algorithm returns two optimal solutions; see Fig.~\ref{fig:eg-step-p1}D, where the edges shown are after the corresponding switching operations. 

\begin{figure}[htbp]
    \centering
    \includegraphics[width=\textwidth]{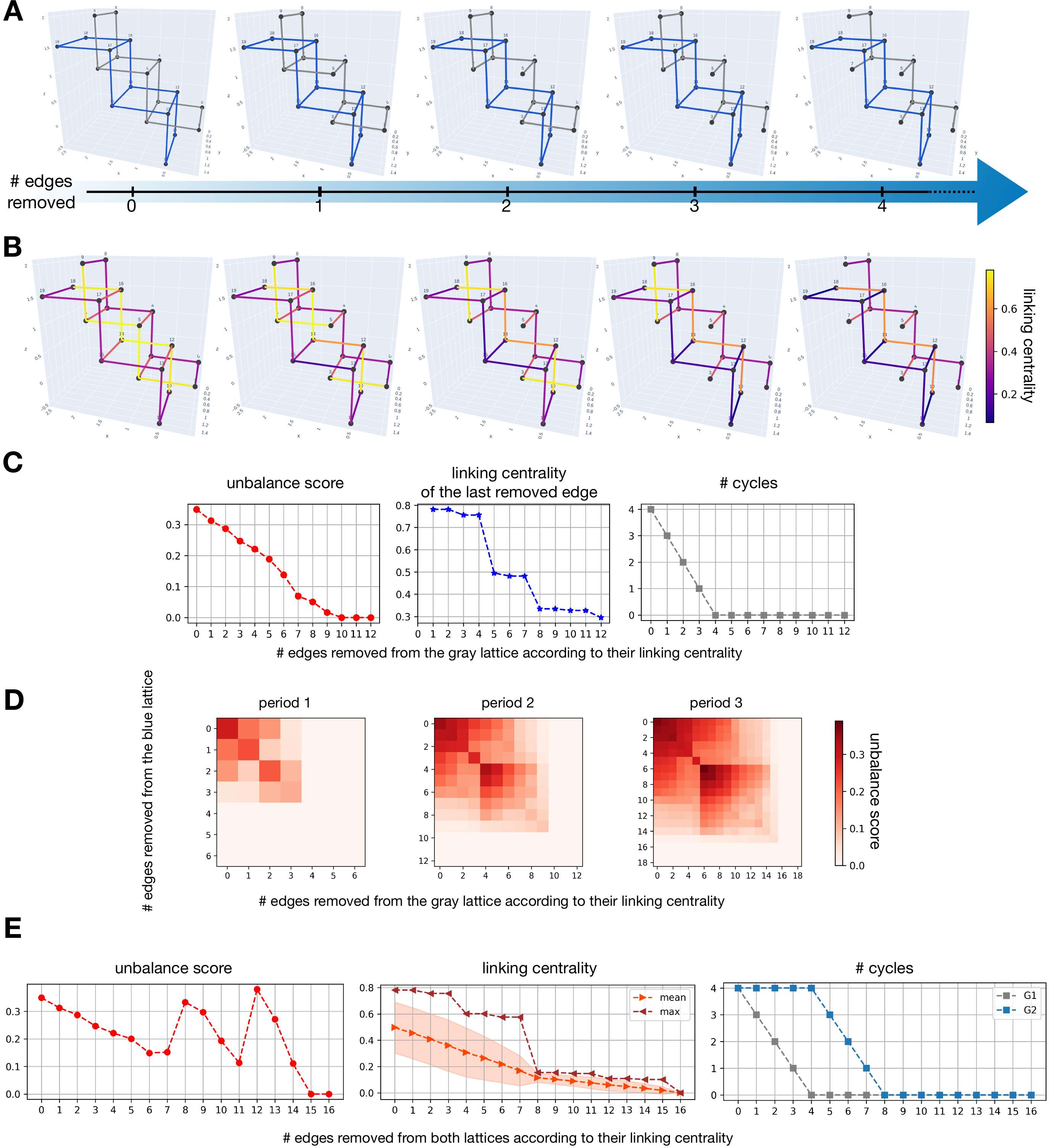}
    \caption{Example of removing edges from the ensnarled step lattices. \textbf{A}. Visualization of the remaining spatial network after removing the edges from one of the two ensnarled step lattices of period 2 in order. \textbf{B}. Results from the removal process in A: unbalance score (left), centrality of the recently removed edges (middle), and the size of the minimum cycle basis (right), where the x-axis is the number of removed edges. \textbf{C}. Unbalance score from removing edges from both of the two ensnarled step lattices (x-, y-axes) of the highest linking centrality (in lattices of periods 1 to 3, from left to right). \textbf{D}. Linking centrality of edges from the removal process in A. \textbf{E}. Unbalance score (left), mean and max linking centrality of edges (middle), and the sizes of the minimum cycle basis in the two lattices (right), where the x-axis is the number of removed edges. }
    \label{fig:eg-remv}
\end{figure}

The existence of these two optimal solutions are to be expected from the underlying symmetries of the example, and indeed, these symmetries are prominent in the two optimal solutions. 
Specifically, after switching one part in an optimal solution, there is either a strong intrinsic negative edge between $(3,1)$ and $(8,6)$ or between $(5,3)$ and $(8,10)$, which are symmetric to each other according to either of the two symmetries. 
These edges in the ensnarled step lattices are particularly important because they are the ones that actually go through the corresponding cycle in the other set of the network, and hence where the complexity of these two sets of edges originates.
The importance of these edges is also emphasized in the linking centrality; see Fig.~\ref{fig:eg-step-p1}C. Therefore, the signed graph perspective, together with the newly proposed linking centrality, can provide important information characterizing the complexity of the two sets of edges.

\subsection{Disruption of the symmetric ensnarled states}

In the previous subsection, there are well-defined cycles in both lattices; while in many real-world scenarios, cycles may not be present in one or both of the networks. 
With our proposed \name operator, we can still characterize the generalized ensnarlment between such networks. To demonstrate this, we will remove edges from ensnarled lattices to break the cyclic structure in either of the networks, and then both of them. For the former, the \name operator characterizes the generalized ensnarlment between cyclic and tree-like structures, while for the latter, the \name operator describes the features between two tree-like structures. 

With extra symmetry, we observe the same behavior from either of the ensnarled lattices as removing the edges from the other; see Fig.~\ref{fig:eg-remv}. Specifically, after removing $4$ edges from one lattice, it becomes a tree, but there is still more than half of the original unbalance between the remaining tree structure and the other lattice, emphasizing the nontrivial incomplete linking between the two and the complexity in the system. We also note that the first four edges removed have similar centrality values, and removing them leads to a similar decrease in the unbalance score. 

When removing edges from both lattices simultaneously, we observe even more interesting patterns: after we break all cycles in both lattices, a new regime starts with high unbalance, which again reduces as we further remove edges; see Fig.~\ref{fig:eg-remv}. 
Specifically, the second regime starts when we remove (i) $2$ edges from both lattices when period $=1$, (ii) $4$ edges from both lattices when period $=2$, and (iii) $6$ edges from both lattices when period $=3$. 
At this point, we reach a new stage where the two intrinsic symmetries within the step lattices reoccur, the plane symmetry and the point symmetry, which contribute to the overall complexity of the system. 
This pattern is consistent as we vary the number of periods in constructing the lattice. 

\section{\Name analysis of brain vasculature data}
We now demonstrate \Name's functionality and capabilities by turning to a complex, real data set, specifically the mouse whole brain vessel graph (VesselGraph) \cite{paetzold_whole_2021}. We first introduce the null models considered for our analysis, and then illustrate the results from the real data. 

\subsection{Mouse brain vasculature}
The VesselGraph dataset comprises high-resolution spatial graphs representing the complete cerebral vasculature of the mouse brain. These graphs are derived from recently advanced whole-brain imaging and segmentation protocols that enable visualization and reconstruction of the entire microvascular network at mesoscopic scales. Vascular organization, particularly regarding vessel sizes and the number of capillary links, varies substantially between brain regions. Moreover, the organization and changes to the structure are early signs for the development of specific diseases, such as Alzheimer's disease \cite{farkas2000pathological,bennett2018vessel}. In particular, we use the dataset corresponding to ``CD1-E-1'', of a total of $3,645,963$ nodes (junctions) and $5,791,309$ edges (vascular tubes); see Fig.~\ref{fig:real-all} for a visualization. 

\begin{figure}[htbp]
    \centering
    \includegraphics[width=.9\textwidth]{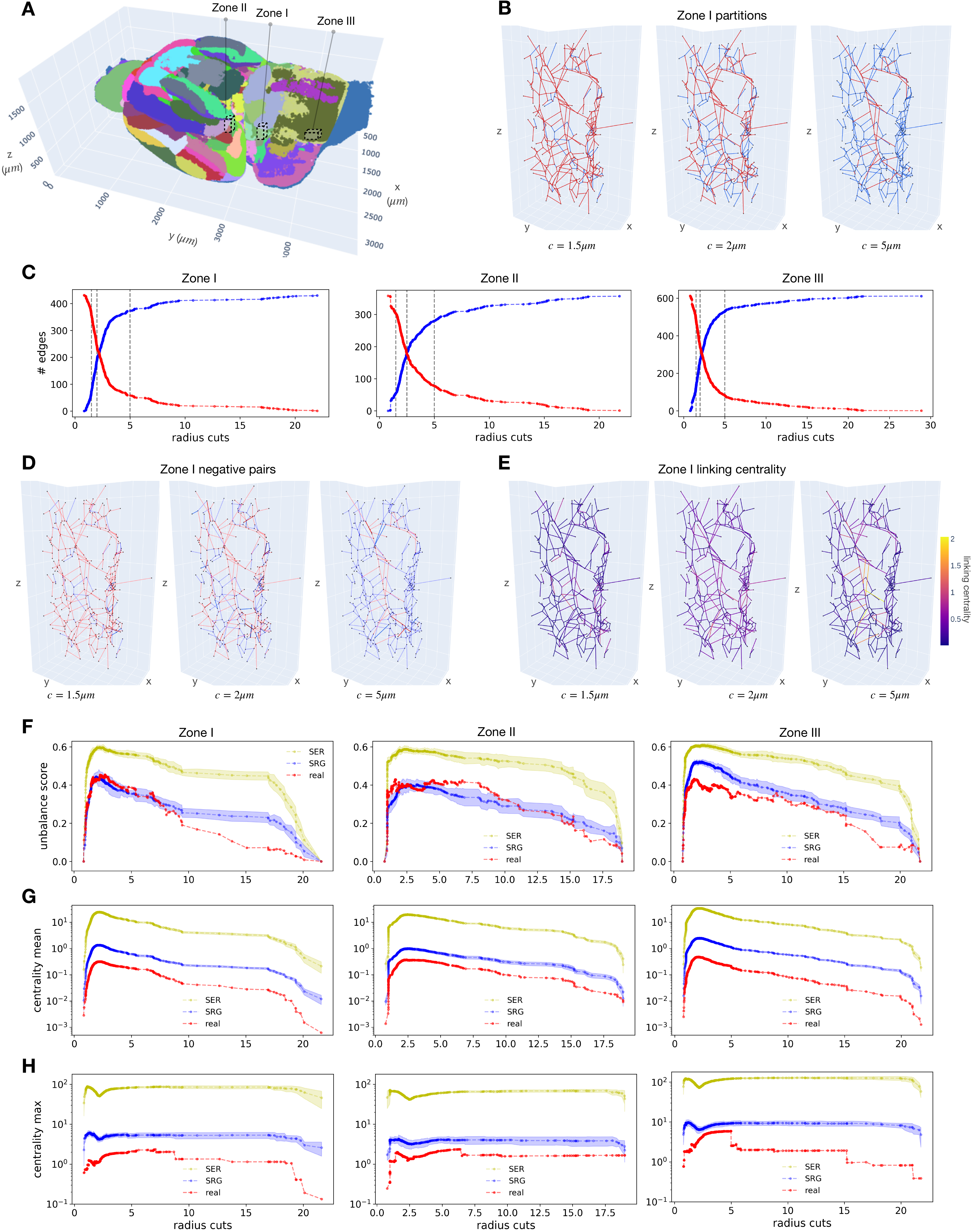}
    \caption{Mouse brain vasculature visualization and results. \textbf{A}. Junctions of the mouse whole brain vessel graph analyzed in this section, i.e., nodes of the spatial network, where color indicates the Allen mouse brain atlas regions. \textbf{B}. Spatial networks of the brain vessel in Zone I, where the edges are separated by three different values, $c$, for the threshold on the radius (1.5, 2, and 5$\mu m$), with the edges above the threshold in red and those below in blue. The displayed box spans $x\in [1030, 1140]$, $y\in [3140, 3300]$ and $z\in [820,1120]$. \textbf{C}. Counts of the edges with radii above (red) and below (blue) all thresholds for Zones I, II, and III. 
    \textbf{D}. Pairs of edges in Zone I with negative linking after an optimal orientation assigned (dark red/blue if incomplete linking $< -0.1$, slightly lighter red/blue if $<-0.05$, and the lightest otherwise). \textbf{E}. Linking centrality of edges in Zone I. \textbf{F}-\textbf{H}. Change of unbalance scores (\textbf{F}), mean (\textbf{G}) and max (\textbf{H}) linking centrality of all edges with respect to all thresholds on the radius for Zones I, II, and III (mean: dashed line; std: shade of all samples).}
    \label{fig:real-all}
\end{figure}

Blood vessels of different radii can correspond to different types, such as arteries, veins, and capillaries, and how they interact with blood vessels of the same or different types can indicate important functional information about the system.
Here, we characterize such interactions through the lens of \name operator. Specifically, we consider cuboids of a fixed volume, which we refer to as zones, in different positions of the mouse brain, and in each zone, we filter vascular tubes, i.e., edges in the spatial networks, into two subsets according to their radii. We refer to the edges with radius greater than the threshold as ``high-radius'' edges and those with radius smaller than or equal to the threshold as ``low-radius'' edges. 
We then examine how the generalized ensnarlment changes as the threshold varies and across different regions, and how these features relate to the underlying functions.  
In this section, we present the results from three different zones; see Fig.~\ref{fig:real-all} for examples of the partitions in Zone I and the partition vessel count statistics for all three zones. We refer the reader to Appendix \ref{app:brain_vesculature_analysis} for more details.

\subsubsection{Null models}
We employ two different null models that prove useful for the downstream analysis: the Node-Spatial Edge-Erd\"{o}s-R\'{e}myi Model, and the Node-Spatial Edge-Random-Geometric Model. In the Node-Spatial Edge-ER Model (SER), the model has two parameters: the number of nodes $n$ and the number of edges $m$, with the space given. The model samples the positions of the $n$ nodes uniformly at random in the given space, and for each pair of nodes, they have the same probability $p=m/\binom{n}{2}$ of being connected by an edge, independent of their distance. The Node-Spatial Edge-Random-Geometric Model (SRG), on the other hand, more directly incorporates information about the spatial embedding of the model in its network sampling. The model also has two parameters: the number of nodes $n$ and the number of edges $m$, with the space given. The model also samples the positions of the $n$ nodes uniformly at random in the given space, but for each pair of nodes, in contrast to the SER model, the probability of an edge being present is proportional to the distance between the two end nodes, in such a way that the expected number of edges is equal to $m$. We also briefly explore a vessel edge radius shuffling null model that preserves the network structure in the appendix (see Fig.~\ref{fig:vessel-null-radius}).

\subsubsection{Results of brain vasculature}
We begin from the bipartite signed graph perspective, where we first quantify the system complexity by the unbalance score and then identify the part of the spatial networks where the complexity originates by the signed bi-clustering. 
Then, we further characterize the importance of edges by incorporating the information from the linking centrality.  
For visualization purposes, we select three representative radius threshold values for each zone: the first one with many more high-radius edges than low-radius ones, the second one with the two sets of relatively similar sizes, and the third one with many fewer high-radius edges than low-radius edges; see the black lines in Fig.~\ref{fig:real-all}C.  

When the threshold becomes too large, the set of high-radius edges becomes too few to form any cycles, thus exhibiting tree-like structure. \Name provides a systematic analysis of the two sets of edges as the threshold varies, regardless of the presence of cycles in either set of edges. 
In the results of intrinsic negative edges (Fig.~\ref{fig:real-all}D), we only show the edge pairs with highly negative values for visualization purposes: dark blue/red for values less than $-0.1$ and light blue/red for those less than $-0.05$. The region information contained in the data is from the Allen Brain Atlas \cite{sunkin2012atls}.

We track the change of complexity as the radius threshold $c$ varies from the smallest to the largest possible values, where the complexity is quantified by the unbalance score $\lambda_{min}$. 
We observe distinct features in the three different zones; see Fig.~\ref{fig:real-all}F. Specifically, Zone I is featured by a rapid decrease in the unbalance score after approximately $c=3$, Zone II is characterized by a plateau from around $c=2$ to $9$, and Zone III has first a slow decrease regime from about $c=2$ to $15$ and then a rapid decrease afterwards. 
Compared to the two null models, the above features are confirmed to be significant. 

For Zone I, the unbalance score is close to the values from the SRG model when $c$ is small, but becomes significantly smaller for $c > 9$; see Fig.~\ref{fig:real-all}F, (left). 
Meanwhile, the max linking centrality also exhibits a significant decrease at about $c=9$, and the mean linking centrality becomes further from the expected values of the SRG model afterwards; see Figs.~\ref{fig:real-all}G-H, (left).
Thus, after certain radius thresholds, simple distance constraints alone are not enough to explain the complexity in the system, and the system adopts a more efficient architecture on these scales, reducing the complexity originating from incomplete linking. 
These features indicate the possibility that there are two types of blood vessels in this brain zone, with their architecture fine-tuned in such a way that the linking complexity between the two is smaller than simple distance constraints would yield on their own. 
Because the related regions control life-sustaining functions and precise motor timing, we speculate this detected ``extra structure'' could be a signature of robustness and reliability.  

In Zone II, the unbalance score is close to the values from the SRG model when $c$ is either small or large, but goes above and forms a plateau in an intermediate regime when $5\le c \le 10$, toward the purely random SER model; see Fig.~\ref{fig:real-all}F, (middle). 
Consistently, the max linking centrality exhibits a significant increasing trend from $c=5$, followed by a decrease between $6$ and $7$, from the two null models, while the mean linking centrality is closer to the SRG model at the beginning and the end; see Figs.~\ref{fig:real-all}G-H, (middle).
This indicates that for intermediate radius thresholds, apart from the distance constraints, prominent features from the SER model, such as long-distance connections and other architectural motifs, start to play a role.
We note that in contrast to Zone I, this brain zone contains more heterogeneous regions, including the Basal Ganglia that is the gatekeeper of movement and reward, the Hippocampus Formation and the Ammon's Horn that are spatial navigation, and roots that are major vascular trunks or entry points that feed deeper structures, and the span of brain regions in this zone suggests a crossroads network that bridges distinct functions. 
Therefore, the cost of building longer, non-proximal edges is likely to be paid to ensure that regions like the Basal Ganglia and Hippocampus receive specialized and independent blood supplies that are not mixed with the metabolic demands of the intervening Midbrain tissue. 

For Zone III, the unbalance score is consistently below the values from the SRG model across all radius thresholds, Fig.~\ref{fig:real-all}F, (right). 
Meanwhile, the max linking centrality exhibits a significant increasing trend when $c$ is small, then a significant decrease at about $5$, followed by another significant decrease at $c=10$, while the mean linking centrality is consistently further below the expected values of the SRG model than the other two zones; see Figs.~\ref{fig:real-all}G-H, (right).
This emphasizes that the minimal wiring cost underlying the SRG model is overruled by a higher organizational principle, where the edges are carefully placed such that the complexity from the incomplete linking between them, no matter which radius threshold is chosen, becomes much smaller than expected. 
This high degree of organization is consistent with the highly specialized brain region in this zone, exclusively the Medulla. 
In contrast to the cerebellum, the medulla is centered on the brain's life support system, and the concentration of nodes in motor-related and behavioral state-related regions suggests that the vascular network may be optimized for extreme reliability and rapid response to systemic physiological changes, consistent with the detected highly organized features. 

%

\section{Discussion}
In this work, we introduced \Name, an edge-centric, geometry-aware framework for characterizing spatial networks. 
By shifting the analytical focus from nodes to edges, \Name is explicitly complementary to many existing spatial network methodologies, while also addressing a fundamental limitation of these methodologies which typically rely on cycle-based descriptors that can become unstable in realistic data regimes, assume reliable connectivity, or simply not apply if there are no cycles in the network. 

A central feature of \Name is the emergence of signed interactions driven by the relative orientation of edges in space. This feature is locally arbitrary, since the sign can flip by reversing the orientation of one edge, but globally intrinsic, since there may always be negative interactions.
In particular, the spectral properties remain invariant under edge orientation. 
The proposed unbalance score offers an interpretable measure of edge-level complexity, capturing how spatial organization departs from simple alignments. 
Through signed bi-clustering, \Name further enables the identification of edges and regions where the complexity of the organization originates, providing insights beyond what existing node-centric or purely topological methods can reveal. 

The validation in both synthetic lattices, with or without cycles, and real biological networks demonstrates the ability of \name to reveal latent organization, even capturing effective ensnarlment between cycle-free structures.
In the mouse brain vasculature, \Name revealed that in the cerebellum and medulla, flow alignment is prioritized in the network architecture, while in the bridging region including the midbrain, delivery between specific points is emphasized, highlighting how system-level complexity can be consistently identified in empirical systems. 

More broadly, this work opens several directions for future research. 
From a theoretical perspective, the operator provides a foundation for developing geometry-aware, edge-based dynamics and multiscale descriptors that extend beyond cycles or node-based views. 
The analysis also motivates further research towards an effective ``writhe-twist" decomposition of spatial networks, capturing the global coiling of the curve and the local rotation along the curve, respectively \cite{kamien2002geometry}. 
From a methodological standpoint, the framework can be naturally integrated with learning-based pipelines, enabling representation learning and inference on spatial networks, also in the case of incomplete connectivity. 
Finally, from the application point of view, the approach is broadly applicable to physical, biological, and engineering systems, including transport infrastructure, material structures, and tissue organizations, where geometry is measurable and relevant apart from topology. 

To summarize, our results establish incomplete linking as a meaningful and defining feature of spatial networks and provide a general, edge-centric framework for their characterization. 
By bridging geometric, signed, and spectral analysis, this work contributes a new perspective on spatial network organization that complements existing node-centric and topological approaches, opening up a robust pathway for analyzing complex spatial systems in realistic data settings.

\bibliography{refs}

\appendix

\section{Appendix}

\subsection{Incomplete Gauss Linking Integral}
\label{app:incomplete_GLI}

\begin{definition}[incomplete Gauss linking integral]
    Given two disjoint open or closed curves $l_1, l_2$, parametrized as $\gamma_1(s), \gamma_2(s)$, respectively, where $\gamma_1, \gamma_2: [0, 1]\to \mathbb{R}^3$, the following double integral gives the incomplete Gauss linking integral that characterizes the incomplete linking between $l_1$ and $l_2$:
    \begin{align}
        \label{eq:incomplete gli-outer}
        \mathfrak{L}(\gamma_1, \gamma_2) 
        &= \frac{1}{4\pi}\int_{0}^1\int_{0}^1({\gamma'}_1(s)\times {\gamma'}_2(t)) \cdot \frac{(\gamma_1(s) - \gamma_2(t))}{\norm{\gamma_1(s) - \gamma_2(t)}^3} ds dt\\
        &= \frac{1}{4\pi}\int_{0}^1\int_{0}^1\frac{\text{det}({\gamma'}_1(s), {\gamma'}_2(t), \gamma_1(s) - \gamma_2(t))}{\norm{\gamma_1(s) - \gamma_2(t)}^3} ds dt,
        \label{eq:incomplete gli-app}
    \end{align}
    where ${\gamma'}_1(s)$ and ${\gamma'}_2(t)$ are the derivatives of $\gamma_1(s)$ and $\gamma(t)$, respectively. 
    \label{def:incomplete GLI}
\end{definition}
Definition \ref{def:incomplete GLI} follows the same formula as the Gauss linking integral in the literature, e.g., in \cite{ricca2011gauss}, but relaxes the conditions for the underlying curves under consideration to not only closed but open ones. 

Note that the incomplete Gauss linking integral is still a symmetric, bilinear mapping. Therefore, (i) a flip in the orientation of one curve from $\gamma_i$ to $-\gamma_i$ will result in a sign change in the incomplete linking as 
\begin{align*}
    \mathfrak{L}(-\gamma_1, \gamma_2) = - \mathfrak{L}(\gamma_1, \gamma_2) = \mathfrak{L}(\gamma_1, -\gamma_2). 
\end{align*}
Furthermore, (ii) if we can rewrite any curves as the composition of ``smaller'' curves $\gamma_1 = \cup_{i=0}^k\sigma_{h_i}\gamma_{h_i}$ with $\sigma_{h_i} = \pm 1$ being the suborientation of $\gamma_{h_i}$ in $\gamma_1$, then the incomplete linking can be rewritten as 
\begin{align*}
    \mathfrak{L}(\gamma_1, \gamma_2) = \mathfrak{L}(\bigcup_{i=0}^k\sigma_{h_i}\gamma_{h_i}, \gamma_2) = \sum_{i=1}^k\sigma_{h_i}\mathfrak{L}(\gamma_{h_i}, \gamma_2). 
\end{align*}
Hence, as a prominent property of the incomplete linking, the linking number between any cycles can be retrieved from the incomplete linking between their composed open curves, as in Corollary \ref{cor:gli-decompose}. Therefore, the topological information is maintained among the curves. 
\begin{corollary}
    Suppose $\gamma_1, \gamma_2$ are two closed curves, and $\gamma_1 = \cup_{i=0}^{k_1}\gamma_{1i}, \gamma_2 = \cup_{j=0}^{k_2}\gamma_{1j}$, then the linking number between $\gamma_1, \gamma_2$, denoted by $l(\gamma_1, \gamma_2)$, can be obtained from the sum of incomplete linking: 
    \begin{align*}
        \mathfrak{L}(\gamma_1, \gamma_2) = \sum_{j=1}^{k_2}\sum_{i=1}^{k_1}\mathfrak{L}(\gamma_{1i}, \gamma_{2j}).
    \end{align*}
    \label{cor:gli-decompose}
\end{corollary}

\subsection{Invariance properties}
\label{app:incomplete_GLI-invariance}
Apart from the fundamental properties discussed before, we further show that the incomplete Gauss linking integral exhibits various invariance properties subject to isometric transformations. 

\begin{proposition}
    If $\gamma_1, \gamma_2$ are symmetric in the same way as $\gamma_3, \gamma_4$, where there exists an isometry $\Phi\in \text{Isom}(\mathbb{R}^3)$ such that $\Phi(\gamma_1) = \gamma_2$ and $\Phi(\gamma_3) = \gamma_4$ setwise,
    then $L(\gamma_2, \gamma_4) = \pm L(\gamma_1, \gamma_3)$. 
    \label{pro:gli-symm-1234}
\end{proposition}
\begin{proof}
    For $\Phi: \mathbb{R}^3\to \mathbb{R}^3$ be an Euclidean isometry, where
    \begin{align*}
        \Phi(\mathbf{x}) = \mathbf{R}\mathbf{x} + \mathbf{v}, 
    \end{align*}
    where $\mathbf{R}$ is an orthogonal matrix ($\mathbf{R}\in O(3)$), and $\mathbf{v}\in\mathbb{R}^3$. Hence, $\det(\mathbf{R})$ is either $1$ for rotation or $-1$ for reflection. Assume $\Phi(\gamma_1) = \gamma_2$ and $\Phi(\gamma_3) = \gamma_4$ setwise. Then, there exist differentiable monontone bijections between the parametrization parameters, where 
    \begin{align*}
        s = \phi_2(\hat{s}),\quad t = \phi_4(\hat{t}),
    \end{align*}
    such that 
    \begin{align*}
        \gamma_2(\hat{s}) = \Phi(\gamma_1(\phi_2(\hat{s}))),\quad \gamma_4(\hat{t}) = \Phi(\gamma_3(\phi_4(\hat{t}))),
    \end{align*}
    provide parametrizations of $\gamma_2$ and $\gamma_4$, repsectively. Let
    \begin{align*}
        \epsilon_2 \coloneqq \sign{(\phi'_2)},\quad \epsilon_4 \coloneqq \sign{(\phi'_4)},
    \end{align*}
    where $\epsilon_2 = 1$ if the induced parametrization of $\gamma_1$ via $\phi_2(\hat{s})$ preserves the curve orientation and $-1$ otherwise, and the same holds for $\epsilon_4$. 
    
    Then for the terms in the integrand of $L(\gamma_2, \gamma_4)$,
     \begin{align*}
         \gamma_2'(\hat{s}) = \mathbf{R}\gamma_1'(\phi_2(\hat{s}))\phi_2'(\hat{s}),\quad \gamma_4'(\hat{t}) = \mathbf{R}\gamma_3'(\phi_4(\hat{t}))\phi_4'(\hat{s}).
     \end{align*}
     Then, 
     \begin{align*}
         \gamma_2'(\hat{s})\times \gamma_4'(\hat{t}) 
         &= (\mathbf{R}\gamma_1'(\phi_2(\hat{s}))\phi_2'(\hat{s}))\times (\mathbf{R}\gamma_3'(\phi_4(\hat{t}))\phi_4'(\hat{t})) \\
         &= \phi_2'(\hat{s})\phi_4'(\hat{t})\det(\mathbf{R})\,\mathbf{R}\,(\gamma_1'(\phi_2(\hat{s}))\times \gamma_3'(\phi_4(\hat{t}))), 
     \end{align*}
     where the last equality is from the property of cross product and also that $\mathbf{R}$ is orthogonal. Therefore, 
     \begin{align*}
          &(\gamma_2'(\hat{s})\times \gamma_4'(\hat{t}))\cdot (\gamma_2(\hat{s})- \gamma_4(\hat{t})) \\
         =\, &\phi_2'(\hat{s})\phi_4'(\hat{t})\det(\mathbf{R})\,(\mathbf{R}\,(\gamma_1'(\phi_2(\hat{s}))\times \gamma_3'(\phi_4(\hat{t})))) \cdot (\mathbf{R}\, (\gamma_1(\phi_2(\hat{s}))- \gamma_3(\phi_4(\hat{t})))) \\
         =\, &\phi_2'(\hat{s})\phi_4'(\hat{t})\det(\mathbf{R})(\gamma_1'(\phi_2(\hat{s}))\times \gamma_3'(\phi_4(\hat{t}))) \cdot (\gamma_1(\phi_2(\hat{s}))- \gamma_3(\phi_4(\hat{t}))).
     \end{align*}

    We can then write 
    \begin{align*}
        L(\gamma_2, \gamma_4) 
        &= \frac{1}{4\pi}\int_{0}^1\int_{0}^1({\gamma_2'}(\hat{s})\times {\gamma_4'}(\hat{t})) \cdot \frac{(\gamma_2(\hat{s}) - \gamma_4(\hat{t})}{\norm{\gamma_2(\hat{s}) - \gamma_4(\hat{t})}^3} d\hat{s} d\hat{t}\\
        &= \frac{1}{4\pi}\int_{0}^1\int_{0}^1\phi_2'(\hat{s})\phi_4'(\hat{t})\det(\mathbf{R})(\gamma_1'(\phi_2(\hat{s}))\times \gamma_3'(\phi_4(\hat{t}))) \cdot \frac{\gamma_1(\phi_2(\hat{s}))- \gamma_3(\phi_4(\hat{t}))}{\norm{\gamma_1(\phi_2(\hat{s}))- \gamma_3(\phi_4(\hat{t}))}^3}d\hat{s} d\hat{t}\\
        &= \det(\mathbf{R})\epsilon_2\epsilon_4\frac{1}{4\pi}\int_{0}^1\int_{0}^1(\gamma_1'(s)\times \gamma_3'(t)) \cdot \frac{\gamma_1(s)- \gamma_3(t)}{\norm{\gamma_1(s)- \gamma_3(t)}^3}ds dt,
    \end{align*}
    where the second equality is also because orthogonal transformations perserve the norm, and $\epsilon_2,\epsilon_4$ are necessary in the third equality to align the orientations. Therefore, $L(\gamma_2, \gamma_4) = \pm L(\gamma_1, \gamma_3)$
\end{proof}

As an example, if $\gamma_2, \gamma_4$ are the mirror images of $\gamma_1, \gamma_3$, respectively, across the plane, with the orientation maintained, then $\det(\mathbf{R}) = -1$ and $\epsilon_2 = \epsilon_4 = 1$, thus $L(\gamma_2, \gamma_4) = - L(\gamma_1, \gamma_3)$. As another example, if $\gamma_2, \gamma_4$ are obtained by a rotation of $\gamma_1, \gamma_3$, respectively, along some direction, with both orientations flipped, then $\det(\mathbf{R}) = 1$ and $\epsilon_2 = \epsilon_4 = -1$, thus $L(\gamma_2, \gamma_4) = L(\gamma_1, \gamma_3)$. 

From Proposition \ref{pro:gli-symm-1234}, simple translations will not change the incomplete Gauss linking integral between two curves, and operations such as rotation and reflection can only change the sign of the incomplete Gauss linking integral between the curves.   
The cases where two curves have the same incomplete Gauss linking integral values with a third curve can be considered as a special case, as in Corollary \ref{col:gli-symm-123}.

\begin{corollary}
    If $\gamma_1, \gamma_2$ are symmetric with respect to a third curve $\gamma_3$, where there exists an ambient isometry $\Phi\in \text{Isom}(\mathbb{R})$ such that $\Phi(\gamma_1) = \gamma_2$ and $\Phi(\gamma_3) = \gamma_3$ setwise,
    then $L(\gamma_1, \gamma_3) = \pm L(\gamma_2, \gamma_3)$. 
    \label{col:gli-symm-123}
\end{corollary}

\begin{remark}
    As an example, if $\gamma_3$ is a straight line, the stabilizer\footnote{The stabilizer, also called the isotropy subgroup, of a curve in an isometry group is a subgroup of all isometries that map the curve to itself (as a set).} of $\gamma_3$ in the Euclidean isometry group consists of: (i) rotation along $\gamma_3$, (ii) reflections in planes containing $\gamma_3$, (iii) reflections in the perpendicular bisector plane through the midpoint and perpendicular to $l$, and (iv) rotation $\pi$ along any direction through the midpoint and perpendicular to $\gamma_3$. Any composition of these still map $\gamma_3$ to itself setwise, where the start and end points may swap. 
\end{remark}

In principle, the incomplete Gauss linking integral can return any value in the real line, including the origin. Specifically, we show that when two planar curves are parallel, their incomplete Gauss linking integral is $0$, as in Proposition \ref{pro:gli-parallel}. 
\begin{proposition}
    If two planar curves are in parallel, where $\gamma_1(t) - \gamma_2(t) = \mathbf{v}$ that is independent of $t$ and nonzero, and $\mathbf{v}$ is either in the plane or perpendicular to the plane, then $L(\gamma_1, \gamma_2) = 0$.
    \label{pro:gli-parallel}
\end{proposition}
\begin{proof}
    Let $\gamma: [0,1]\to \mathbb{R}^3$ be a $C^1$ open curve, and $\mathbf{v}\in \mathbb{R}^3$ be a constant vector. We write the two curves as
    \begin{align*}
        \gamma_1(s) = \gamma(s),\quad \gamma_2(t) = \gamma(t) + \mathbf{v},\quad s,t \in [0,1].
    \end{align*}
    Then, their Gauss linking integral can be written as 
    \begin{align}
        L(\gamma_1, \gamma_2) 
        &= \frac{1}{4\pi}\int_{0}^1\int_{0}^1({\gamma'}(s)\times {\gamma'}(t)) \cdot \frac{(\gamma(s) - \gamma(t) - \mathbf{v})}{\norm{\gamma(s) - \gamma(t) - \mathbf{v}}^3} ds dt. 
        \label{eq:incomplete gli parallel}
    \end{align}
    
    Assume that the curve $\gamma$ lies in an affine plane $P$ with unit normal vector $\mathbf{n}$, and suppose $\mathbf{v}$ also lies in the plane $P$, so $\mathbf{n}\cdot \mathbf{v} = \mathbf{0}$. Then $\gamma'(s), \gamma'(t) \in P$, thus $\gamma'(s)\times \gamma'(t)$ is parallel to $\mathbf{n}$. However, $\gamma(s) - \gamma(t) - \mathbf{v} \in P$, thus it is orthogonal to $\mathbf{n}$, and then the product in Eq.~\eqref{eq:incomplete gli parallel} vanishes pointwise, leading to $L(\gamma_1, \gamma_2) = 0$. 
    
    Now, we maintain the curve $\gamma$ to stay in the affine plane $P$ with unit normal vector $\mathbf{n}$, but take $\mathbf{v} = \lambda\mathbf{n}$ with $\lambda\ne 0$. We write $\gamma'(s)\times \gamma'(t) = J(s,t)\mathbf{n}$, where $J$ is a scalar with $J(t,s) = -J(s,t)$. Then the numerator of Eq.~\eqref{eq:incomplete gli parallel} is
    \begin{align*}
        ({\gamma'}(s)\times {\gamma'}(t)) \cdot (\gamma(s) - \gamma(t) - \mathbf{v}) = J(s,t)\mathbf{n}\cdot(-\mathbf{v}) = -\lambda J(s,t). 
    \end{align*}
    The denominator of Eq.~\eqref{eq:incomplete gli parallel} depends on 
    \begin{align*}
        \norm{\gamma(s) - \gamma(t) - \mathbf{v}}^2 = \norm{\gamma(s) - \gamma(t)} + \lambda^2,
    \end{align*}
    which is symmetric under swapping $s,t$. Therefore, integrating over the square $[0,1]^2$, the integrand cancels by anti-symmetry, leading to $L(\gamma_1, \gamma_2) = 0$.
    
    Alternatively, we can also show it through the definition via the average of half the algebraic sum of intercrossings in the projection of the two curves in any possible projection direction.
\end{proof}
\begin{remark}
    In the case when $\mathbf{v}$ has both in-plane and normal components, the symmetry is broken and $L$ is not necessarily $0$. Let us keep the curve $\gamma$ to stay in the affine plane $P$ with unit normal vector $\mathbf{n}$, but take $\mathbf{v} = \mathbf{v}_{\parallel} + \mathbf{v}_\perp$ with $\mathbf{v}_{\parallel}\in P$, $\mathbf{v}_\perp = \lambda\mathbf{n}$, with both parts being non-zero. Then the denominator of Eq.~\eqref{eq:incomplete gli parallel} becomes
    \begin{align*}
        \norm{\gamma(s) - \gamma(t) - \mathbf{v}}^2 = \norm{\gamma(s) - \gamma(t) - \mathbf{v}_{\parallel}} + \lambda^2,
    \end{align*}
    which is not symmetric in $(s,t)$ due to the in-plane shift $\mathbf{v}_{\parallel}$. Hence, the anti-symmetry that forced cancellation no longer applies, and the integral is in general nonzero. The case when the curves are not planar is more general.  
\end{remark}

As a special example, straight lines are always planar. Therefore, from Proposition \ref{pro:gli-parallel}, the incomplete Gauss linking integral is always $0$ if they are parallel to each other, as in Corollary \ref{cor:gli-parallel-line}. 
\begin{corollary}
    If two straight lines are in parallel, where $\gamma_1(t) - \gamma_2(t) = \mathbf{v}$ that is independent of $t$ and nonzero, then $L(\gamma_1, \gamma_2) = 0$. 
    \label{cor:gli-parallel-line}
\end{corollary}

\subsection{Special case: line segments}
In this section, we discuss a special case that bridges the incomplete Gauss linking integral with spatial networks: line segments. 
In this case, the incomplete Gauss linking integral can be written as an algebraic sum, instead of a double integral, as in Theorem \ref{the:gli-line} with proofs in \cite{klenin2000glicomp}. Therefore, the computation complexity can be largely reduced. 
\begin{theorem}[incomplete Gauss linking integral between line segments]
    When $\gamma_1, \gamma_2$ are line segments where $\gamma_1$ is from $a$ to $b$ and $\gamma_2$ is from $c$ to $d$ ($a,b,c,d \in \mathbb{R}^3$), the computation of the incomplete Gauss linking integral can be simplified as 
    \begin{align}
        L(\gamma_1, \gamma_2) = \frac{\sigma_{12}}{4\pi}\left(\arcsin(\mathbf{n}_a\mathbf{n}_d) + \arcsin(\mathbf{n}_d\mathbf{n}_b) + \arcsin(\mathbf{n}_b\mathbf{n}_c) + \arcsin(\mathbf{n}_c\mathbf{n}_a)\right),
        \label{eq:incomplete gli-line}
    \end{align}
    where
    \begin{align*}
        \mathbf{n}_a = \frac{\overrightarrow{ac} \times \overrightarrow{ad}}{\norm{\overrightarrow{ac} \times \overrightarrow{ad}}}, \quad 
        \mathbf{n}_b = \frac{\overrightarrow{bd} \times \overrightarrow{bc}}{\norm{\overrightarrow{bd} \times \overrightarrow{bc}}}, \quad 
        \mathbf{n}_c = \frac{\overrightarrow{bc} \times \overrightarrow{ac}}{\norm{\overrightarrow{bc} \times \overrightarrow{ac}}}, \quad
        \mathbf{n}_d = \frac{\overrightarrow{ad} \times \overrightarrow{bd}}{\norm{\overrightarrow{ad} \times \overrightarrow{bd}}}, 
    \end{align*}
    and 
    \begin{align*}
        \sigma_{12} = \sign{((\overrightarrow{cd}\times \overrightarrow{ab})\cdot\overrightarrow{ac})}.
    \end{align*}
    \label{the:gli-line}
\end{theorem}

As in Eq.~\eqref{eq:incomplete gli-line}, the integral now only depends on the position of the endpoints of the line segments, which can be interpreted geometrically as the signed solid angle subtended by the tetrahedron formed by the four endpoints of the segments.
This representation provides an intuitive picture of how the interaction arises and simultaneously captures topological and geometric information: segments that are close and arranged in configurations resembling intertwined or crossing structures yield larger contributions, whereas distant or nearly parallel segments contribute little. 
Importantly, the sign of the contribution is determined by the relative orientation of the segments, representing the direction of their winding in space.

We now show how the incomplete Gauss linking integral between two line segments changes, subject to transformations of one of the two lines. 
Specifically, we maintain one line segment fixed along the x-axis, and the second one is obtained by first a translation along the y-axis, with various lengths, and then a rotation in the x-z plane around the y-axis, with a full spectrum of angles from $0$ to $2\pi$. 
As expected from Corollary \ref{col:gli-symm-123}, we observe a periodic behavior, in terms of the absolute value, as the rotation angle changes from $0$ to $2\pi$; see Fig.~\ref{fig:gli-line-rot2} (left) (also in Fig.~\ref{fig:schematic}B). We also observe that as one increases the length of translation between the two line segments, the Gauss linking integral values become smaller, as inferred from Theorem \ref{the:gli-line}. 

Furthermore, as the line segment further rotates in the x-y plane around the z-axis, periodic behavior in the incomplete Gauss linking integral remains; see Fig.~\ref{fig:gli-line-rot2} right. In particular, for each rotation angle $\theta_1 \ne 0, \pi, 2\pi$ in the x-z plane, as we add another rotation $\theta_2$ in the x-y plane, the incomplete Gauss linking integral first becomes larger in absolute value as $\theta_2$ changes from $0$ to $\pi/2$, then it becomes smaller in absolute value as it further changes from $\pi/2$ to $\pi$, which can be inferred from Theorem \ref{the:gli-line}; this process repeats as it further changes from $\pi$ to $3\pi/2$ and further to $2\pi$, which is expected from Corollary \ref{col:gli-symm-123}. The whole process will not change the sign of the incomplete Gauss linking integral obtained from the first rotation. 
\begin{figure}[htbp]
    \centering
    \begin{tabular}{cc}
        \includegraphics[width=0.45\linewidth]{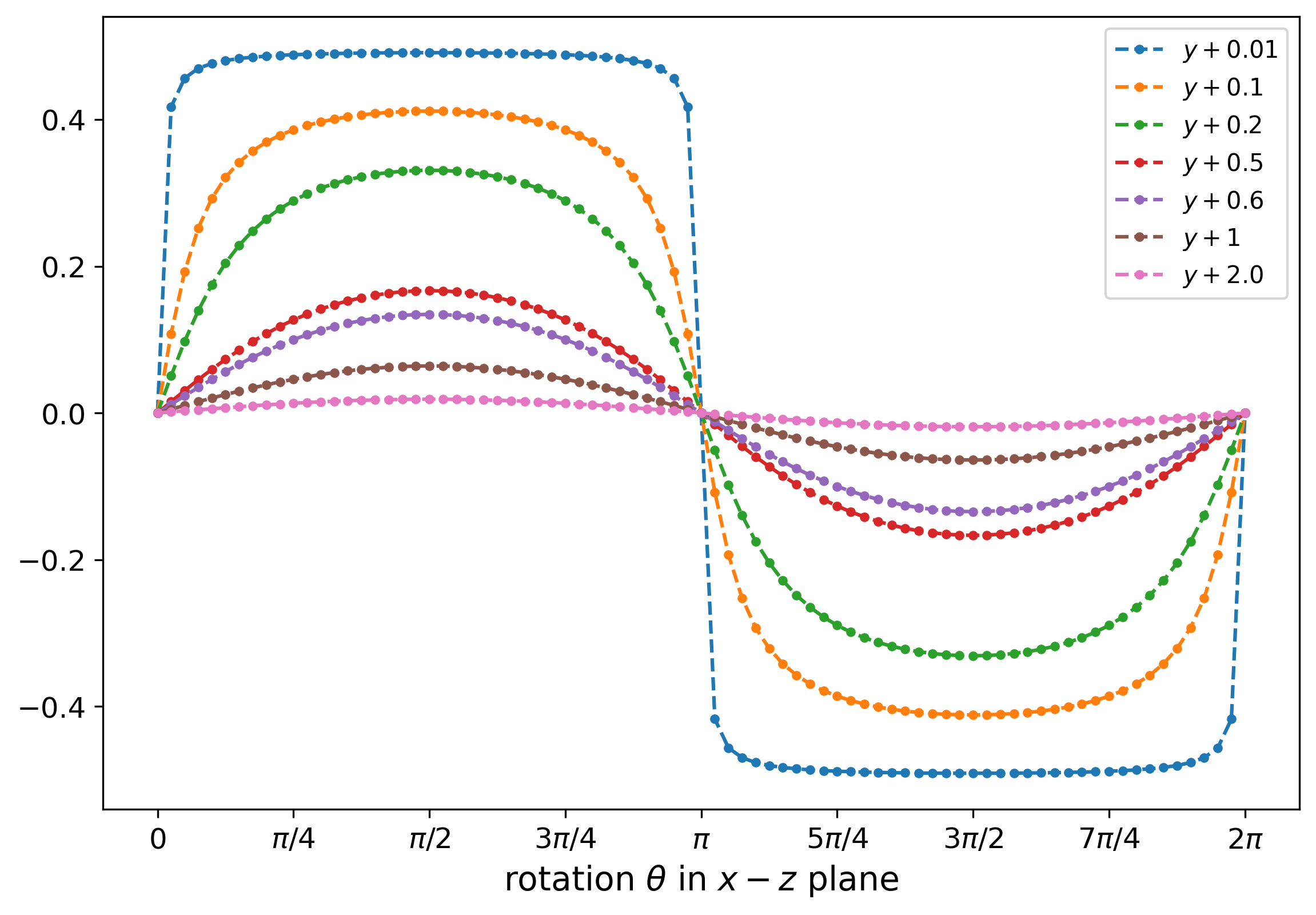} & \includegraphics[width=0.5\linewidth]{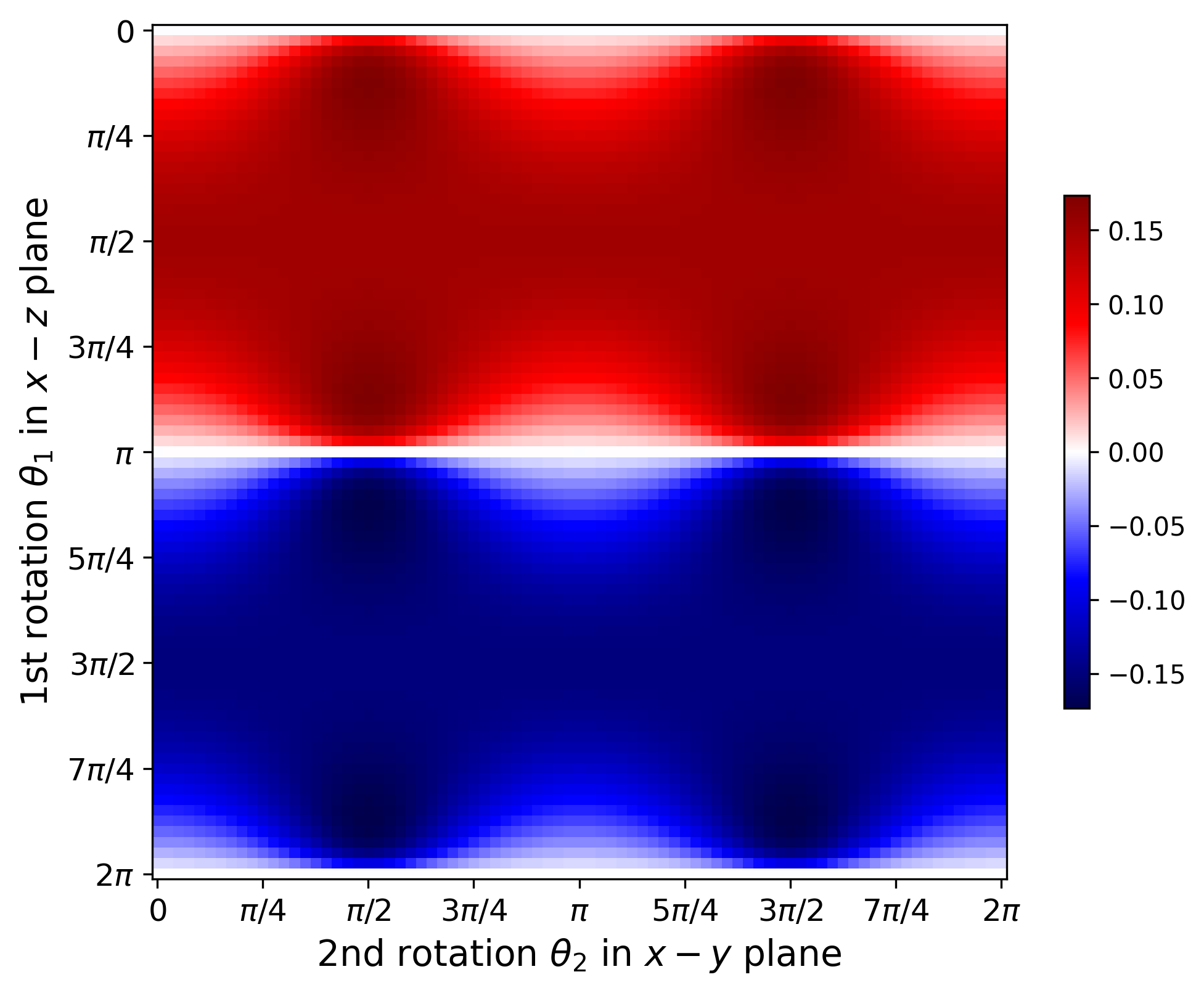}
    \end{tabular}
    \caption{Incomplete Gauss linking integral between a line segment $\gamma_1$ from $(-0.5, 0, 0)$ to $(0.5, 0, 0)$, and $\gamma_2$ obtained by different transformations: \add{(left) first translating $\gamma_1$ along the y-axis with different lengths (different colors) and then rotating angle $\theta$ in the $x-z$ plane around the y-axis ; (right)} first translating $\gamma_1$ along the y-axis with length $0.55$, second rotating angle $\theta_1$ in the x-z plane around the y-axis and then rotating angle $\theta_2$ in the x-y plane around z-axis. }
    \label{fig:gli-line-rot2}
\end{figure}

\subsection{\Name operator and its properties}
\label{app:gemini-property}

One may consider a form of the \name operator that explicitly contains all choices of arbitrary direction along each edge. Specifically, we denote the operator that contains the incomplete linking between each pair of edges in both directions as $\mathbf{\Lambda}_f \in \mathbb{R}^{2m_1\times 2m_2}$, then
\begin{align*}
    \mathbf{\Lambda}_f = 
    \begin{pmatrix}
        \mathbf{\Lambda}_g & -\mathbf{\Lambda}_g\\
        -\mathbf{\Lambda}_g & \mathbf{\Lambda}_g
    \end{pmatrix},
\end{align*}
where the first $m_1$ rows correspond to all edges in $E_1$ along the direction in the rows of $\mathbf{\Lambda}_g$ and the following $m_1$ rows correspond to all edges in the same order but along the opposite direction, and the columns are aligned in the same way but for edges in $E_2$. 
Therefore, all computations on $\mathbf{\Lambda}_f$ can be interpreted instead through the \name operator $\mathbf{\Lambda}_g$. In particular, the eigen-space of $\mathbf{\Lambda}_f$ can be completely characterized by the eigen-space of $\mathbf{\Lambda}_g$, as in Proposition \ref{pro:net-lambda-f}. 

\begin{proposition}
    If $\mathbf{u}, \mathbf{v}$ is a left-right eigenvector pair of $\mathbf{\Lambda}_g$ with respect to eigenvalue $\lambda$, then $(\mathbf{u}; -\mathbf{u}), (\mathbf{v}; -\mathbf{v})$ is a left-right eigenvector pair of $\mathbf{\Lambda}_f$ with respect to eigenvalue $2\lambda$, and $(\mathbf{u}; \mathbf{u}), (\mathbf{v}; \mathbf{v})$ is a left-right eigenvector pair of $\mathbf{\Lambda}_f$ with respect to eigenvalue $0$.
    \label{pro:net-lambda-f}
\end{proposition}
\begin{proof}
    Suppose $\mathbf{u}, \mathbf{v}$ is a left-right eigenvector pair of $\mathbf{\Lambda}_g$ with respect to the eigenvalue $\lambda$, then 
    \begin{align*}
        \mathbf{\Lambda}_g\mathbf{v} = \lambda\mathbf{u},\quad \mathbf{\Lambda}_g^T\mathbf{u} = \lambda\mathbf{v}.
    \end{align*}
    From $(\mathbf{v}; -\mathbf{v})$, we have
    \begin{align*}
        \mathbf{\Lambda}_f
        \begin{pmatrix}
            \mathbf{v}\\
            -\mathbf{v}
        \end{pmatrix} = 
        \begin{pmatrix}
        \mathbf{\Lambda}_g & -\mathbf{\Lambda}_g\\
        -\mathbf{\Lambda}_g & \mathbf{\Lambda}_g
    \end{pmatrix}
    \begin{pmatrix}
            \mathbf{v}\\
            -\mathbf{v}
        \end{pmatrix}
        = 
        \begin{pmatrix}
            2\mathbf{\Lambda}_g\mathbf{v}\\
            -2\mathbf{\Lambda}_g\mathbf{v}
        \end{pmatrix}
        =2\lambda
        \begin{pmatrix}
            \mathbf{u}\\
            -\mathbf{u}
        \end{pmatrix}.
    \end{align*}
    The case for $(\mathbf{u}; -\mathbf{u})$ can be shown in a similar manner. Therefore, $(\mathbf{u}; -\mathbf{u}), (\mathbf{v}; -\mathbf{v})$ is a left-right eigenvector pair of $\mathbf{\Lambda}_f$ with respect to eigenvalue $2\lambda$. 
    Finally, it is straightforward to check that $(\mathbf{u}; \mathbf{u}), (\mathbf{v}; \mathbf{v})$ are left, right eigenvectors, respectively, associated with the eigenvalue $0$. 
\end{proof}
We should note that for all nonzero $\mathbf{x}\in \mathbb{R}^{m_2}$, $(\mathbf{x}; \mathbf{x})$ is a right eigenvector associated with the eigenvalue $0$. Therefore, the extra $m_2$ dimensions added to the columns of the full operator $\mathbf{\Lambda}_f$ does not contribute to the nonzero eigenvalues, and so do the extra $m_1$ dimensions added to the rows. Therefore, the overall rank of the two operators remain the same, as in Corollary \ref{cor:lambda-f-rank}.  

\begin{corollary}
    The full operator $\mathbf{\Lambda}_f$ has the same rank as the \name operator $\mathbf{\Lambda}_g$.
    \label{cor:lambda-f-rank}
\end{corollary}

\begin{definition}[Switching]
    The operation of reversing the signs of all edges connecting a subset $S\subseteq V$ and its complement $V\backslash S$ is called switching the subset $S$. Two signed configurations $\sigma, \sigma': E\to \{-1,1\}$ are said to be switching equivalent if there exists $S\subseteq V$ such that $\sigma'$ can be obtained from $\sigma$ by switching the subset $S$, denoted by $\sigma \approx \sigma'$. 
    \label{def:switching}
\end{definition}
Switching equivalence is an equivalence relation on sign configurations $\sigma$ of a fixed underlying graph, and the corresponding equivalent classes are called switching classes. We can show that the bipartite signed graphs derived from the \name operator, subject to different orientations of the spatial edges, are switching equivalent, as in Lemma \ref{lem:swtiching-equiv}. As a prominent feature of switching equivalence, the spectra of signed graphs remains invariant in the same switching equivalence class \cite{Atay_signedCheeger_2020,zaslavsky1982signed}. Therefore, we can obtain the invariance property of the \name operator as in Theorem \ref{the:switching-eigvals}. 
\begin{lemma}
    The bipartite signed graphs derived from the \name operator, subject to different orientations of the spatial edges, belong to the same switching equivalence class. 
    \label{lem:swtiching-equiv}
\end{lemma}
\begin{proof}
    Let us start from one arbitrary orientation of the edges, construct the \name operator $\mathbf{\Lambda}_g$, and derive the bipartite signed graph from $\mathbf{\Lambda}_g$. Now, suppose we reverse the orientation of one edge $e\in E_1$, then the corresponding row in $\mathbf{\Lambda}_g$ will have the opposite sign. This transfers to reversing the sign of all edges connecting with the corresponding \add{node} in $e\in V_b$, i.e., switching the subset $S = \{e\}$. Therefore, the new bipartite signed graph from this reversed orientation of one edge $e\in E_1$ in computing the \name operator is switching equivalent to the original bipartite signed graph, by Definition \ref{def:switching}. The case of reversing an arbitrary number of edges in $E_1$ and/or $E_2$ can be obtained by induction. 
\end{proof}
 
\begin{theorem}
    The eigenvalues of the \name operator remain the same, subject to different orientations of the spatial edges.
    \label{the:switching-eigvals}
\end{theorem}
\begin{proof}
    From Lemma \ref{lem:swtiching-equiv}, we know that the bipartite signed graphs derived from the \name operator, subject to different orientations of edges, belong to the same switching equivalence class; therefore, from the spectral invariance property of the switching equivalence class, the eigenvalues of $\mathbf{W}_b$ remain the same. We know that for bipartite graphs, the weight matrix $\mathbf{W}_b$ share the same eigenvalues as the bipartite weight matrix $\mathbf{\Lambda}_g$ (but the multiplicity doubles). Therefore, the eigenvalues of the \name operator $\mathbf{\Lambda}_g$ remain invariant subject to different orientations of edges in the original spatial networks. 
\end{proof}

\subsubsection{Signed bi-clustering}
\label{app:signed_analysis}
In particular, we apply the following spectral clustering algorithm for signed bi-clustering problem; see Algorithm \ref{alg:sc_signed}. Specifically, based on the vanilla spectral clustering algorithm in \cite{Kunegis_signspect_2010}, we add two extra steps: to combine information in multiple eigenvectors associated with the smallest eigenvalue, and to allocate nodes with $0$ values in the eigenvector(s).  

\begin{algorithm}[H]
		\caption{Spectral clustering algorithm on signed networks.}
		\label{alg:sc_signed}
		\begin{algorithmic}[1] 
			\State{Input: signed network $G$ with weight matrix $\mathbf{W}$.} 
			\State{Construct the signed graph Laplacian $\mathbf{L} = \mathbf{D} - \mathbf{A}$, where $\mathbf{D}$ is a diagonal matrix and the $(i,i)$ element is $d_i = \sum_{j}\abs{W_{ij}}$, the degree of node $i$.} 
			\State{Compute the set of eigenvector(s) associated with the smallest eigenvalue of $\mathbf{L}$.}
            \State{Compute the set of linear combinations of the smallest eigenvectors that have a different sign distribution on their elements, denoted by $\mathcal{X} = \{\mathbf{x}_i\}_{i=1}^{n_c}$ and $|\mathcal{X}| = n_c$.}
            \For{$i=1$ to $n_c$}
			\State{For $j = 1,\dots, n$, let $(\mathbf{x}_i)_j$ be the value corresponding to the $j$-th element of $\mathbf{x}_i$. Cluster node $1$ into community 1 if  $(\mathbf{x}_i)_j > 0$, and community 2 if $(\mathbf{x}_i)_j < 0$. This gives the initial clustering $C_i$.}
			\If{The set $\mathcal{U} = \{j: (\mathbf{x}_i)_j=0\}$ is not empty}
			\State{Construct the power set of $\mathcal{U}$, $\mathcal{P}(\mathcal{U})$.}
            \State{For each set $\mathcal{S} \in \mathcal{P}(\mathcal{U})$, create a new clustering $C'_i$ from $C_i$ by adding nodes in $\mathcal{S}$ into community 1 and the remaining nodes in $\mathcal{U}$ into community 2. Evaluate the objective value with clustering $C'_i$.}
            \State{Update the initial clustering $C_i$ with the $C'_i$ of the optimal objective value.}
			\EndIf
            \EndFor
            \State{Choose the optimal clustering $C^*$ from the set $\{C_i\}_{i=1}^{n_c}$.}
			\State{Output: clustering $C^*$.} 
		\end{algorithmic}
	\end{algorithm}

\subsection{Ladder and Step Lattice Examples}
\label{app:ladder-step-egs}

We here consider the \add{ensnarled ladder lattices}, and construct the \name operator accordingly; see Fig.~\ref{fig:eg-ladder}. The incomplete linking values between edge pairs $\{(0,2), (1,3)\}\times \{(4,6), (5,7), (6,8), (7,9)\}$ are zero because they are parallel to each other by Corollary \ref{cor:gli-parallel-line}. Meanwhile, the absolute values of the incomplete linking are the same between edge pairs $\{(0,1)\}\times \{(4,5), (4,6), (5,7), (6,7)\}$, $\{(2,3)\}\times \{(6,7), (6,8), (7,9), (8,9)\}$ and $\{6,7)\}\times \{(0,1), (0,2), (1,3), (2,3)\}$ because they are ``symmetric'' to each other in the way that is described in Proposition \ref{pro:gli-symm-1234}. The case for the remaining edge pairs, which also have the same absolute value of the incomplete linking, can be explained in the same way.

\begin{figure}[t]
    \centering
    \includegraphics[width=.98\textwidth]{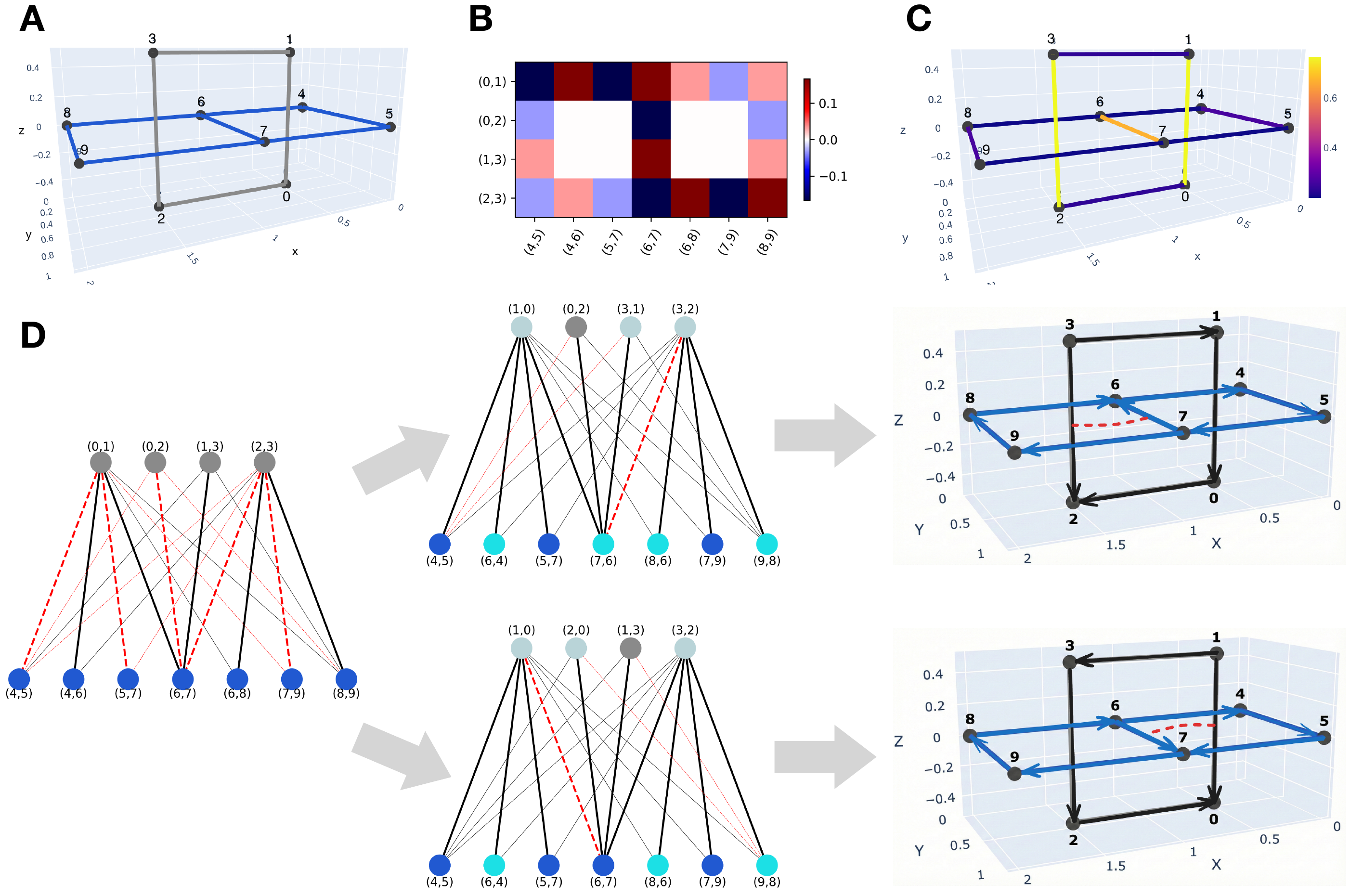}
    \caption{Example of the ensnarled ladder lattices with period 1. \textbf{A}. Visualization of the spatial network. \textbf{B}. \name operator from edges in gray to edges in blue. \textbf{C}. Linking centrality of edges in A. \textbf{D}. Bipartite signed graphs: each spatial edge is oriented (left) from the node with a smaller label to the node with a larger label, and then (middle) according to the optimal orientations that minimize the negative incomplete linking, with spatial edges whose orientations have been reversed marked in a lighter color, where the top and bottom figures correspond to two equivalent solutions. (right) The corresponding orientation of edges in \textbf{A} is shown on the right, where a red dashed line connects the edge pair with the most negative linking. }
    \label{fig:eg-ladder}
\end{figure}

Next, we take into account the signs and examine the bipartite signed graph derived from the \name operator; see Fig.~\ref{fig:eg-ladder}D. We measure how far it is from structural balance by the smallest eigenvalue of the signed random-walk Laplacian, \add{i.e., the unbalance score,} where $\lambda_{\min} = 0.153$; 
therefore, the signed graph is not far from balance, and the signed bi-clustering is expected to perform reasonably well, although the unbalance is also not negligible. 
Indeed, the signed bi-clustering provides two optimal solutions where the difference is in the allocation of nodes $(0,2), (1,3), (6,7)$; see Fig.~\ref{fig:eg-ladder}D. 
We should note that after removing these nodes, the signed graph will be balanced; in fact, as we will see later, we only need to remove two nodes to make the signed graph balanced. 

\textit{What are the intrinsic negative edges in the bipartite signed graph?} We show the signed graphs after switching one set of nodes, equivalently changing the orientation of the corresponding edges in the regular lattices, from the optimal solutions of the signed bi-clustering in Fig.~\ref{fig:eg-ladder}D (middle). We see that there is always a negative edge between $(7,6)$ and either $(1,0)$ or $(3,2)$, and there are always two negative edges between node in $\{(0,2), (3,1)\}$ and either $(4,5)$ or $(9,8)$. Therefore, if we remove edges $(7,6)$ and $(4,5)$ from the regular lattices, one of many choices, the derived signed graph will be balanced. 
\textit{Where do the intrinsic negative edges originate?} By removing spatial edges $(7,6)$ and $(4,5)$, the remaining lattices have the property that for each pair of remaining oriented edges in $E_2$, if their angle is $0$, they are placed in the same ``side'' for each edge in $E_1$, and if their angle is $\pi$, they are placed in the different ``sides'' for each edge in $E_1$, where the space is divided into different sides by the spatial edges in $E_1, E_2$ under consideration. This property is immediately broken once we add any of the two removed spatial edges. In general, we expect a negative edge in the bipartite signed graphs to occur when spatial edges in one set wrap around spatial edges in the other set, which contributes to the ``linking'' of the two sets of edges in an incomplete sense. We leave further detailed mathematical exploration of this concept to future work.

Finally, we explore the newly proposed linking centrality and find that it can identify the edges that are important in the linking between the two sets of edges; see Fig.~\ref{fig:eg-ladder}C. We should note that topologically speaking, not only the three of the highest centrality measures, breaking any of edges in $\{(1,3), (0,2)\}$ will also unlink the two set of edges; however, incorporating the geometric perspective, edge $(6,7)$ really goes through the other set of edges, or inside their cycle and out, similar for edges $\{(0,1), (2,3)\}$ of the highest linking centrality. Therefore, it is reasonable to identify these edges as more important. Meanwhile, this also implies that the \name operator can combine both topological and geometric perspectives.

\begin{figure}[t]
    \centering
    \includegraphics[width=.98\textwidth]{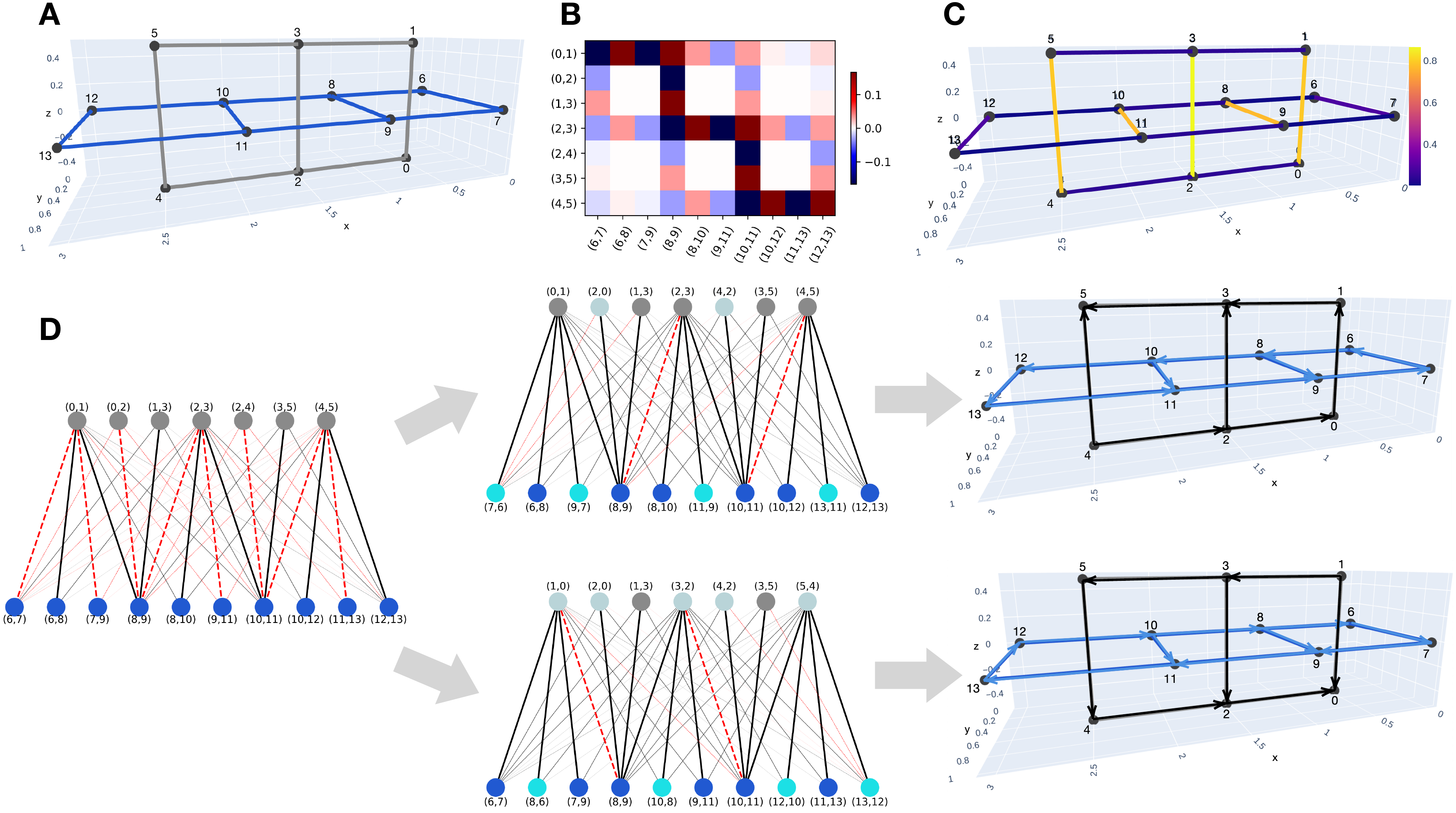}
    \caption{Example of the ensnarled ladder lattices with period 2. \textbf{A}. Visualization of the spatial network. \textbf{B}. \name operator from edges in gray to edges in blue. \textbf{C}. Linking centrality of edges in A. \textbf{D}. Bipartite signed graphs: each spatial edge is oriented (left) from the node with a smaller label to the node with a larger label, and then (middle) according to the optimal orientations that minimize the negative incomplete linking, with spatial edges whose orientations have been reversed marked in a lighter color, where the top and bottom figures correspond to two equivalent solutions. (right) The corresponding orientation of edges in \textbf{A} is shown on the right, where a red dashed line connects the edge pair with the most negative linking. }
    \label{fig:lad-p2}
\end{figure}

We can further increase the size of the ensnarled ladder lattices, and similar results follow; see Fig.~\ref{fig:lad-p2}. Specifically, similar patterns in the \name operator occur periodically, which can be explained in the same way as the case of period $1$ in Fig.~\ref{fig:eg-ladder}.

In terms of the derived signed bipartite graph, now the unbalance score becomes $\lambda_{\min} = 0.183$, larger  
than the previous case; see Fig.~\ref{fig:lad-p2}D (left). This is expected, because the intrinsic negative edges, or conflicts in the ensnarled ladder lattices in Fig.~\ref{fig:eg-ladder}, will repeatedly occur as we increase the size of the lattices. The repeated conflicts become clearer after we perform signed bi-clustering and apply switching to one set of the nodes; see Fig.~\ref{fig:lad-p2}D (middle). 
\add{The signed bi-clustering again provides two optimal solutions, which naturally extend the results from the ensnarled ladder lattices of period $1$. Meanwhile, the spatial edges of high linking centrality also correspond to nodes in the signed graph connected by strong negative edges; see Fig.~\ref{fig:lad-p2}D (right).  }

\begin{figure}[htbp]
    \centering
    \includegraphics[width=.98\textwidth]{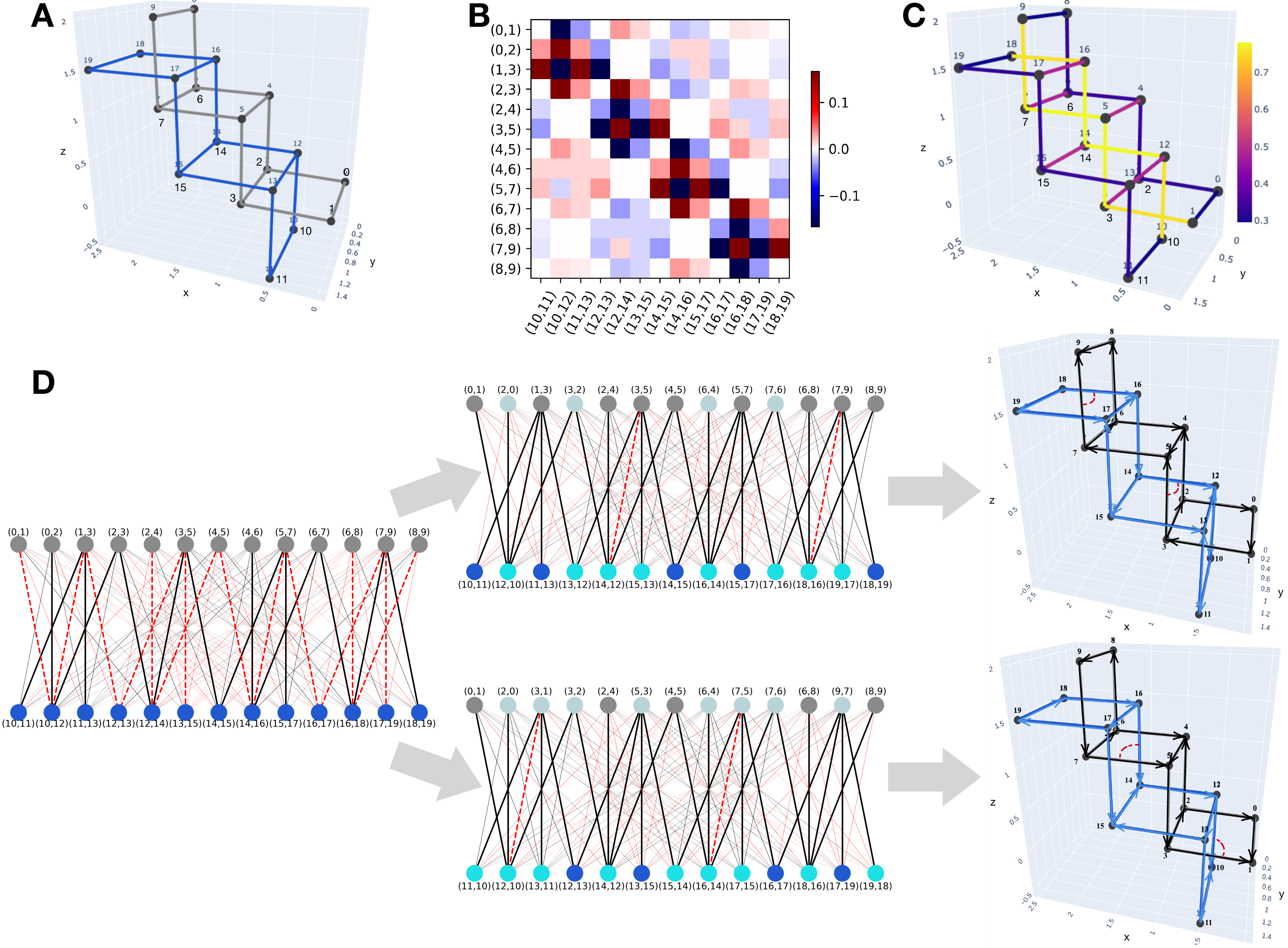}
    \caption{Example of the ensnarled step lattices with period 2. \textbf{A}. Visualization of the spatial network. \textbf{B}. \name operator from edges in gray to edges in blue. \textbf{C}. Linking centrality of edges in A. \textbf{D}. Bipartite signed graphs: each spatial edge is oriented (left) from the node with a smaller label to the node with a larger label, and then (middle) according to the optimal orientations that minimize the negative incomplete linking, with spatial edges whose orientations have been reversed marked in a lighter color, where the top and bottom figures correspond to two equivalent solutions. (right) The corresponding orientation of edges in \textbf{A} is shown on the right, where a red dashed line connects the edge pair with the most negative linking. }
    \label{fig:step-p2}
\end{figure}

For the ensnarled step lattices, the \name operator can capture its intrinsic symmetries through its eigenvalues as in Proposition \ref{pro:step-eigs}.
\begin{proposition}
    For each nonzero eigenvalue of the \name operator of the ensnarled step lattices, it always has multiplicity $2$, i.e., it always has two linearly independent eigenvectors associated.
    \label{pro:step-eigs}
\end{proposition}
\begin{proof}
    For ensnarled step lattices, there are two intrinsic symmetries: the first one is inside each edge set, and the second one is between the two edge sets. On the one hand, the unit is symmetric with respect to the plane where the shared edges of the two squares in the two lattices lie: for each specific edge $e$, its corresponding symmetric edge is in the same lattice as $e$'s. Therefore, each individual lattice is symmetric with respect to this plane. On the other hand, the unit is also symmetric with respect to the midpoint between the midpoints of the two shared edges. With the latter, for each specific edge $e$, we can find another edge $e'$ in the other lattice such that $e,e'$ are symmetric with respect to the point. We also note that the two lattices have the same size. These results can be extended to ensnarled step lattices of arbitrary sizes.  
    
    From the second point symmetry between the two edge sets, we can write the \name operator as a symmetric matrix, subject to appropriate alignment of the two sets of edges such that the $i$-th edge in $E_1$ is plane symmetric to the $i$-th edge in $E_2$. We can change the orientations when necessary, where the eigenvalues are invariant by Theorem \ref{the:switching-eigvals}. 
    From the first plane symmetry inside each edge set, for each eigenpair $(\lambda, \mathbf{x})$ of the \name operator, we can rearrange the eigenvector for a new nonzero vector $\mathbf{x}'$ in the following way such that $(\lambda, \mathbf{x}')$ is also an eigenpair of the \name operator: suppose edges corresponding to dimensions $i,j$ are symmetric to each other, then $x_i' = x_j$ and $x_j' = x_i$.  
    Therefore, for each nonzero eigenvalue, we can always find two linearly independent eigenvectors associated. 
\end{proof}

We can further increase the size of the ensnarled step lattices, and similar results follow; see Fig.~\ref{fig:step-p2}. Specifically, the highly regular patterns occur periodically in the \name operator, and the prominent eigen-property of all eigenvalues with multiplicity $2$ is maintained, which can be explained in the same manner as in Fig.~\ref{fig:eg-step-p1}.

For the derived bipartite signed graph, now the unbalance measure is $\lambda_{\min} = 0.349$, larger 
than the previous case; see Fig.~\ref{fig:step-p2}D for the signed graphs after optimal switching. This is expected because the intrinsic negative edges, or conflicts in the step lattices in Fig.~\ref{fig:eg-step-p1} will repeatedly occur as we increase the size of the lattices (cf. Fig.~\ref{fig:eg-step-p1}D). We also note that there are still two optimal solutions to the signed bi-clustering, because the symmetries are maintained as a step lattice repeats its fundamental unit. The maintenance of symmetries is also evident from the linking centrality; see Fig.~\ref{fig:step-p2}C. The results can naturally extend to lattices of arbitrary sizes.

\subsubsection{Removing edges from ensnarled ladder lattices}

We have distinct results as we remove edges from the smaller lattice vs the larger one in the two ensnarled ladder lattices. 
Take the ensnarled ladder lattices of period 2 as an example. 
For the smaller lattice, the cyclic structure breaks immediately after we remove two edges from it; however, there is still some unbalance, of a lower level, in the two networks; see Fig.~\ref{fig:lad-ps-remv} (row 3). After we remove both edges that are of the second-highest centrality, the unbalance between the two becomes $0$. 
While for the larger lattice, at least three edges have to be removed to break the cyclic structure, and the unbalance between the two networks also disappears after removing the three edges; see Fig.~\ref{fig:lad-ps-remv} (row 4). 
The distinction between the two lattices also emphasizes the asymmetric roles of the two in their ensnarlment. where the smaller lattice plays a more important role. Furthermore, we observe that the first two edges removed in the larger lattice have the same centrality value, which is also much larger than the remaining edges, and removing them leads to almost the same decrease in the unbalance score, which is larger than removing the third one. 

After removing the first two edges in the smaller lattice, it becomes a tree, but with our \name operator, we can still characterize its generalized ensnarlment with the other lattice, which may or may not contain cycles; see Fig.~\ref{fig:lad-ps-remv} (row 3-4: right). 
We note that the unbalance score becomes $0$ after we remove $3$ edges in both lattices, i.e., the size of the cycle basis of the larger lattice. The results can also be extended for ensnarled lattices of varying sizes; see Fig.~\ref{fig:lad-ps-remv} (rows 1-2, 5-6). 
\begin{figure}[htbp]
    \centering
    \begin{tabular}{lc}
        \includegraphics[width=0.72\textwidth]{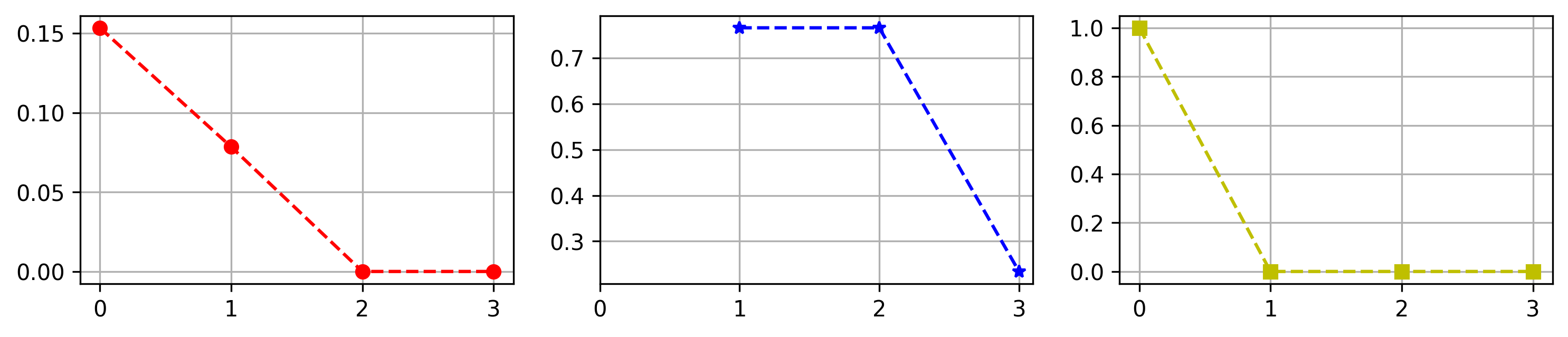} & 
        \multirow{2}{*}{\includegraphics[width=.25\textwidth]{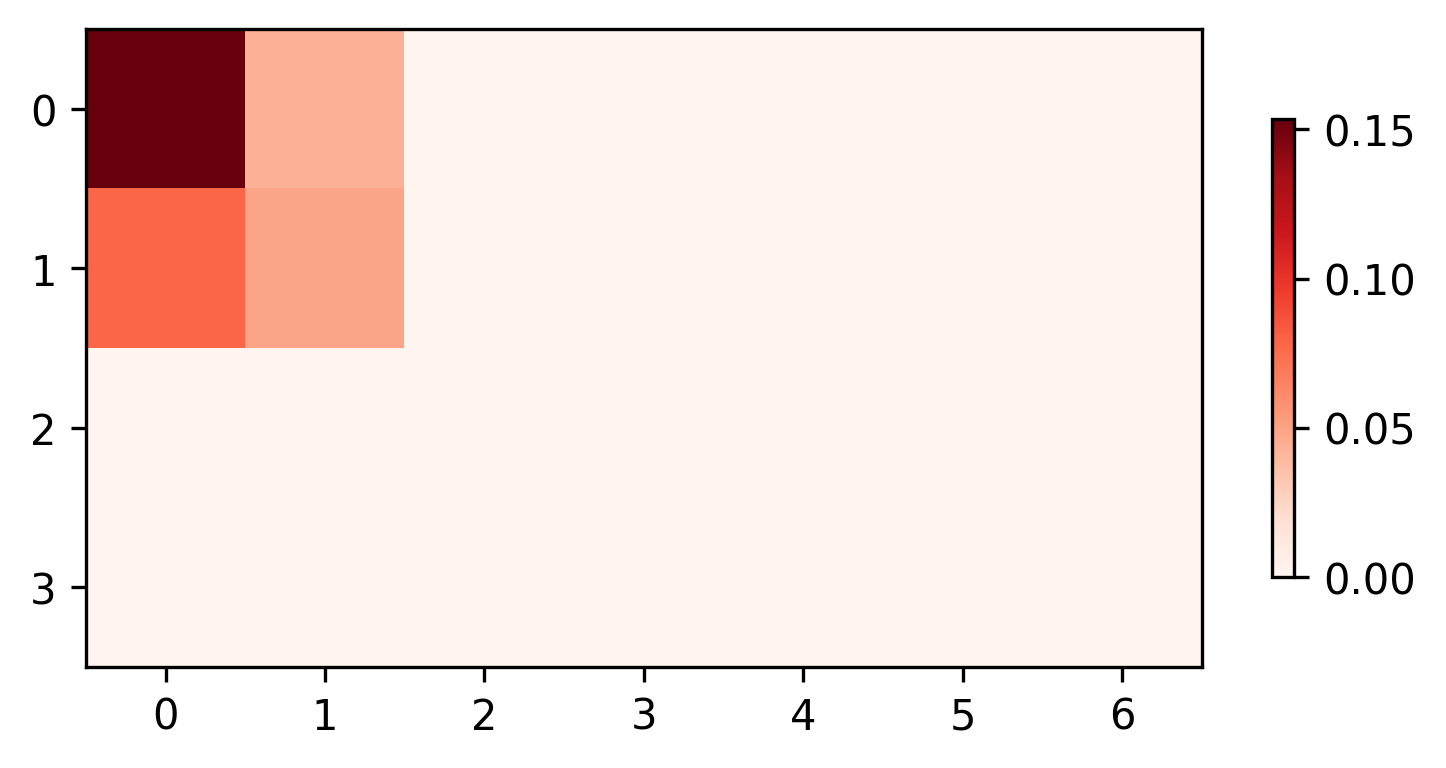}} \\
        \includegraphics[width=0.72\textwidth]{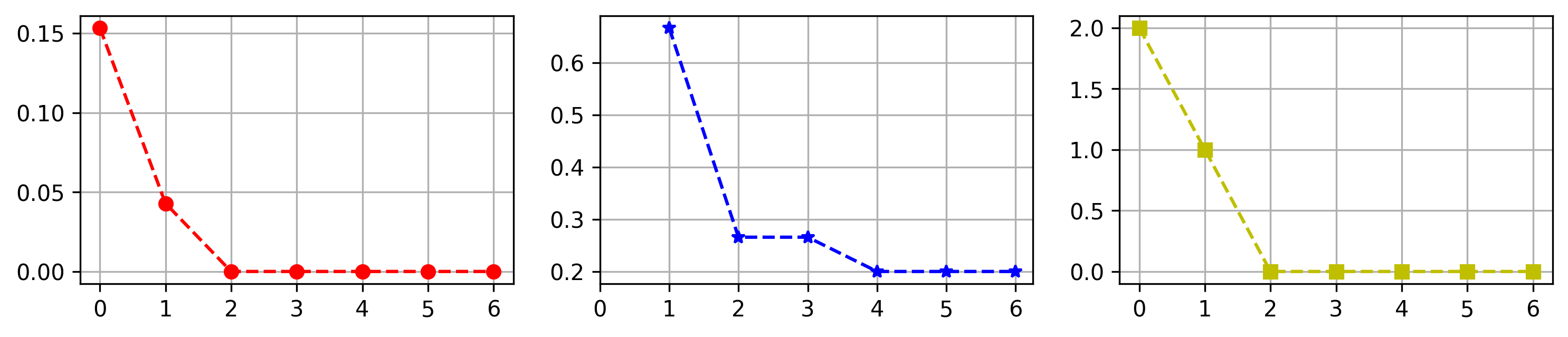} & \\
        \includegraphics[width=0.72\textwidth]{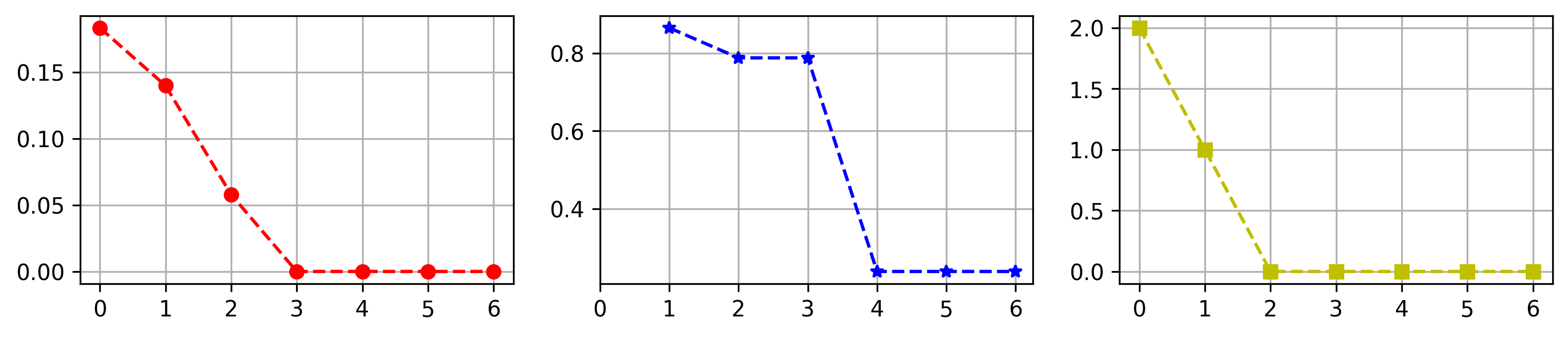} & 
        \multirow{2}{*}{\includegraphics[width=.25\textwidth]{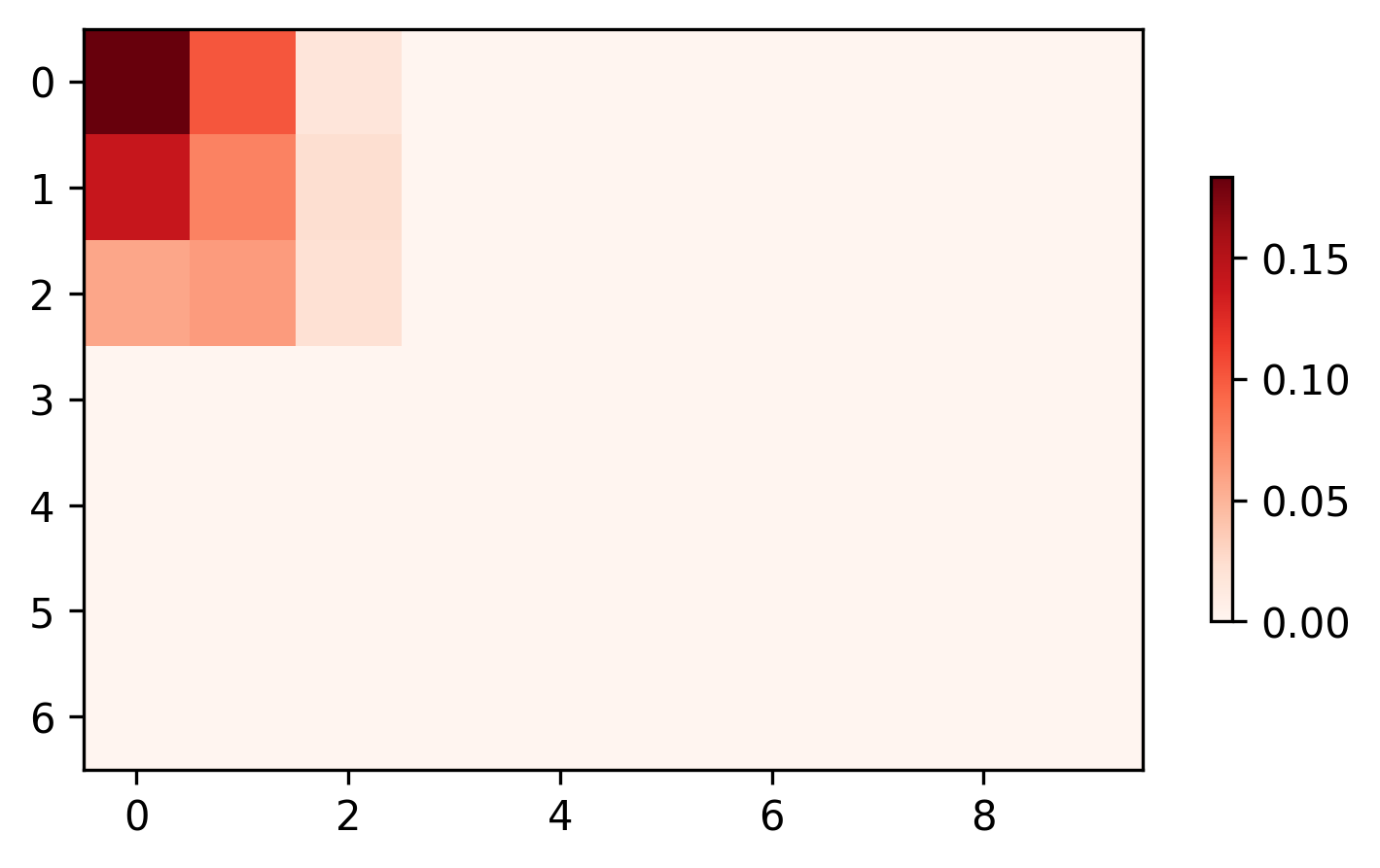}} \\
        \includegraphics[width=0.72\textwidth]{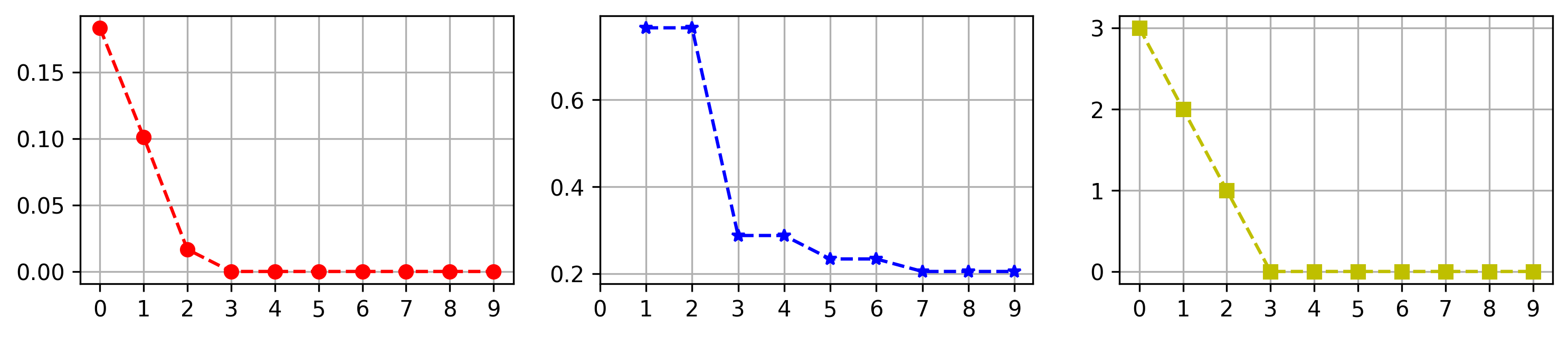} & \\
        \includegraphics[width=0.72\textwidth]{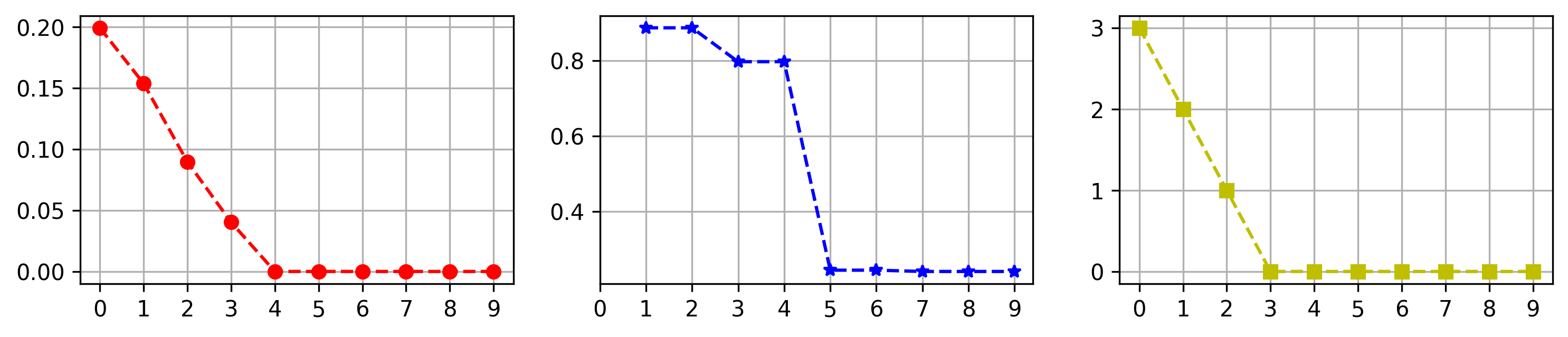} & 
        \multirow{2}{*}{\includegraphics[width=.25\textwidth]{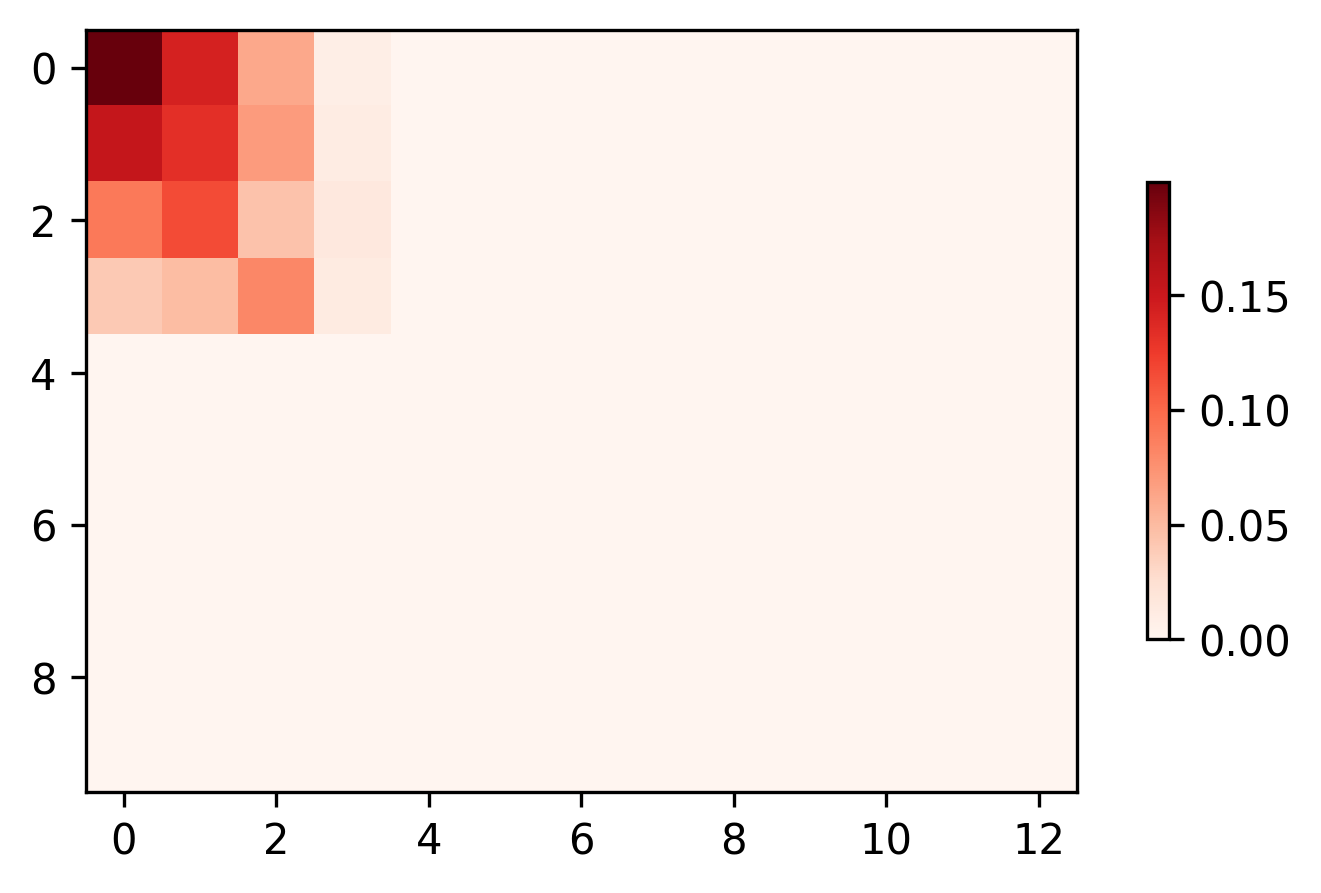}} \\
        \includegraphics[width=0.72\textwidth]{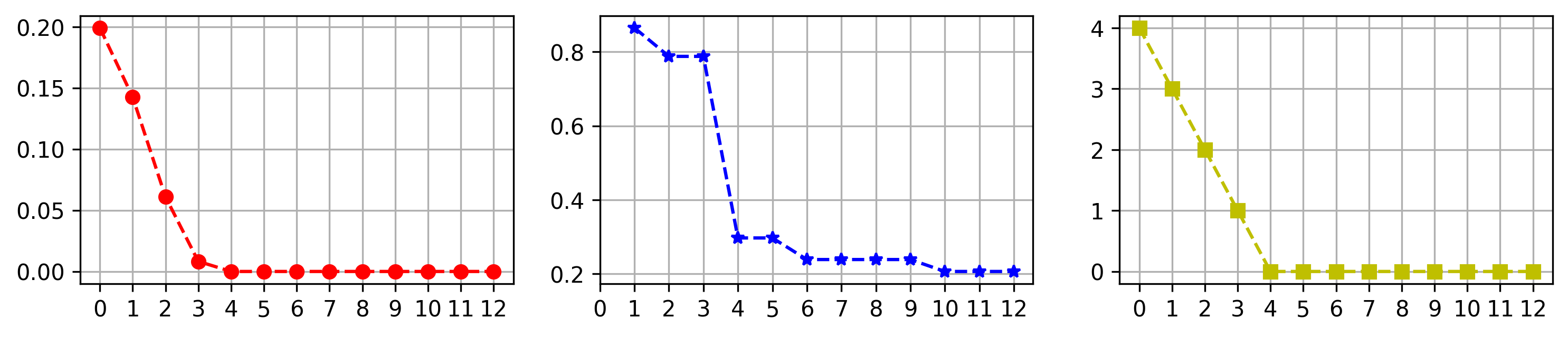} & 
    \end{tabular}
    \caption{Results from removing edges from the ensnarled ladder lattices with varying periods according to the edge linking centrality: (left) unbalance score, (middle-left) centrality of the newly removed edge, and (middle-right) the size of minimum cycle basis, where the x-axis is the number of edges removed (rows 1,3,5 from the smaller lattices; rows 2,4,6 from the larger lattices), and (right) unbalance score from removing edges both from the smaller lattice (y-axis) and the larger lattice (x-axis). (Row 1-2: period 1; Row 3-4: period 2; Row 5-6: period 3)}
    \label{fig:lad-ps-remv}
\end{figure}

The results from removing edges from one of the step lattices of period $1$ and period $3$ are shown in Fig.~\ref{fig:step-ps-remv} (top) and (bottom), respectively. 
\begin{figure}[htbp]
    \centering
    \begin{tabular}{c}
       \includegraphics[width=0.8\textwidth]{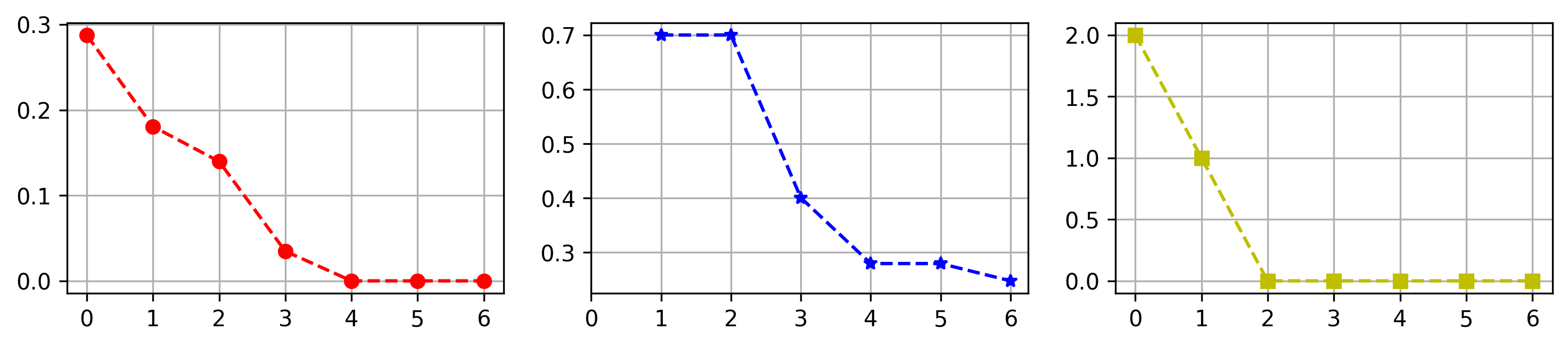} \\
       \includegraphics[width=0.8\textwidth]{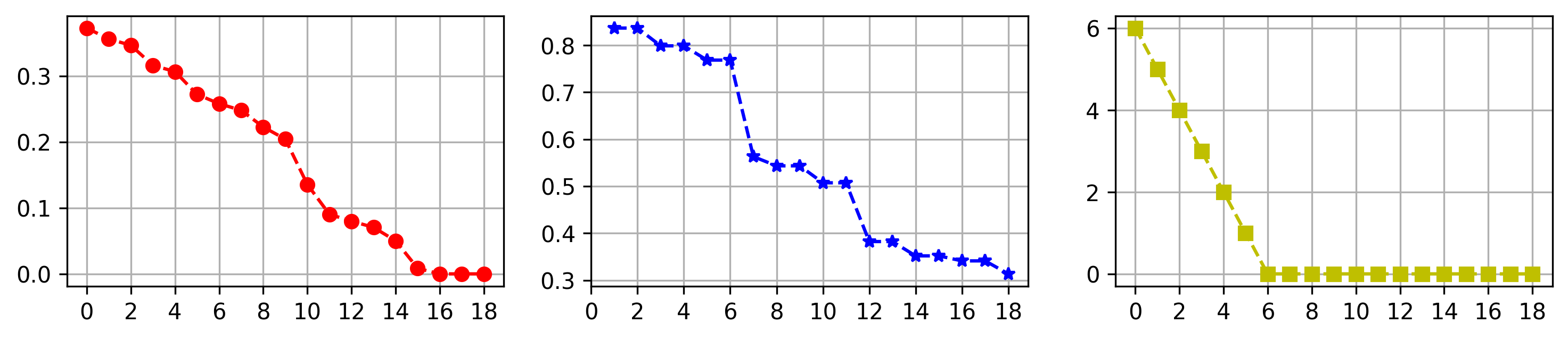}
    \end{tabular}
    \caption{Results from removing edges from the one of the step lattices in the ensnarled step lattice with periods 1 (top) and 3 (bottom) according to the edge linking centrality (cf. Fig.~\ref{fig:eg-step-p1}): unbalance score (left), centrality of the newly removed edge (middle), and the size of minimum cycle basis (right), where the x-axis is the number of edges removed.}
    \label{fig:step-ps-remv}
\end{figure}

\subsection{Brain vasculature analysis}
\label{app:brain_vesculature_analysis}

\paragraph{Null models.}
In our implementations, we impose the hard constraint that the sampled networks have $m$ edges. 
Specifically, for the SER model, we uniformly randomly select a pair of nodes to place an edge until we have $m$ edges.
For the SRG model, we randomly select more than $10m$ but fewer than $1/5\binom{n}{2}$ pairs of nodes from all possible combinations, and then select the top $m$ pairs with the smallest distances between the nodes.
With the networks from the null model, we randomly shuffle the actual radius of the mouse brain vasculature and assign them to the edges in the null model. 
This is motivated by the observation that there is no clear trend and relationship between the distance of the blood vessel and the radius; see the scatter plots in Fig.~\ref{fig:vessel-all-dist-radius}, where the Pearson correlation is consistently small. 
This completes the preparation for a random sample from a null model for our analysis.
\add{All null-model estimates are computed from $10$ independent realizations.}

\begin{figure}[htbp]
    \centering
    \hspace*{-1em}
    \begin{tabular}{ccc}        \includegraphics[width=.33\textwidth]{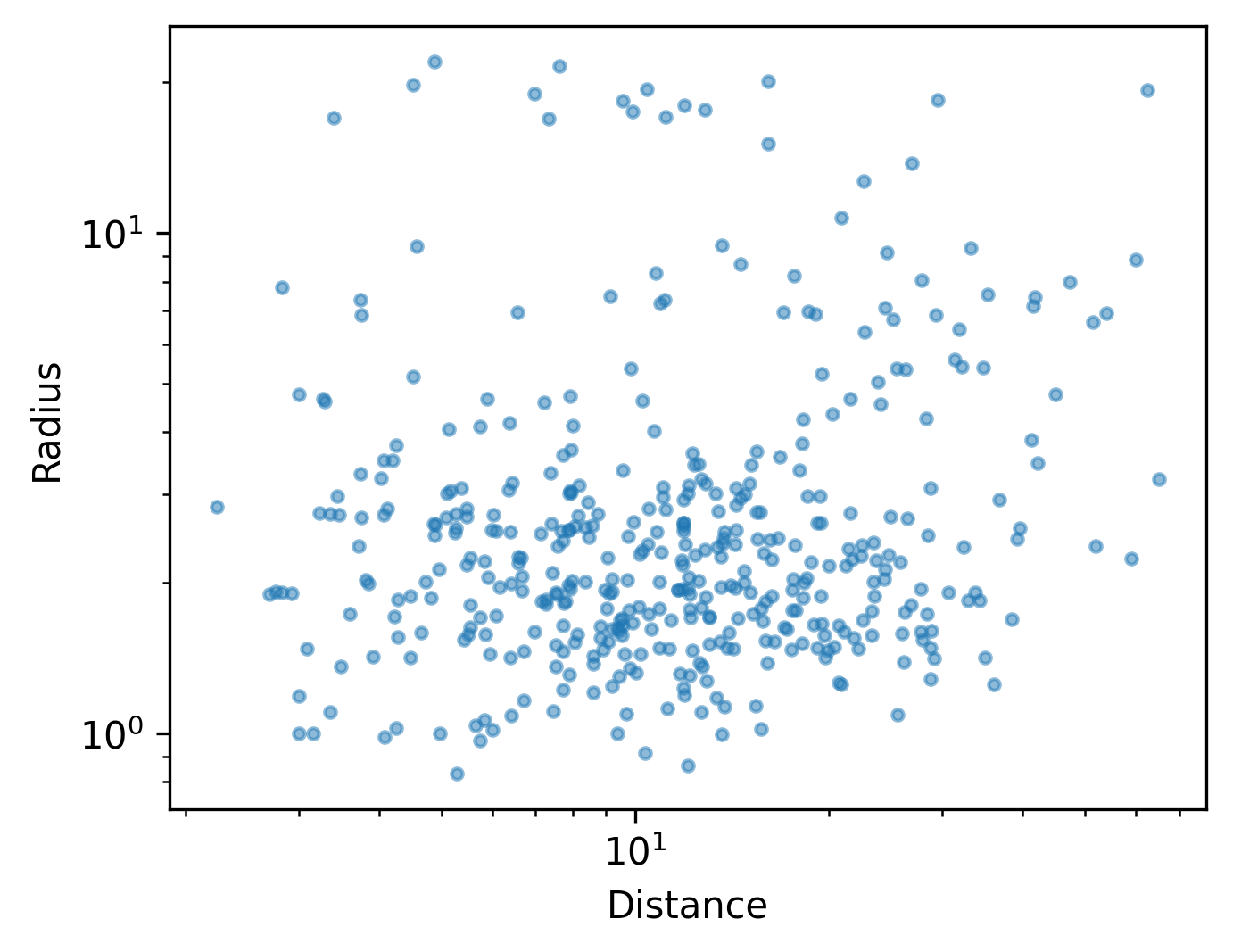} & \includegraphics[width=.33\textwidth]{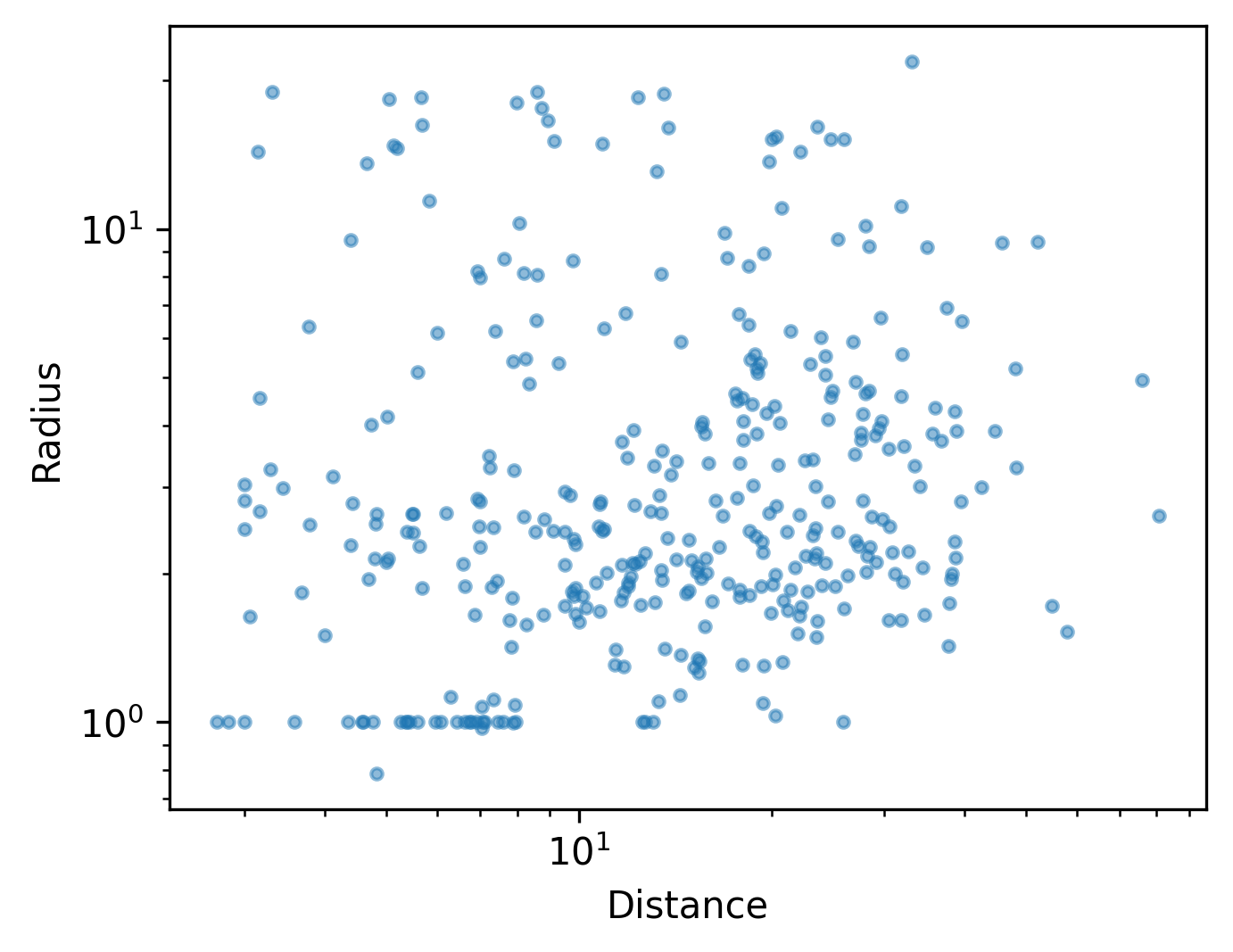} & \includegraphics[width=.33\textwidth]{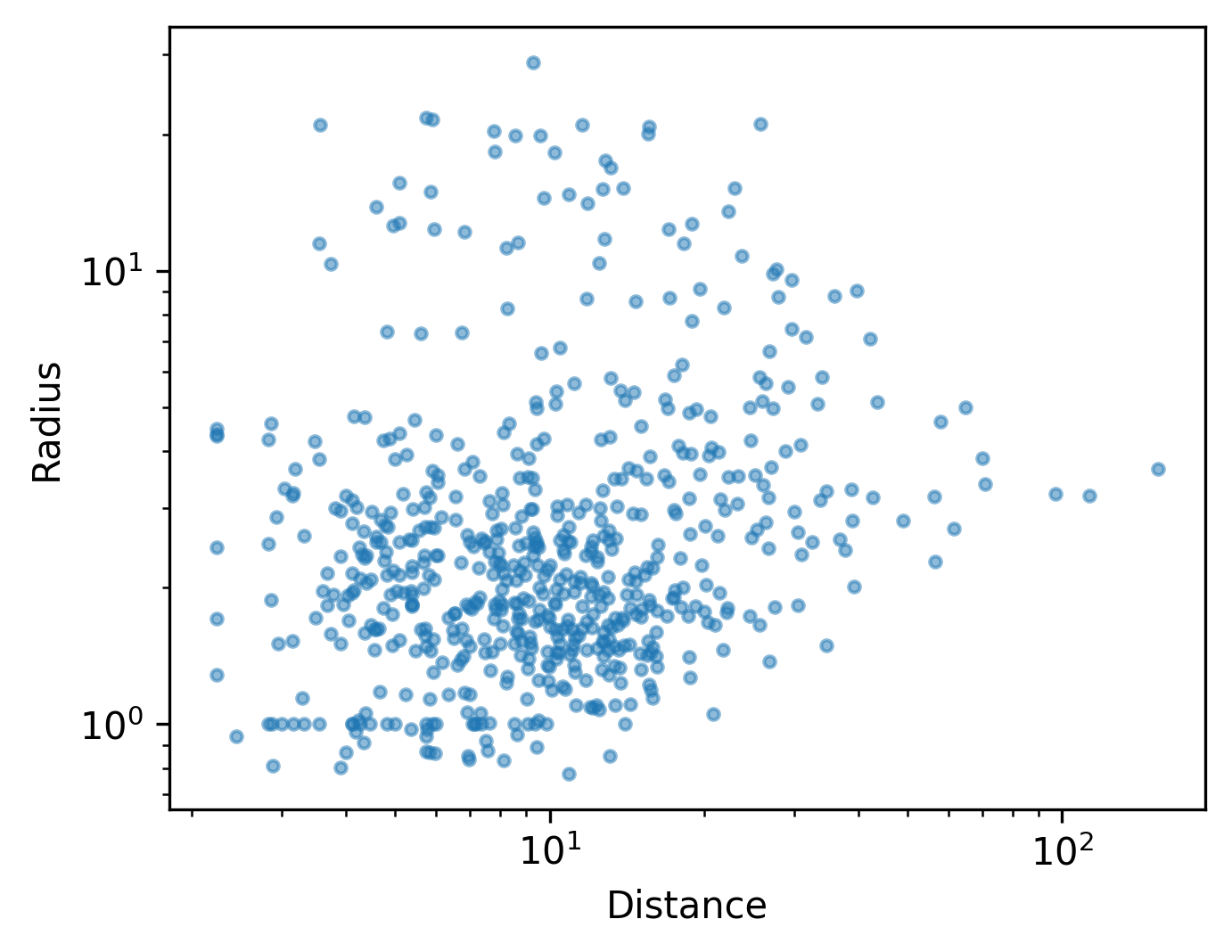}
    \end{tabular}
    \caption{Scatter plots of the distance between two endpoints of the edges and the radius of the vessels (Zones I, II, II from left to right, with Pearson correlation coefficients $0.139$, $0.014$, and $0.085$, respectively).}
    \label{fig:vessel-all-dist-radius}
\end{figure}

\begin{figure}[htbp]
    \centering
    \begin{tabular}{ccc}
        \includegraphics[width=.3\textwidth]{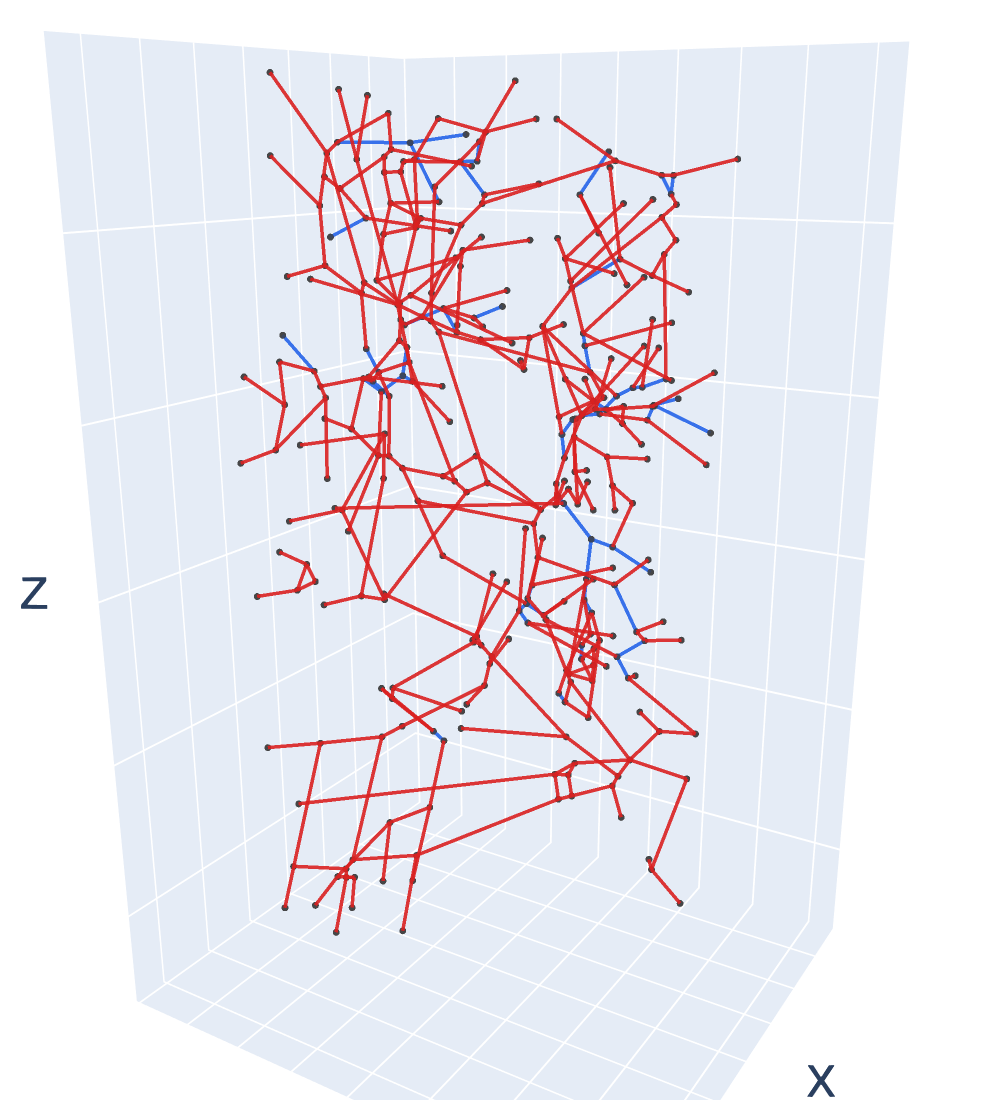} & \includegraphics[width=.3\textwidth]{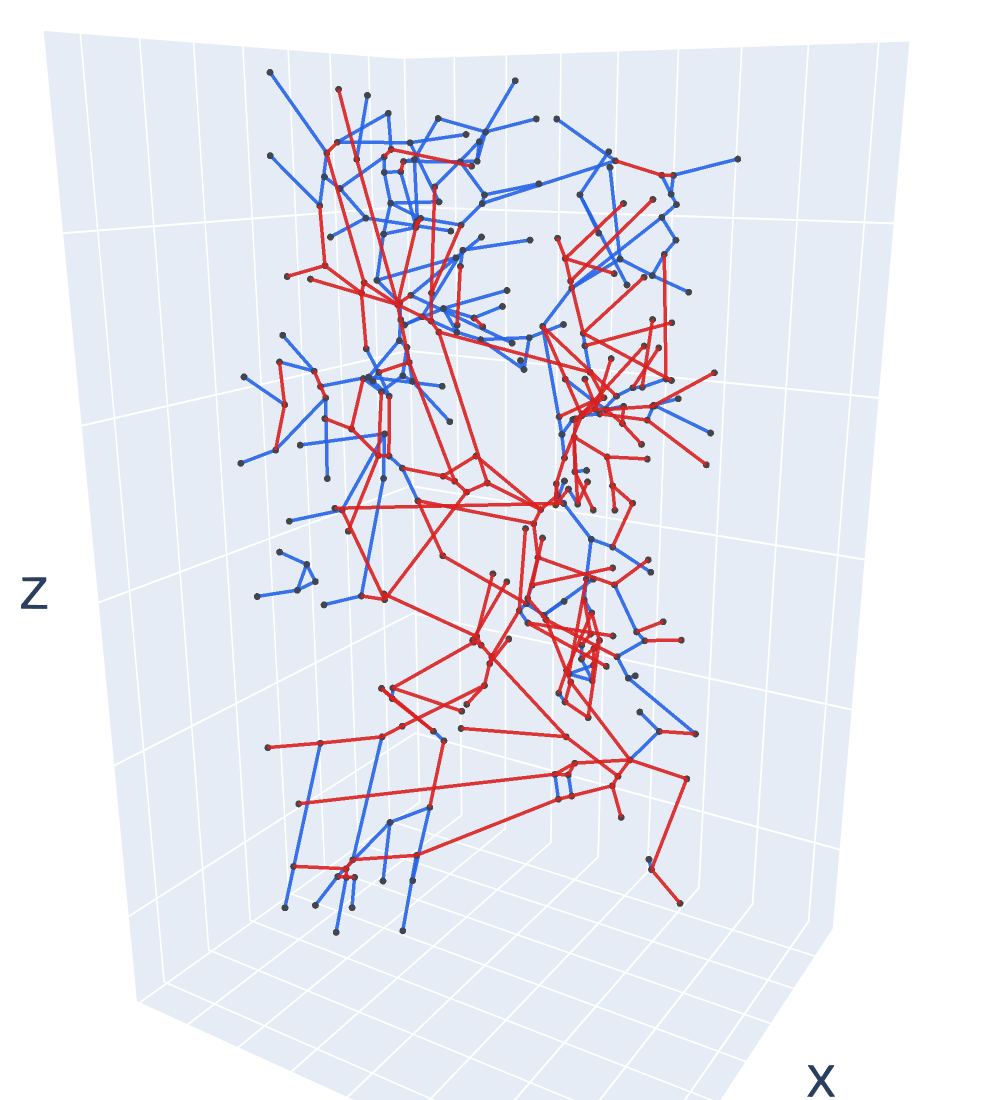} & \includegraphics[width=.3\textwidth]{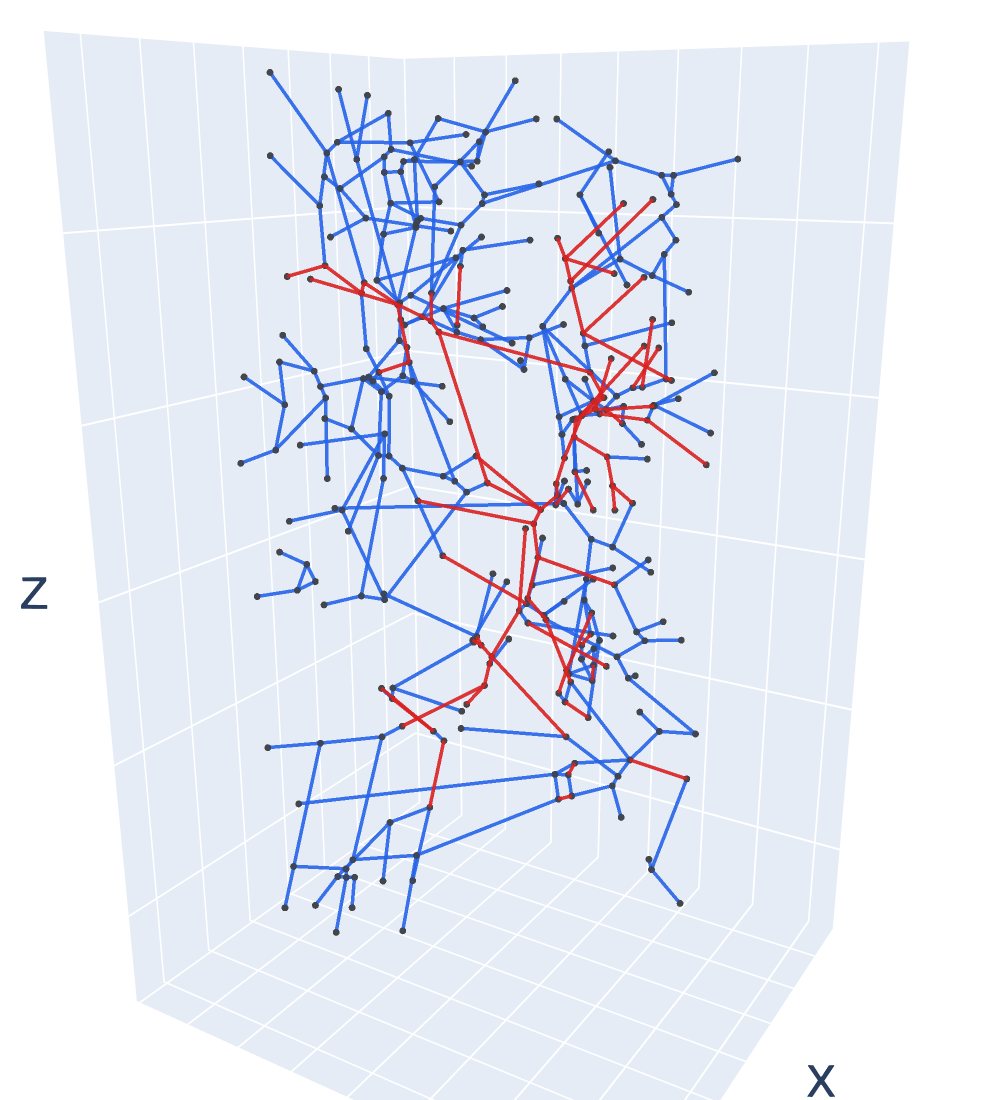} \\
        \includegraphics[width=.3\textwidth]{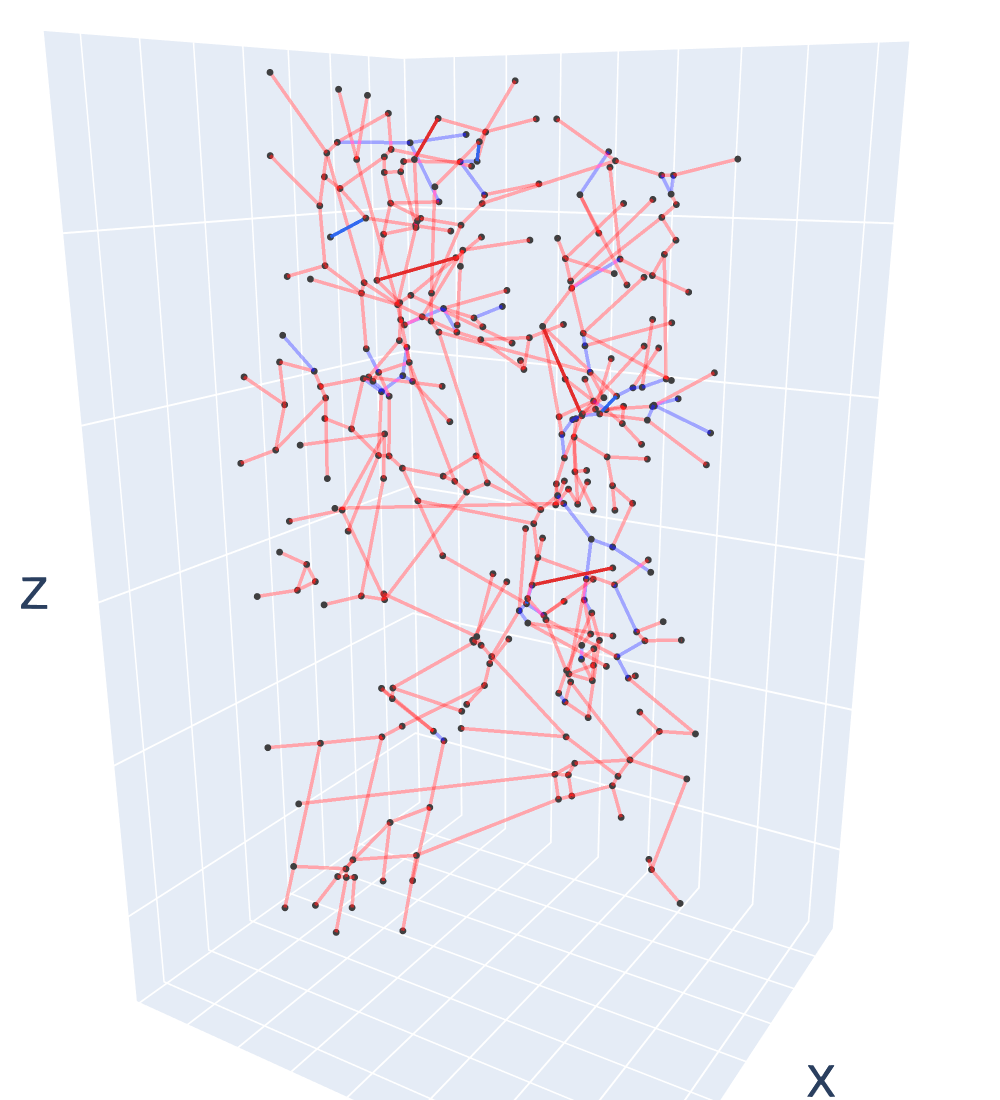} & \includegraphics[width=.3\textwidth]{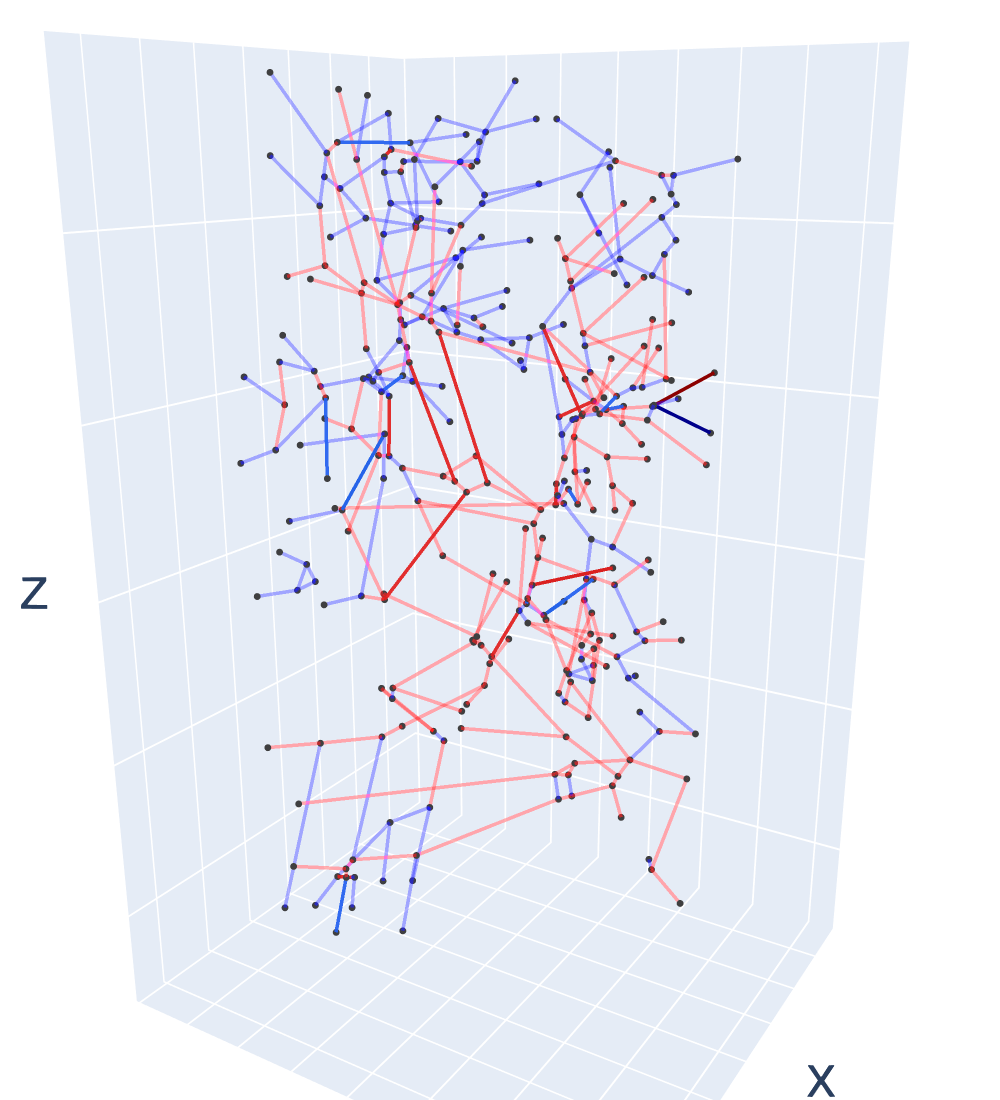} & \includegraphics[width=.3\textwidth]{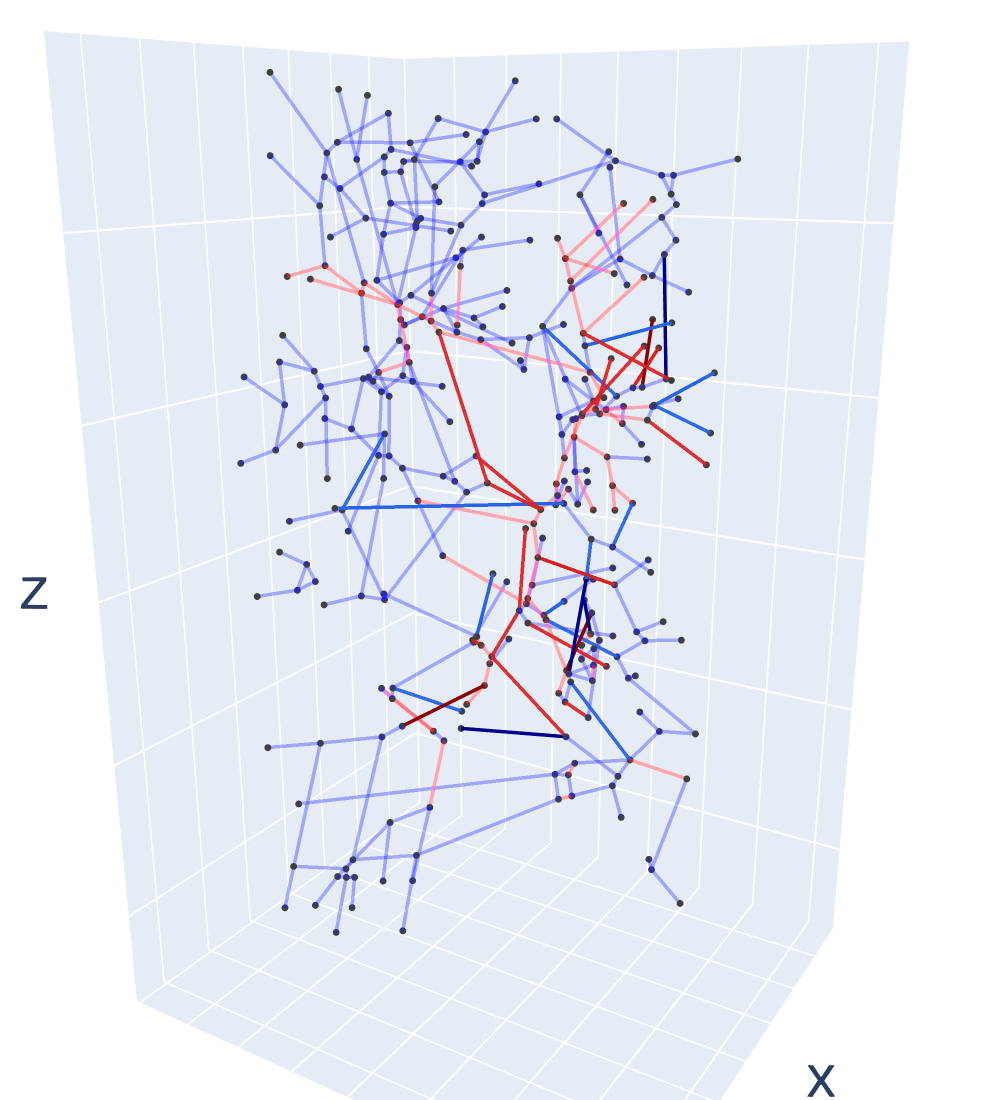} \\
        \includegraphics[width=.3\textwidth]{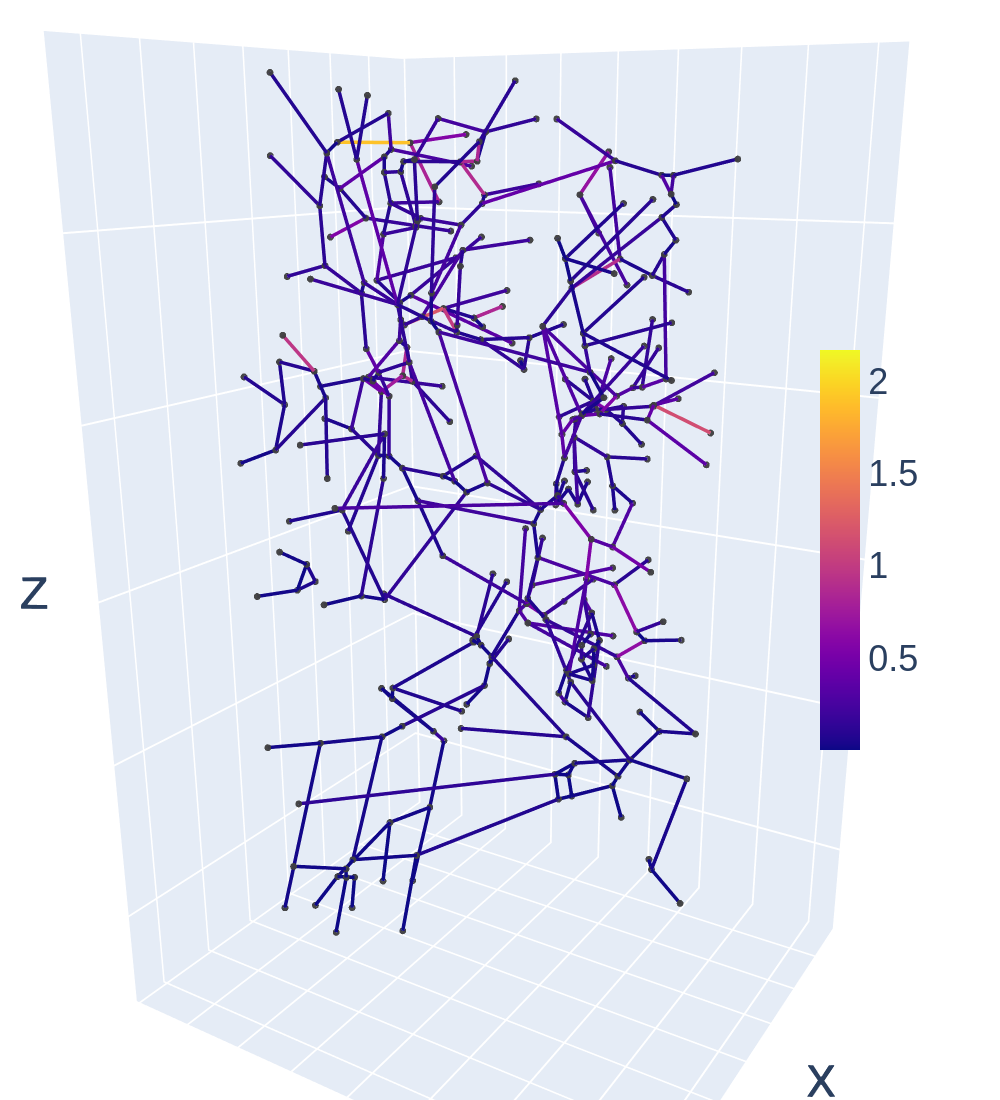} & \includegraphics[width=.3\textwidth]{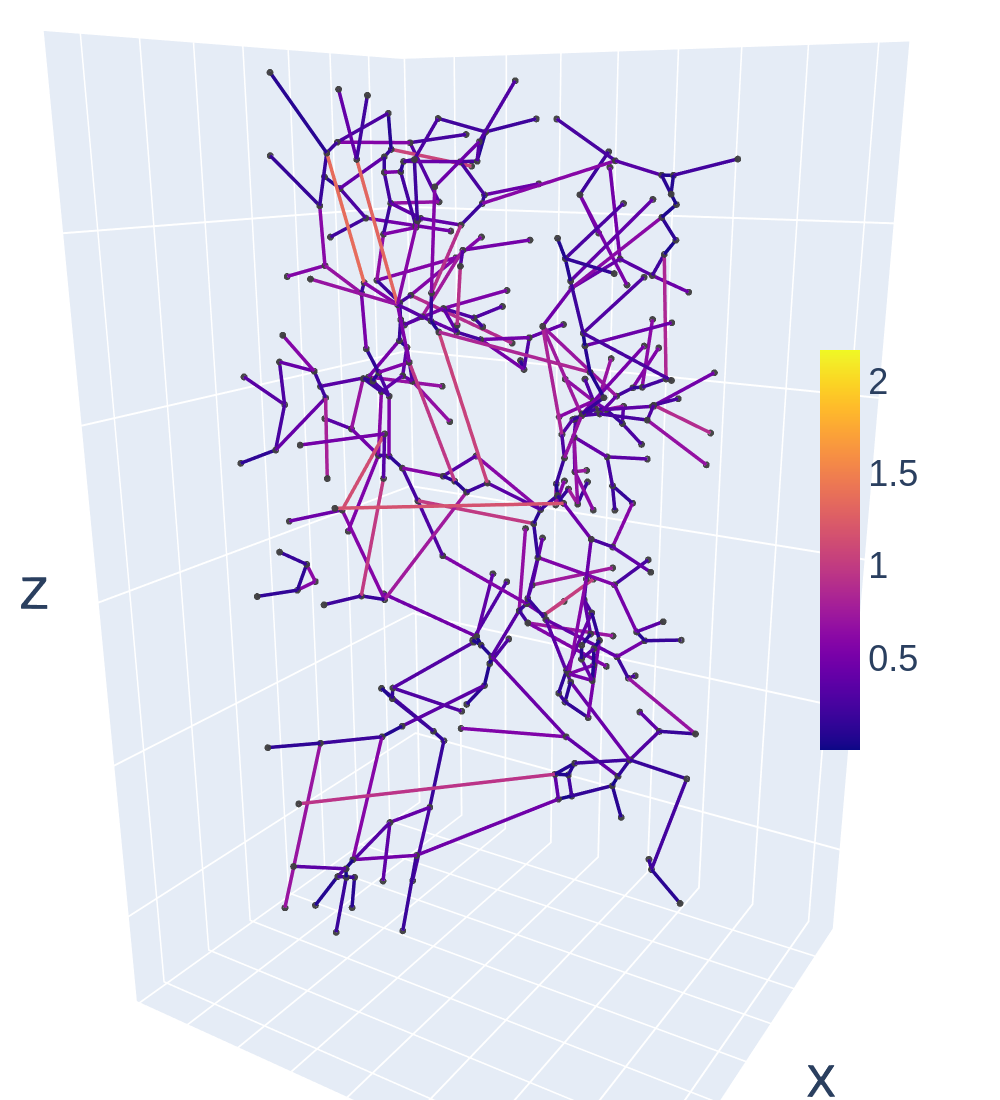} & \includegraphics[width=.3\textwidth]{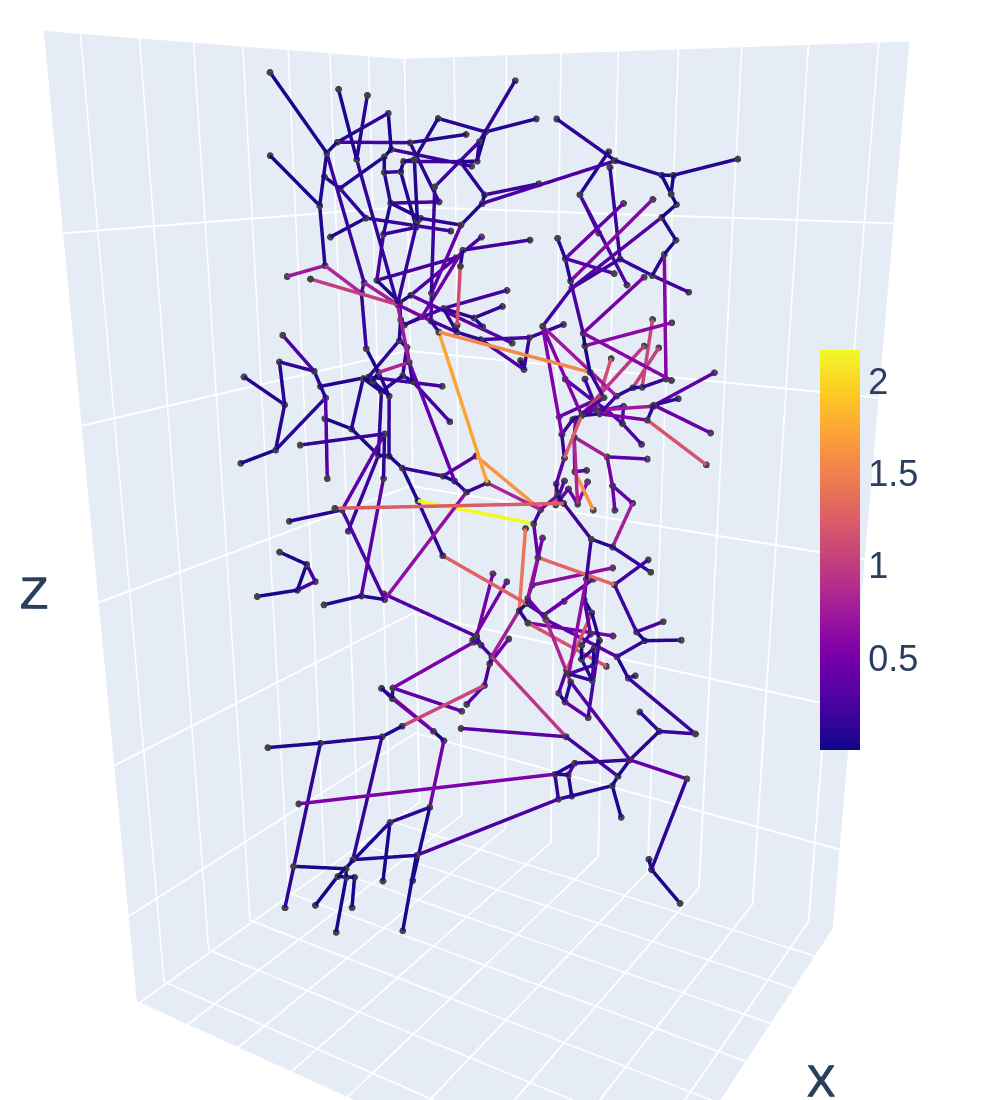}
    \end{tabular}
    \caption{Results of zone II from a heterogeneous region including the mouse midbrain: (top) the spatial networks of the blood vessels, (middle row) the strong negative linking between the edges after an optimal orientation has been assigned (dark red/blue if incomplete linking $<-0.1$, slightly lighter red/blue if $<-0.05$, and the lightest otherwise), and (bottom) the linking centrality of the edges; the edges are separated by three different values for the threshold $c$ on the radius, where (left) $c=1.5\,\mu m$, (middle column) $c=2\,\mu m$ and (right) $c=5\,\mu m$. \add{The displayed box spans $x\in [980, 1130]$, $y\in [2760, 2930]$ and $z\in [920,1200]$.}}
    \label{fig:vessel-zone2} 
\end{figure}

\begin{figure}[htbp]
    \centering
    \begin{tabular}{ccc}
        \includegraphics[width=.33\textwidth]{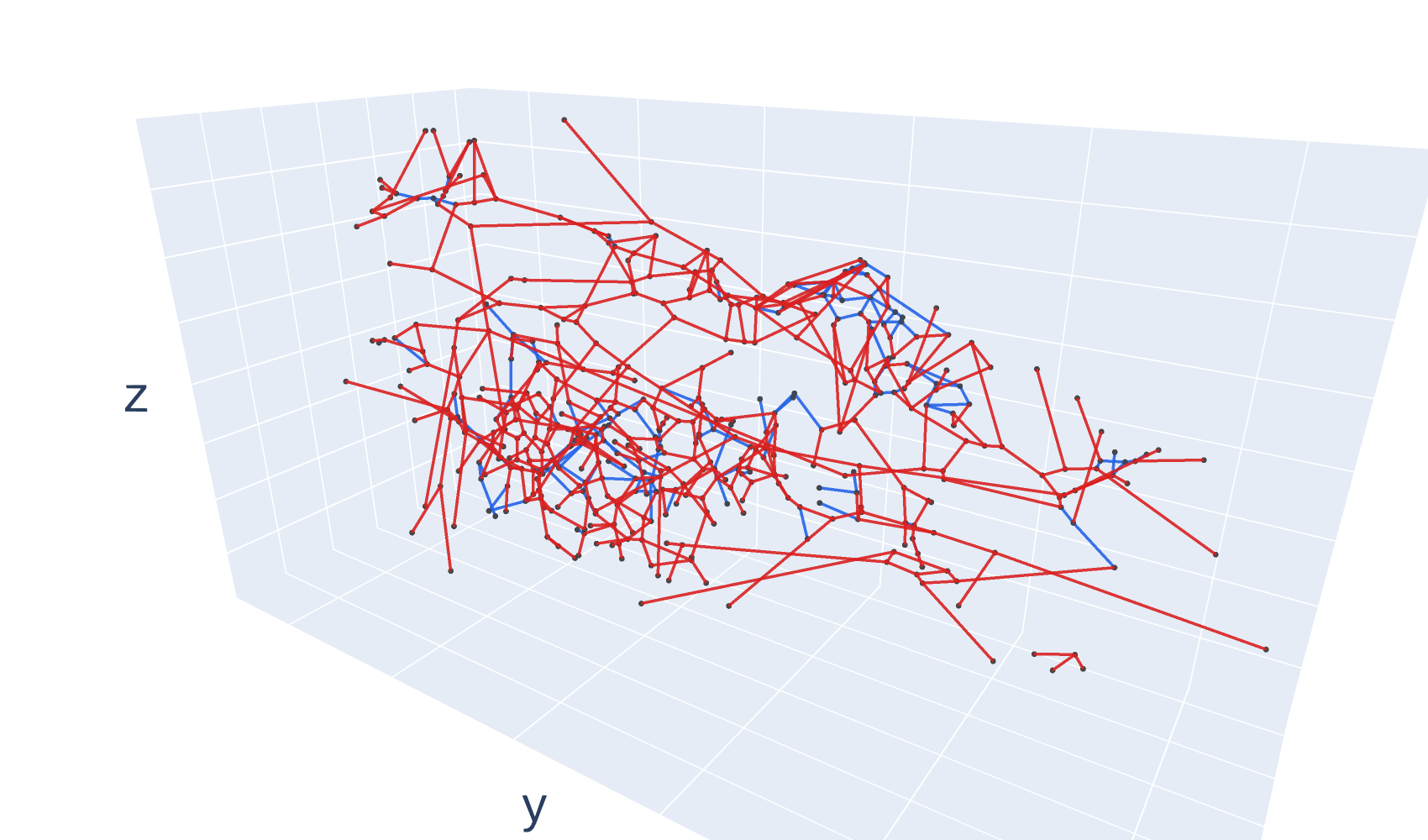} & \includegraphics[width=.33\textwidth]{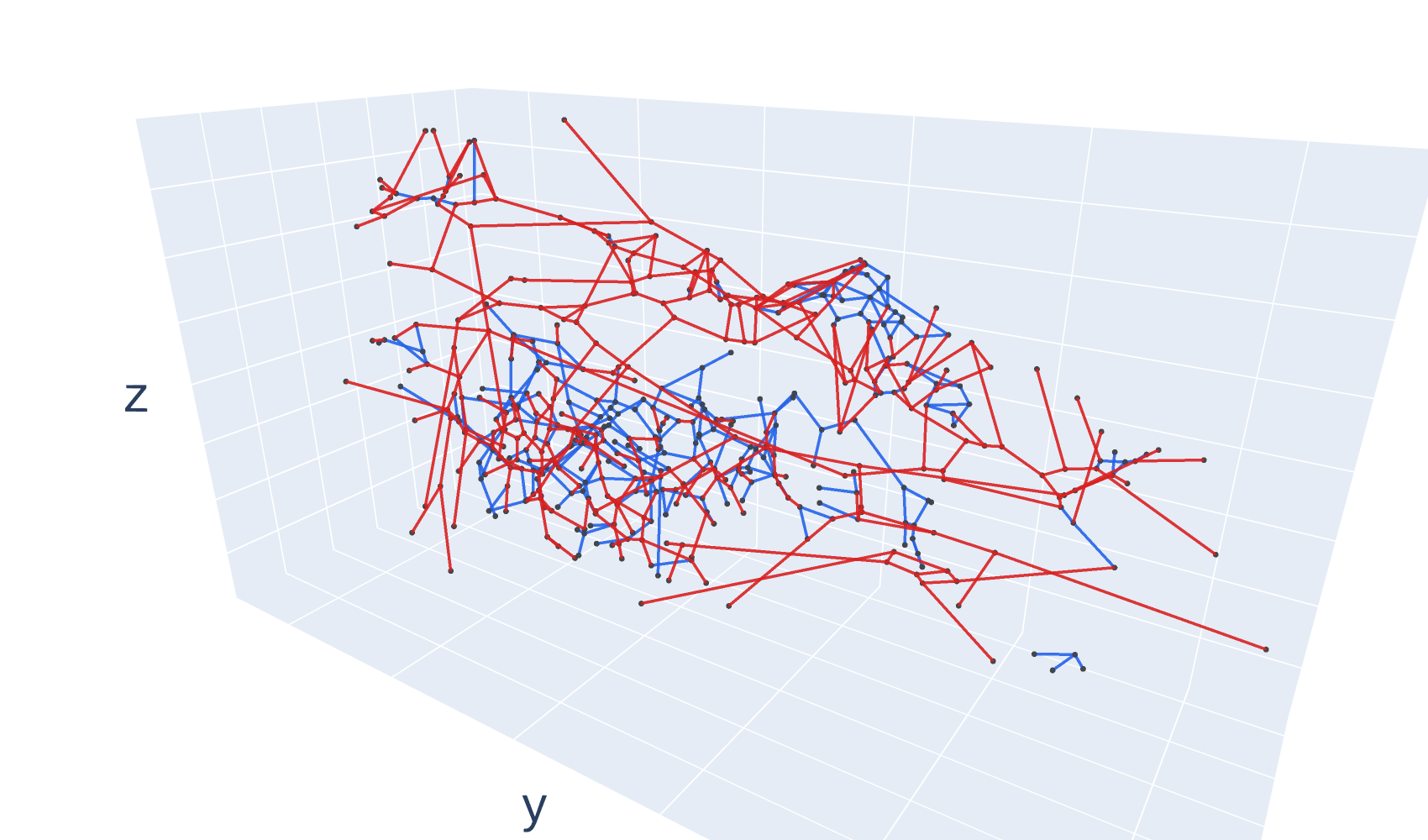} & \includegraphics[width=.33\textwidth]{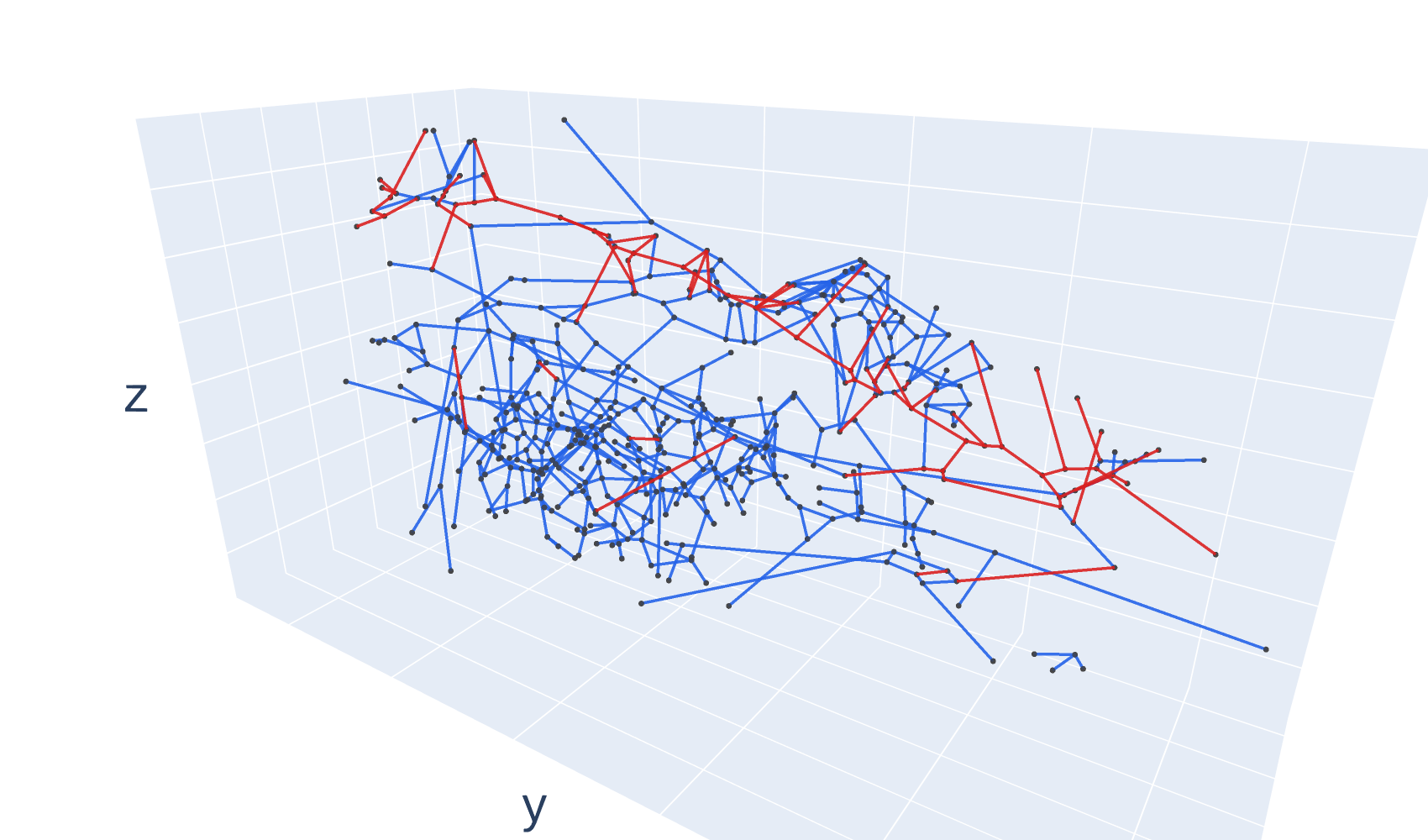} \\
        \includegraphics[width=.33\textwidth]{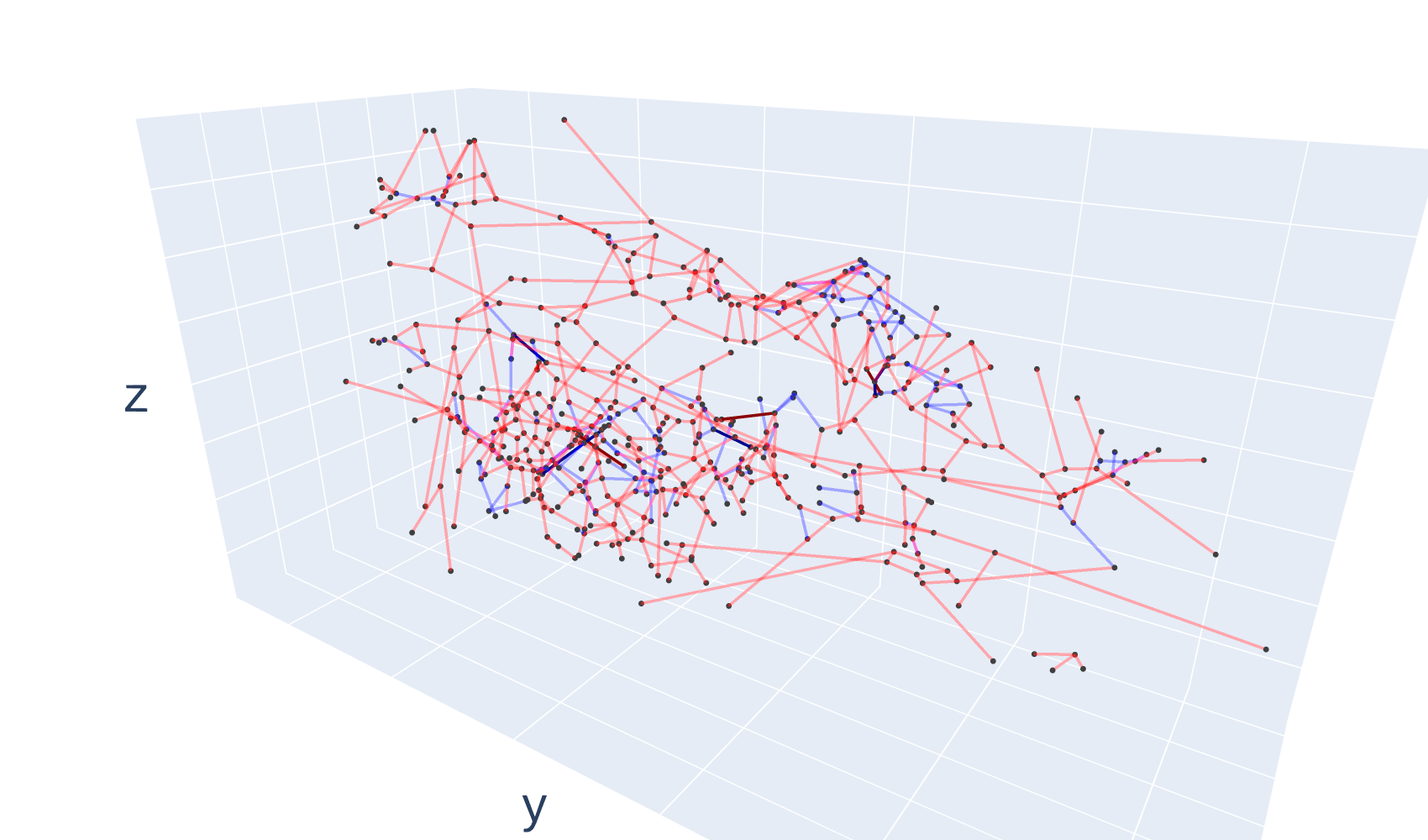} & \includegraphics[width=.33\textwidth]{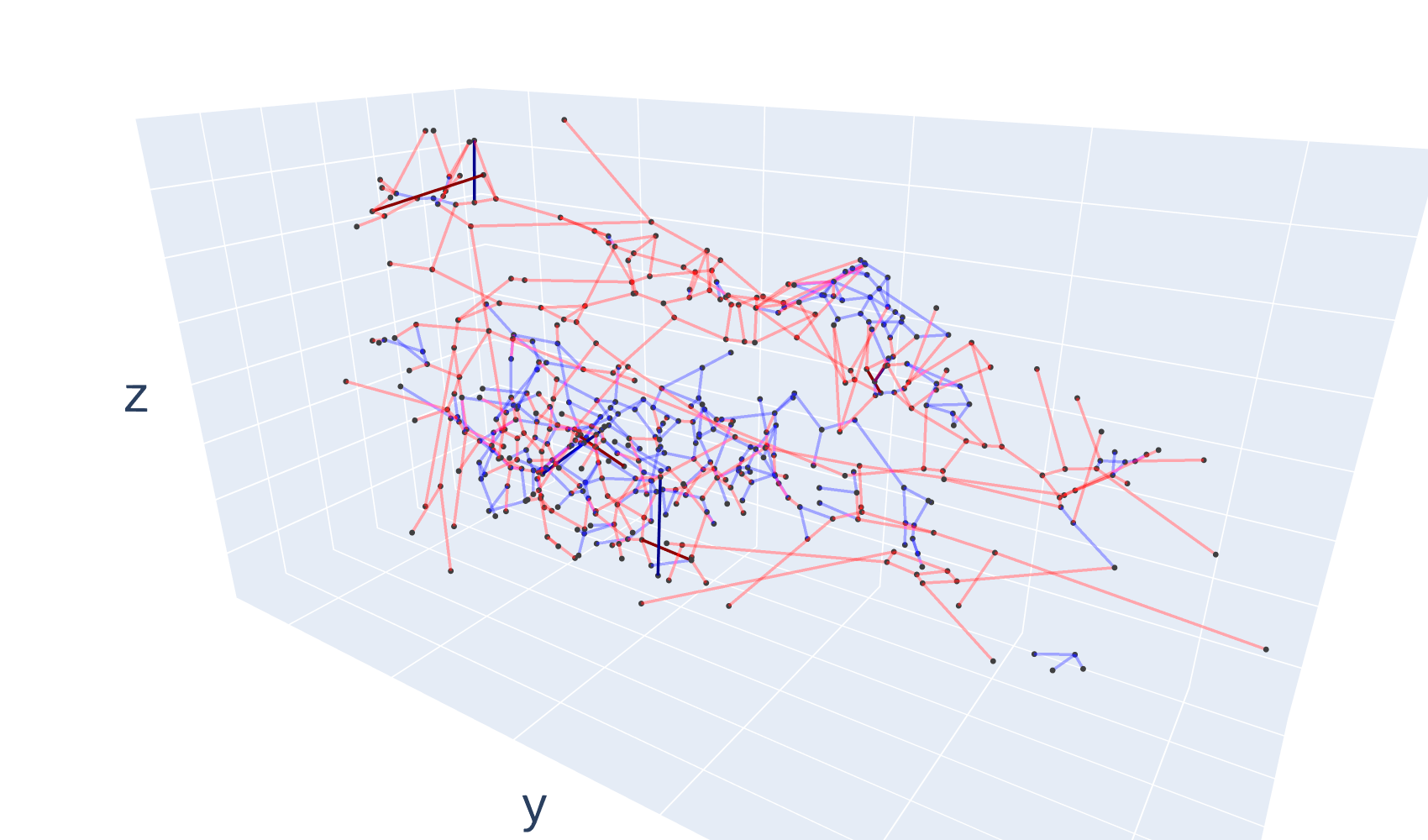} & \includegraphics[width=.33\textwidth]{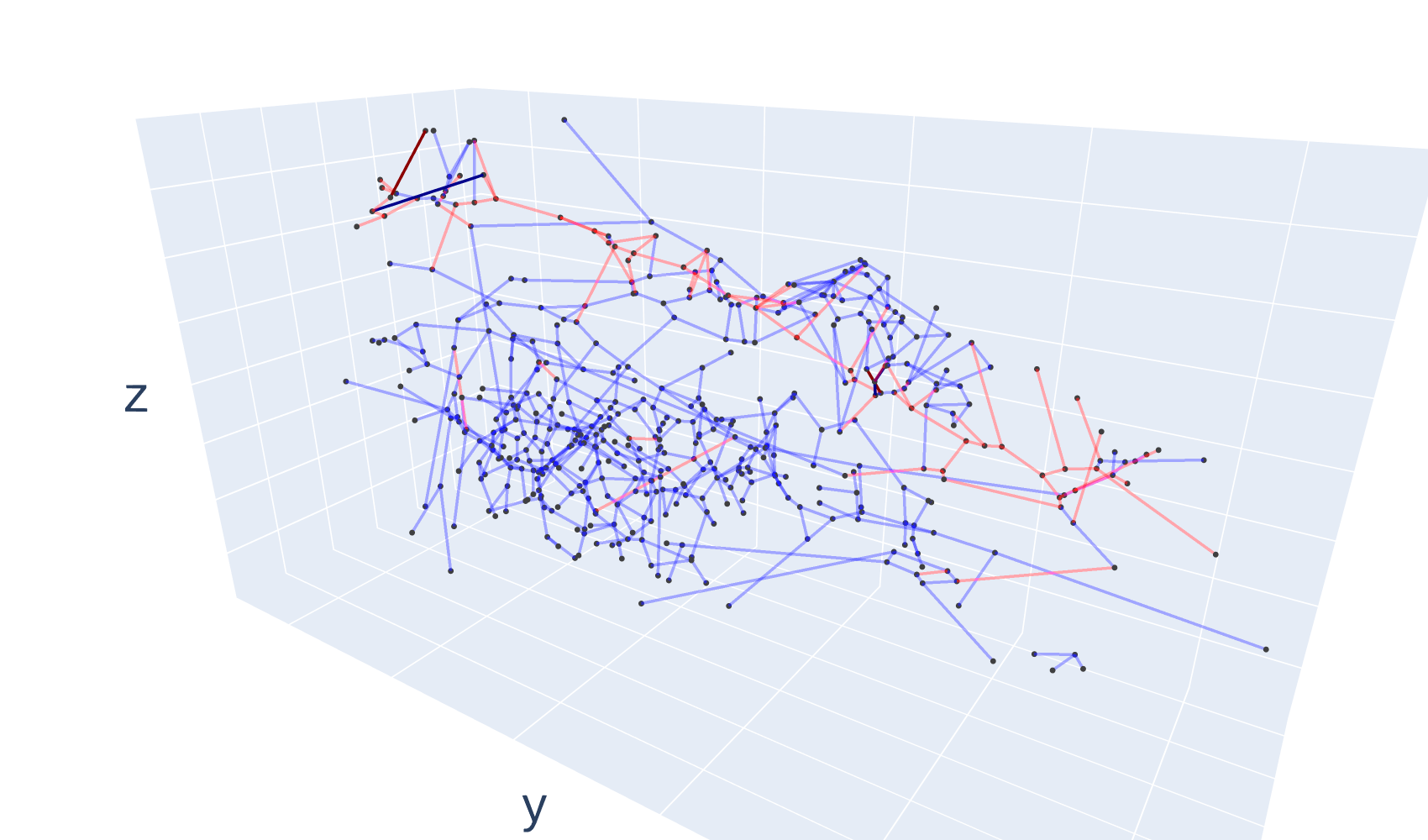} \\
        \includegraphics[width=.33\textwidth]{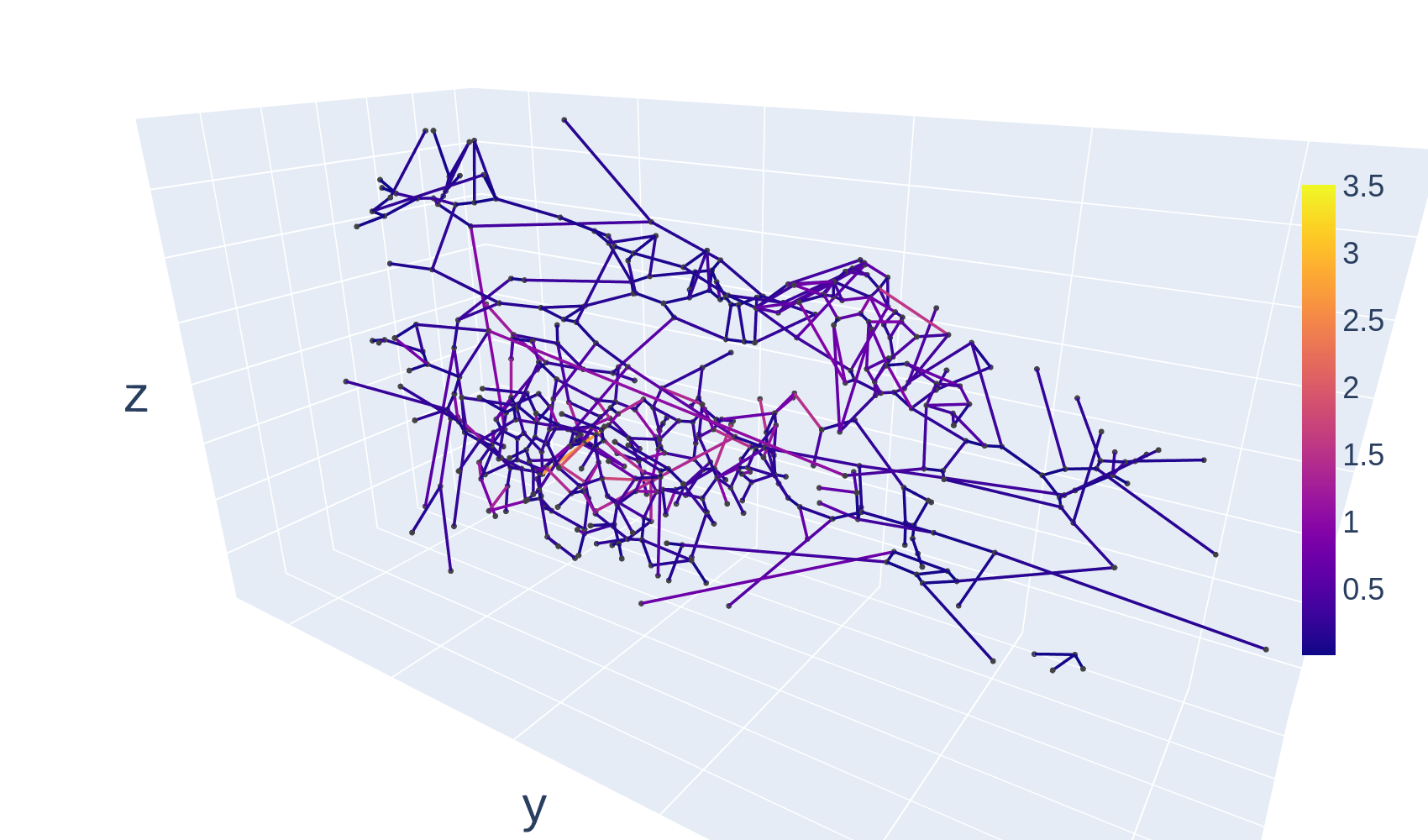} & \includegraphics[width=.33\textwidth]{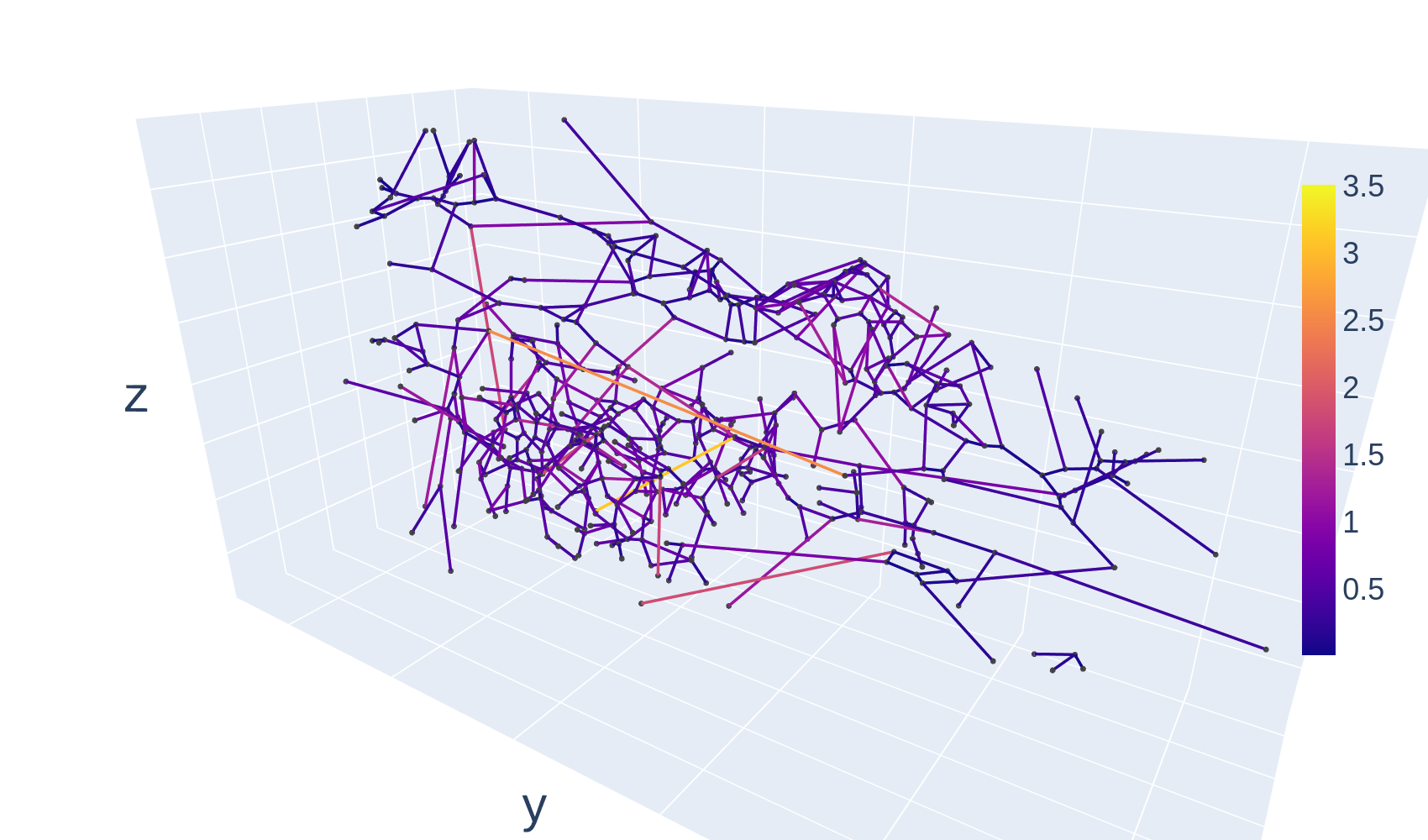} & \includegraphics[width=.33\textwidth]{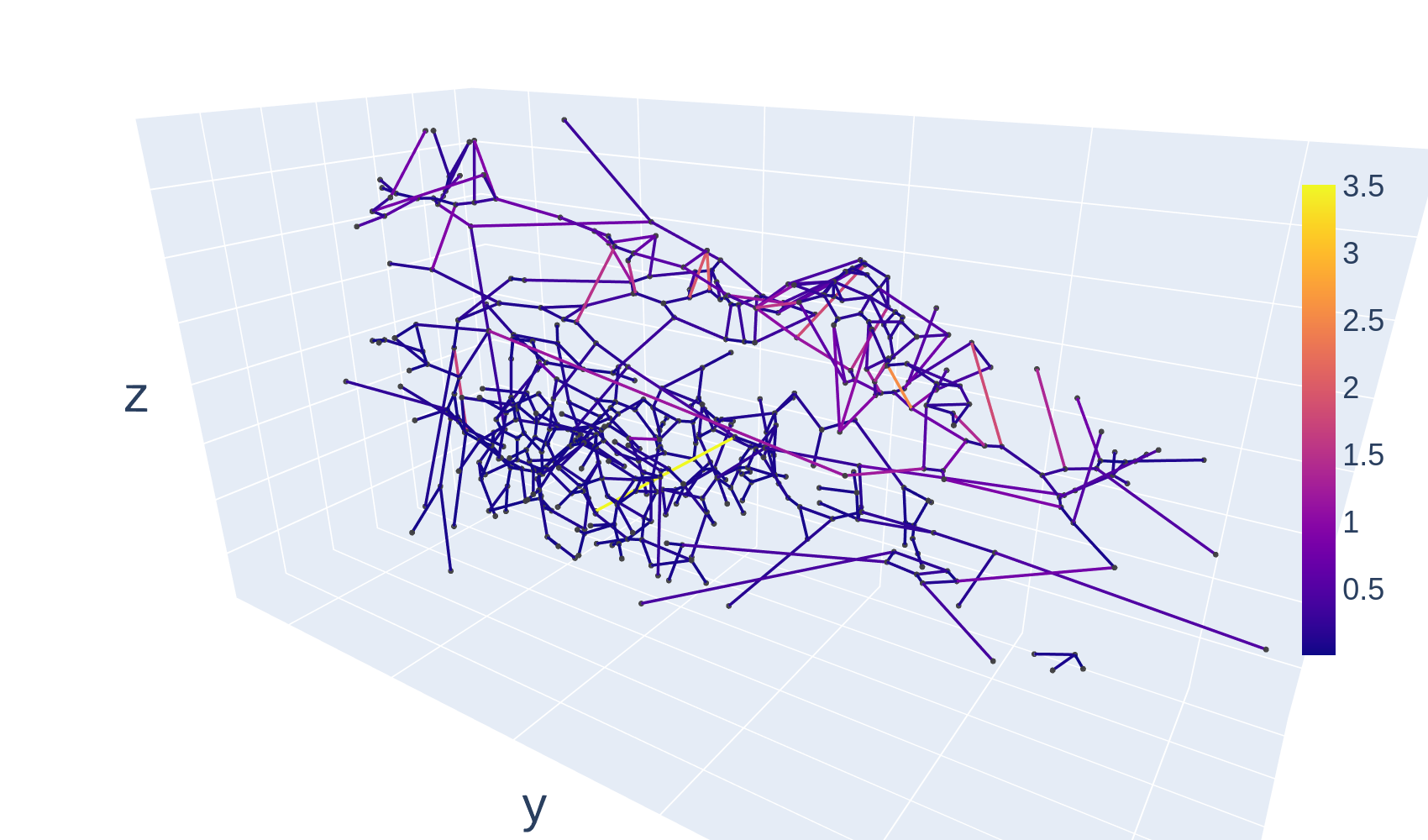}
    \end{tabular}
    \caption{Results from the top right corner of the mouse brain: (top) the spatial networks of the brain vessels, (middle row) the strong negative linking between the edges after an optimal orientation has been assigned (dark red/blue if incomplete linking $<-0.1$, slightly lighter red/blue if $<-0.05$, and the lightest otherwise), and (bottom) the linking centrality of the edges; we separate two different sets of edges by their radius and different cut values, $c$, where (left) $c=1.5\,\mu m$, (middle column) $c=2\,\mu m$ and (right) $c=5\,\mu m$. \add{The displayed box spans $x\in [1670, 1800]$, $y\in [3970, 4280]$ and $z\in [1060,1760]$.}}
    \label{fig:vessel-zone3} 
\end{figure}

\paragraph{Zone I.} 
Zone I is from the cerebellum, containing $353$ nodes, of which $167$ are in region Lobules IV, $60$ correspond to fiber traits, $43$ are in Pons Sensory, $39$ are in Hindbrain, and $36$ are in Cerebellum.  
The spatial network in this zone is characterized by a main branch of high-radius edges going through the top left corner to the bottom right corner; see Fig.~\ref{fig:real-all}B. 
As the radius threshold values increase, the number of intrinsic negative edge-pairs with strong signals first becomes larger, and then becomes smaller; see Fig.~\ref{fig:real-all}D. We note that in the middle column, where the numbers of edges in the two sets are close to each other, the corresponding signed network has the largest number of intrinsic strong negative edge-pairs among the three cases. This can also be implied from the unbalance score $\lambda_{\min} = 0.43$ versus $0.41$ 
for the left column and $0.39$ 
for the right column. 
Zooming in on the zone, we can identify the bottom-right region with high complexity, because there are consistently intrinsic strong negative edge-pairs, \textit{i.e.}, where the main complexity originates, and the edges have consistently high centrality values, as the radius threshold varies; see Fig.~\ref{fig:real-all}E. 

\paragraph{Zone II.} 
Zone II is from a heterogeneous region including the midbrain, with overall $310$ nodes of which $87$ are root, $83$ are in Basal Ganglia, $41$ correspond to fiber traits, $37$ are in Midbrain, $29$ are in Hippocampal Formation, $17$ are in Midbrain Motor-related.    
The spatial network in Zone II also exhibits more heterogeneous patterns, characterized by two main branches of high-radius edges, joining at the middle and going in different directions upwards; see Fig.~\ref{fig:vessel-zone2} (top). 
In this zone, the characteristics indicate qualitative differences from the previous case: after reaching the point where the number of edges in the two sets is approximately the same, further increasing the cut value can lead to more complexity; see Fig.~\ref{fig:vessel-zone2} (middle row) where there are increasing number of intrinsic negative edge-pairs with strong signals as the cut value rises. This can also be implied from the unbalance score $\lambda_{\min} = 0.398, 0.404, 0.408$ 
(from left to right). Therefore, there is a relatively larger range of radius threshold values corresponding to high complexity in Zone II than in Zone I. 
With \name, we can further identify the middle region where the two main branches join and diverge as high complexity, with consistently intrinsic strong negative edge-pairs and edges of consistently high centrality as the radius threshold value varies; see Fig.~\ref{fig:vessel-zone2} (bottom).

\paragraph{Zone III.} 
Zone III is from the medulla, featured by $465$ nodes of which $272$ are in region Medulla Motor-related, $143$ are in Medulla Behavorial State-related, and $50$ are in Medulla.  
The spatial network in Zone III is characterized by not only a main branch of high-radius edges going from left to right, but also two relatively loosely connected parts, where the top part contains the main branch; see Fig.~\ref{fig:vessel-zone3} (top). 
In this zone, we further observe qualitative differences in the characteristics from the previous two cases: as the radius threshold increases from a small value corresponding to many more high-radius edges than the low-radius ones, the complexity may not increase as in Zones I and II, but maintain a similar level and then decreases; see Fig.~\ref{fig:vessel-zone3} (middle row).   
When the threshold is small, we observe nontrivial linking between the two sets that occur in both top and bottom parts; see Fig.~\ref{fig:vessel-zone3} (middle left). As the threshold value increases, the number of intrinsic negative edges with strong signals first maintains and then reduces. This can also be confirmed by the unbalance score $\lambda_{\min} = 0.39, 0.43$ 
for the left and middle columns and $0.35$ for the right column. 
Furthermore, we identify the top middle region as high complexity, confirmed by the consistent occurrence of intrinsic strong negative edge-pairs and edges of high centrality, as the radius threshold varies. 

\paragraph{Radius Null Model.} In our analysis, we focus on the effect of network architecture, and consider primarily null models where the network structure is randomized. Another perspective to explore is from the radius distribution. To start with, we consider the null model where the vessel radius is randomly shuffled but the network structure is maintained, which we refer to as ``Rad'' for Radius Null Model hereafter. We observe that randomizing the vessel radius also leads to behavior distinct to the real case, and the results from the Rad model usually lie between the real data and the SRG model; see Fig.~\ref{fig:vessel-null-radius}. For both the mean and the max values of the linking centrality, the results from the Rad model have the same overall trend as the SRG model, but the values are significantly lower, at a level that is very close to the real case. The results are consistent among the three zones. With respect to the unbalance score, the Rad model shows an overall trend similar to that of the SRG model, while their comparative behavior differs across zones. In Zone I, the real unbalance score is close to and slightly exceeds, the SRG baseline for small radius cuts; the Rad model shows the same qualitative behavior. In Zone II, where the real unbalance score rises significantly above the SRG model for radius cuts between $5$ and $10\, \mu m$, the Rad model again lies above the SRG model, with a stronger deviation than in Zone I. In Zone III, by contrast, the real unbalance score remains consistently below the SRG model, and the Rad model similarly falls below the SRG baseline at small radius cuts. Because most vessels have a radius below $10\, \mu m$, the behavior for radius cuts at or above this threshold is governed by the locations of the relatively few above-threshold edges, therefore, a behavior close to the SRG model is expected. Taken together, these observations suggest that the Rad model captures features of the real vascular data that are orthogonal to the other null models. While edge attributes, such as vessel radius here, provide additional and complementary sources of information, the present study focuses on spatial network organization; their systematic incorporation will be explored in future work.  

\begin{figure}[htbp]
    \centering
    \begin{tabular}{ccc}
        \includegraphics[width=.32\textwidth]{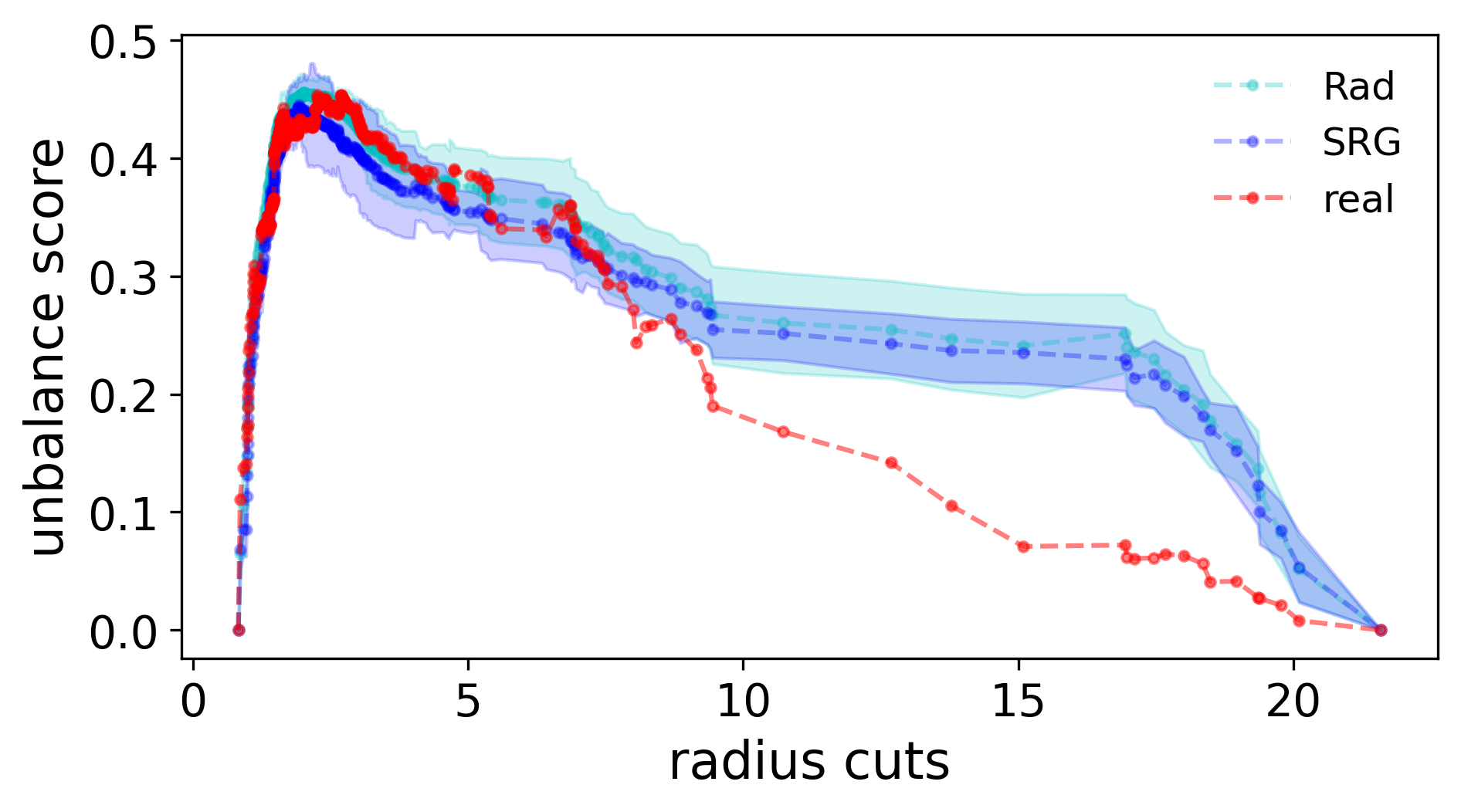} & \includegraphics[width=.32\textwidth]{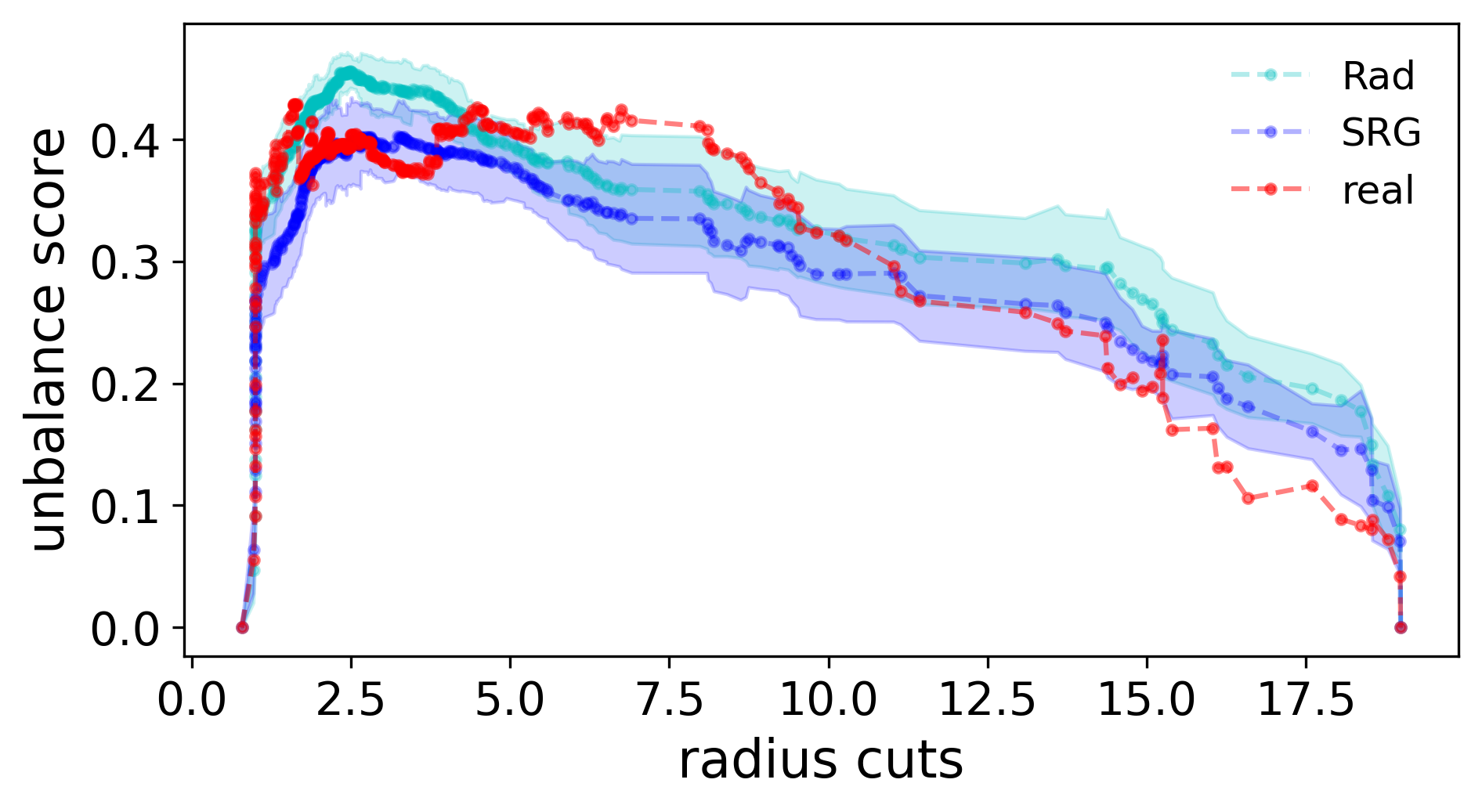} & \includegraphics[width=.32\textwidth]{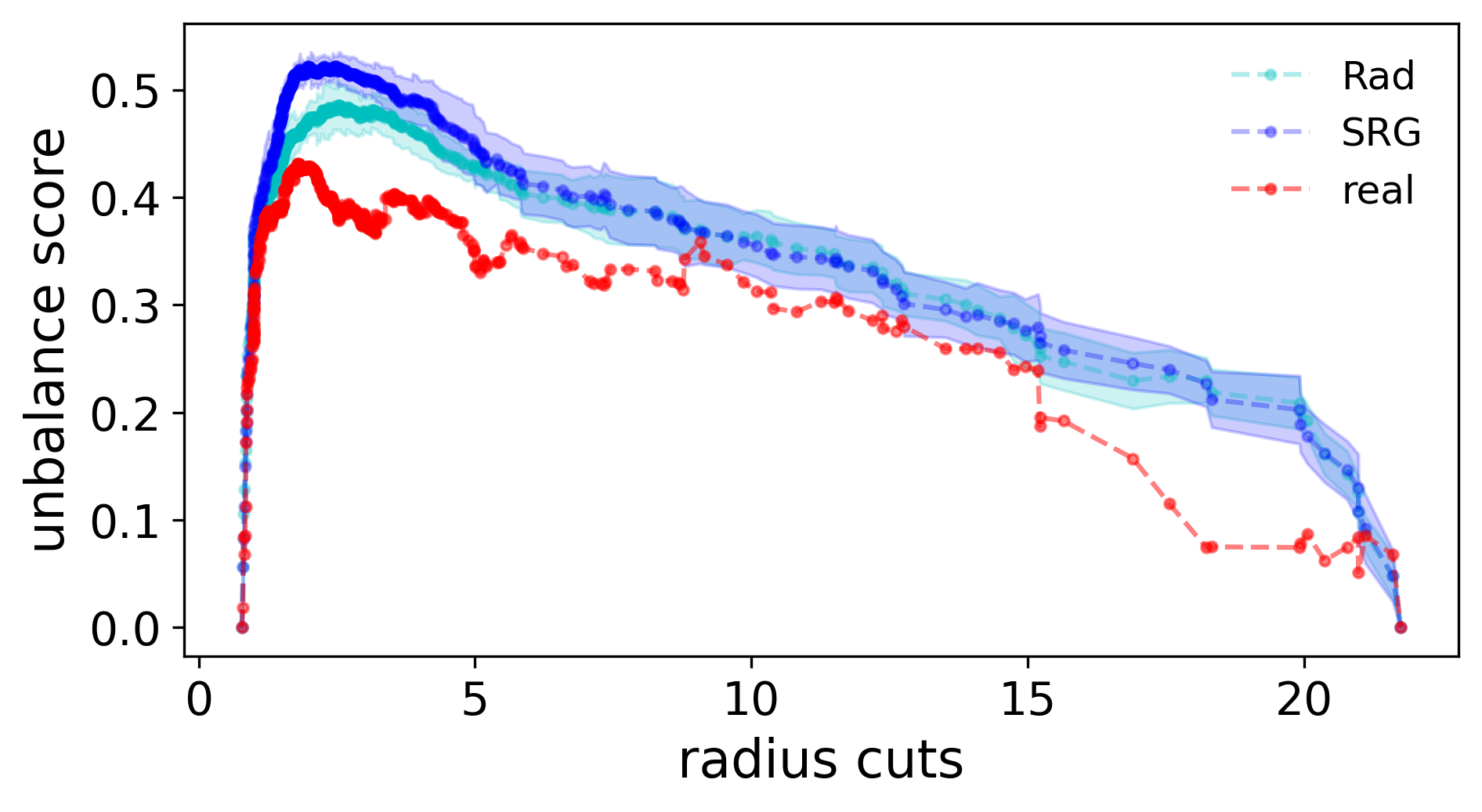} \\
        \includegraphics[width=.32\textwidth]{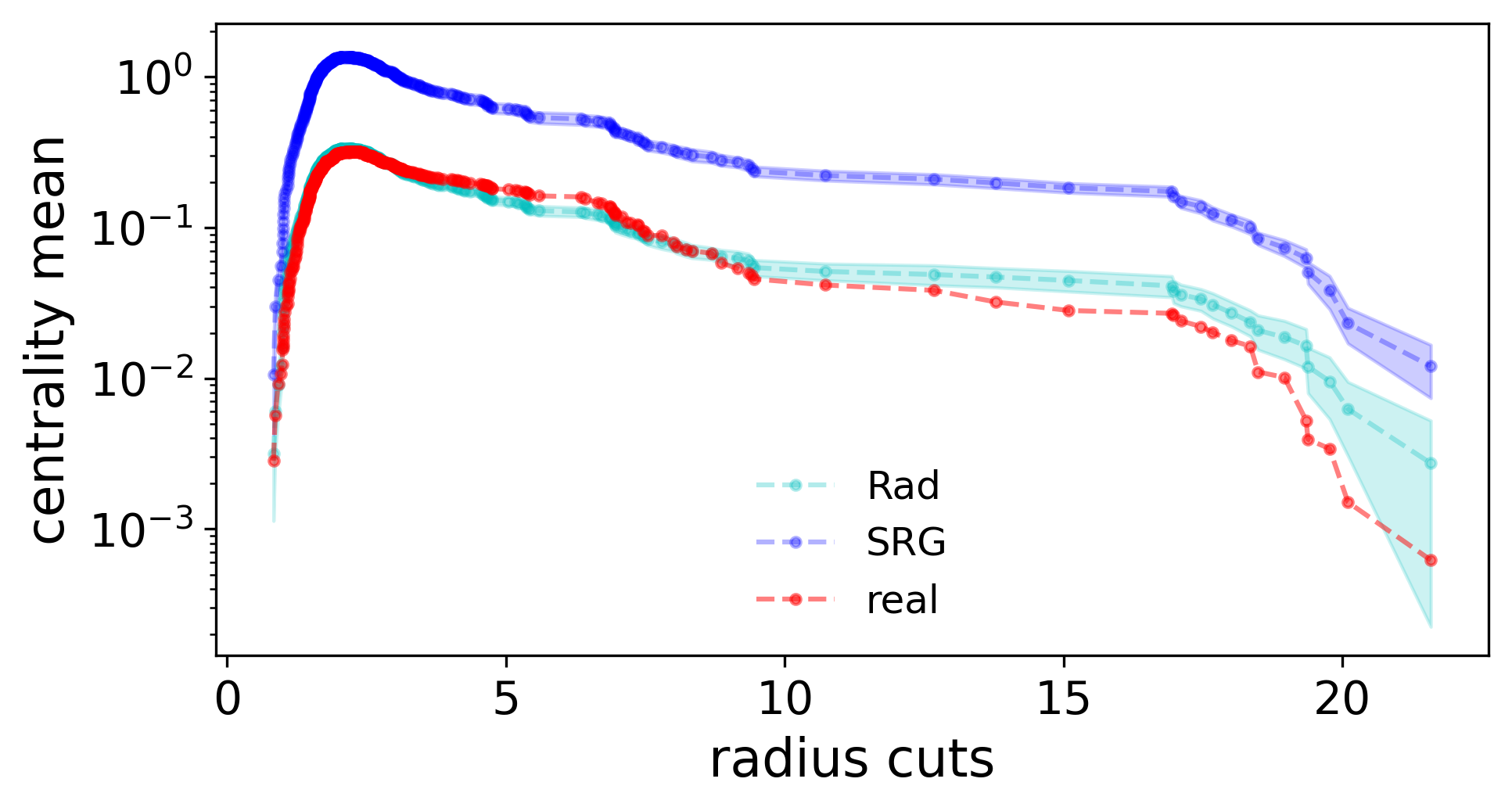} & \includegraphics[width=.32\textwidth]{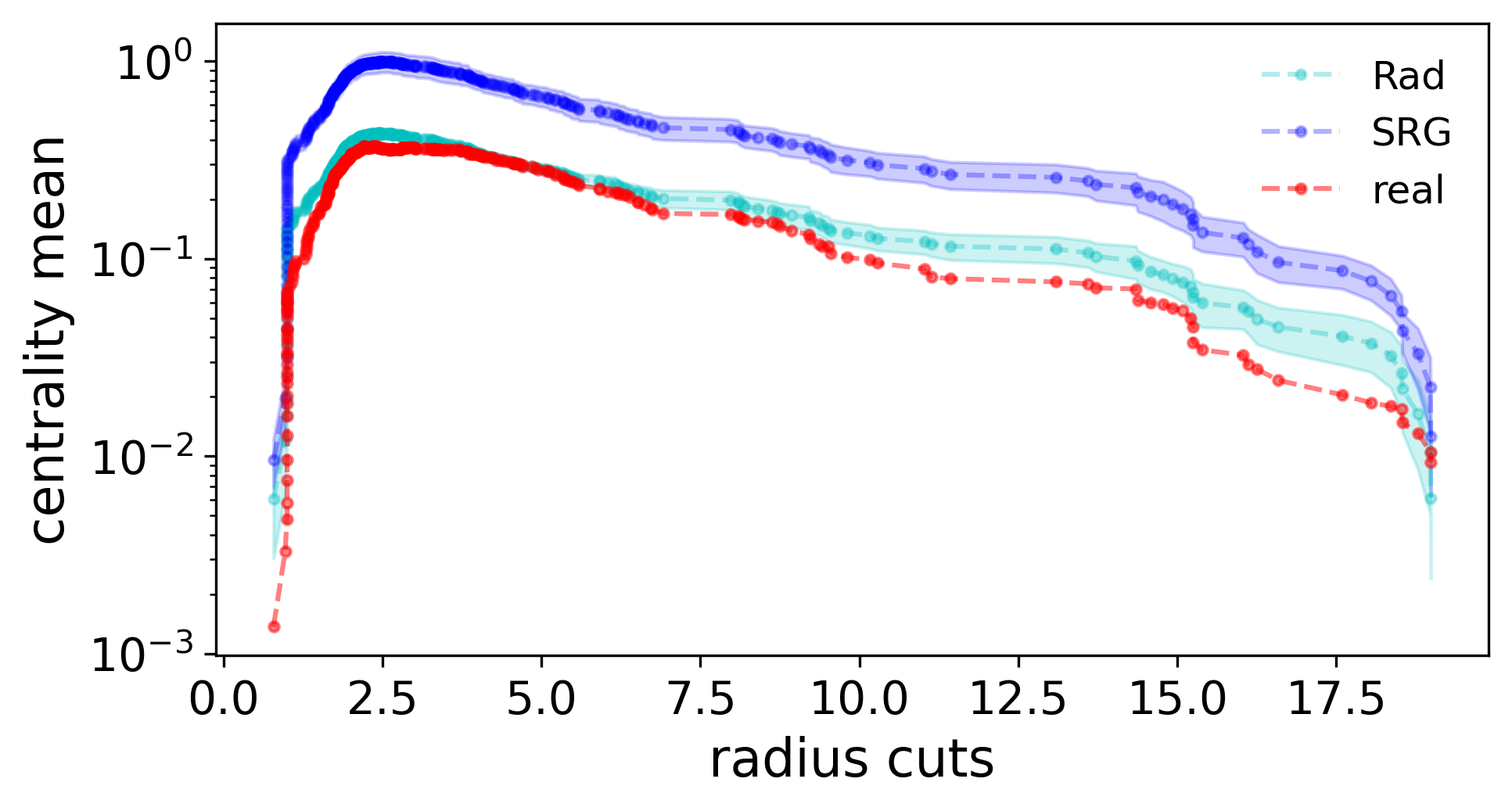} & \includegraphics[width=.32\textwidth]{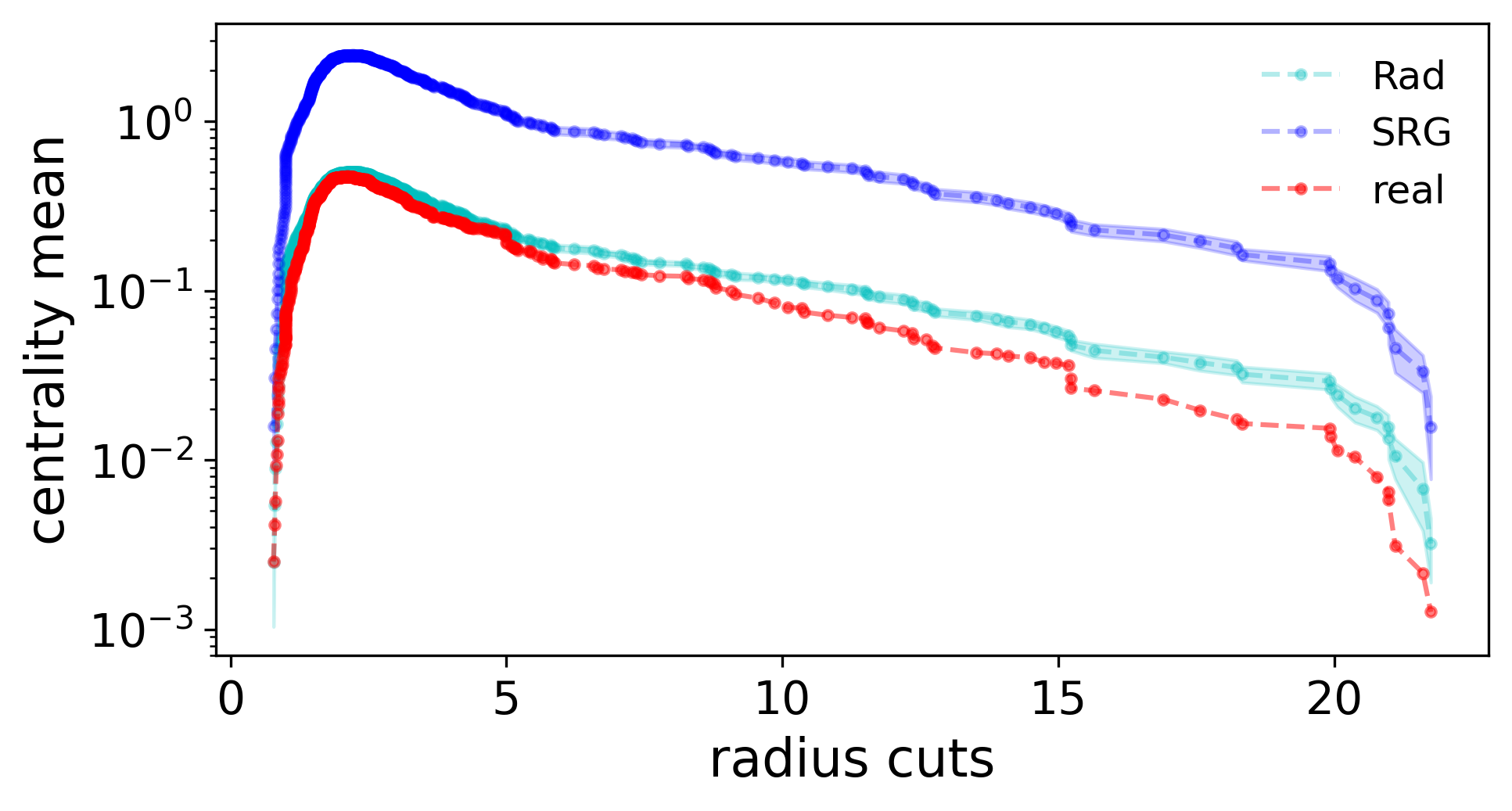}\\
        \includegraphics[width=.32\textwidth]{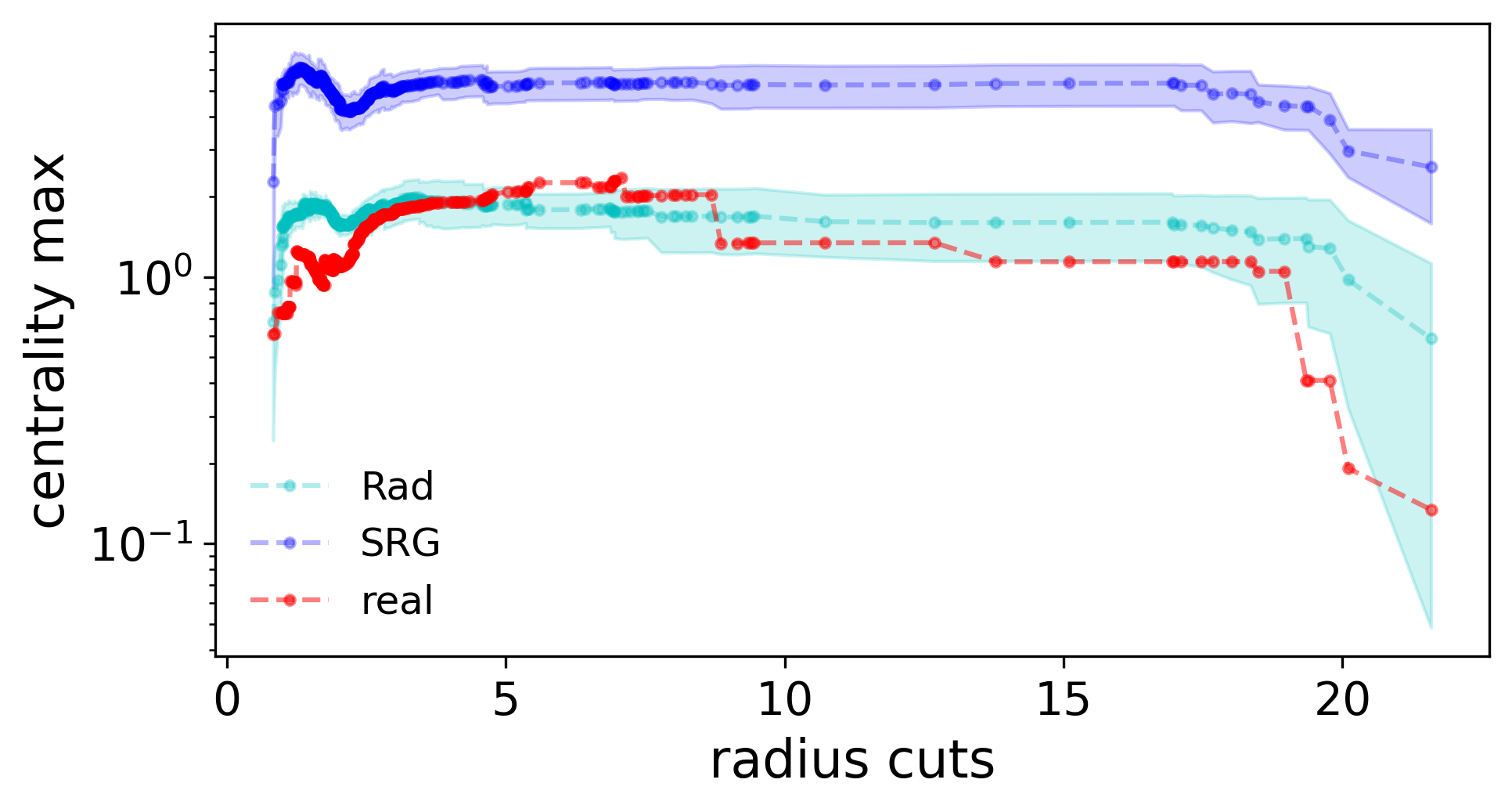} & \includegraphics[width=.32\textwidth]{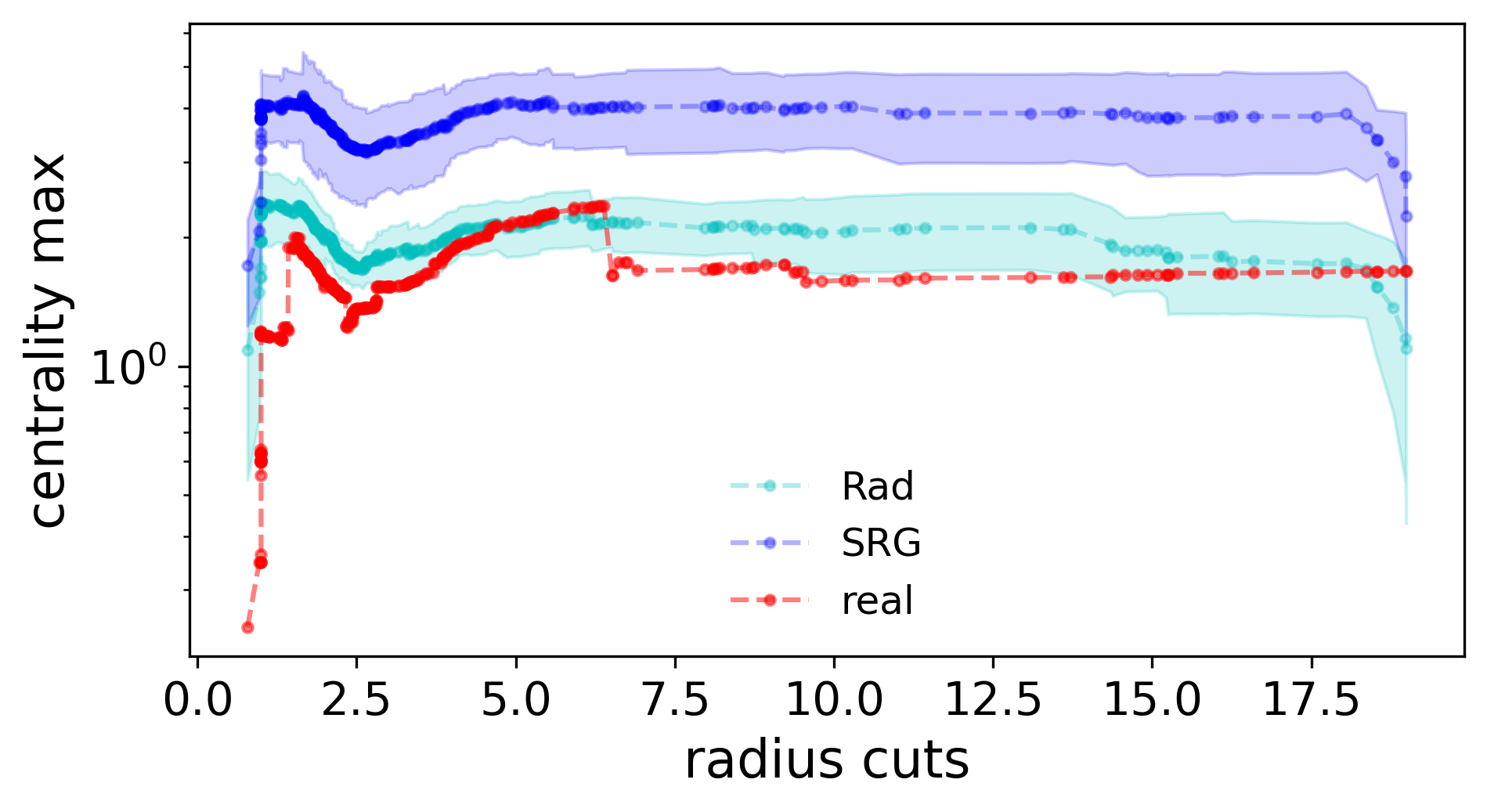} & \includegraphics[width=.32\textwidth]{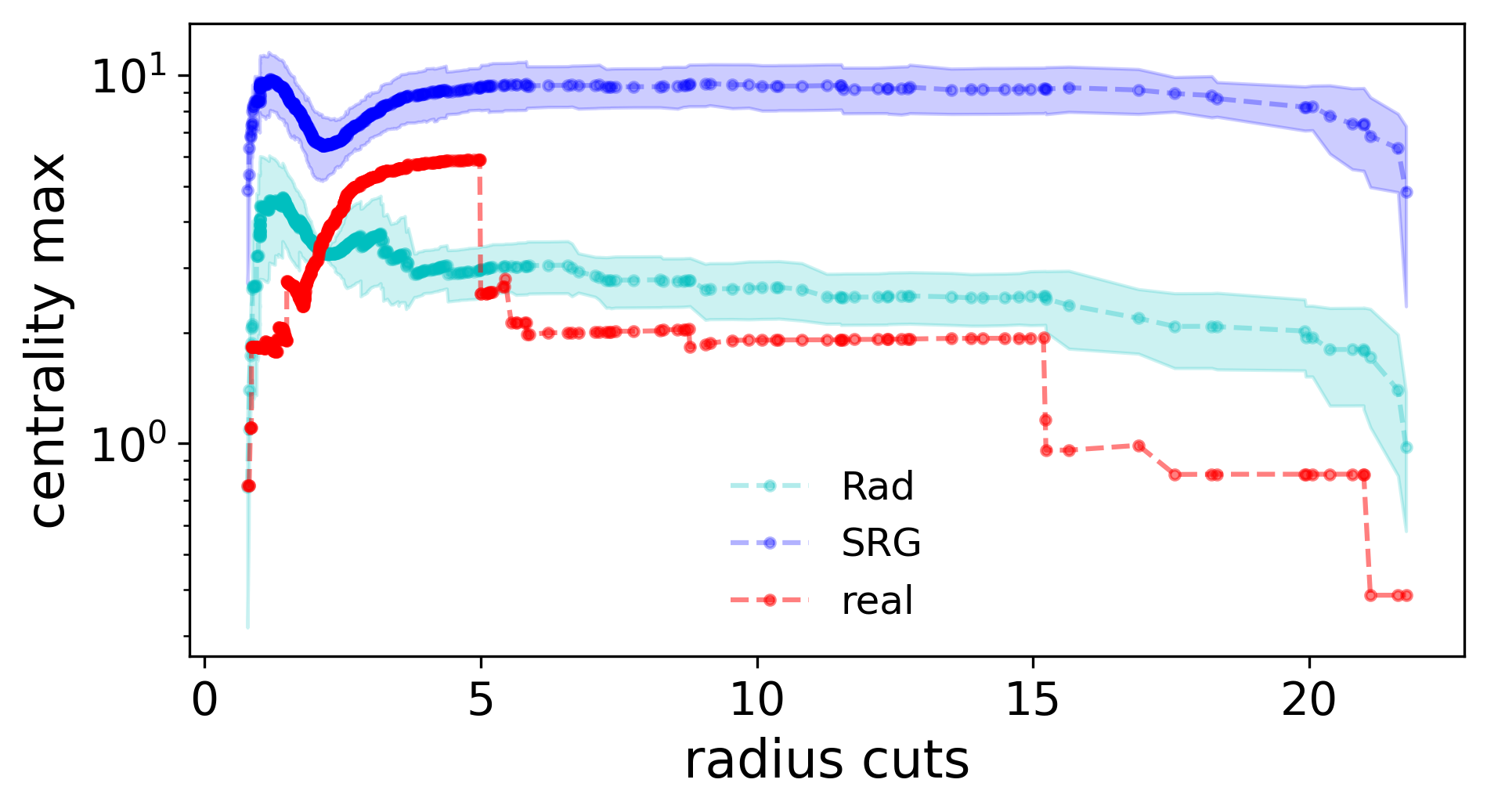}
    \end{tabular}
    \caption{Change of unbalance scores (top row), mean (middle row), and max (bottom row) linking centrality of all edges with respect to all thresholds on the radius for Zones I, II, and III from the real data, the Rad model, and the SRG model (mean: dashed line; std: shade of all samples).}
    \label{fig:vessel-null-radius}
\end{figure}

\end{document}